\documentclass{article}
\usepackage[utf8]{inputenc}
\usepackage[T1]{fontenc}
\usepackage{amsmath}
\usepackage{dsfont}
\usepackage{fancyhdr}
\usepackage[table,xcdraw]{xcolor}
\usepackage{graphicx}
\usepackage{multirow}
\usepackage{endnotes}
\usepackage{float}
\usepackage[nottoc]{tocbibind}
\usepackage{vmargin}
\usepackage{amssymb}
\usepackage{subfigure}
\usepackage{import}
\usepackage{wrapfig}
\usepackage{amsthm}
\usepackage{rotating}
\usepackage[title]{appendix}
\usepackage{listings}

\usepackage{amsmath,amssymb,amsthm,bbm,bm,nicefrac}

\usepackage{soul}

\usepackage{stmaryrd}
\usepackage{tikz}
\usepackage{mleftright}
\mleftright
\usetikzlibrary{matrix}
\usepackage{tikz}
\usetikzlibrary{shapes,arrows,chains, decorations.pathmorphing, decorations.pathreplacing}
\tikzset{quantum/.style={decorate, decoration=snake}}

\newcommand{\floor}[1]{\left\lfloor #1 \right\rfloor}

\newcommand{\ket}[1]{\left|#1\right\rangle}							
\newcommand{\bra}[1]{\left\langle#1\right|}

\newcommand{\ketbra}[2]{\left|#1\rangle\langle#2\right|}
\newcommand{\braket}[2]{\left\langle #1\lvert#2\right\rangle}
\newcommand{\abs}[1]{\lvert #1\rvert}

\newcommand{\expectedbraket}[1]{\langle #1\rangle}
\newcommand{\norm}[1]{\| #1\|}

\newcommand{\prover}{\ensuremath{\textnormal{\textsf{P}}}}

\newcommand{\QPVBB}{$\mathrm{QPV}_{\mathrm{BB84}}$}
\newcommand{\QPVBBeta}{$\mathrm{QPV}_{\mathrm{BB84}}^{\eta}$}
\newcommand{\QPVBBetaf}{$\mathrm{QPV}_{\mathrm{BB84}}^{\eta,f}$}
\newcommand{\QPVBBetam}{$\mathrm{QPV}_{m_{\theta\varphi}}^{\eta}$}
\newcommand{\QPVBBf}{$\mathrm{QPV}_{\mathrm{BB84}}^{f}$}

\newcommand{\QPVBBetafm}{$\mathrm{QPV}^{\eta,f}_{m_{\theta\varphi}}$}

\newcommand{\tr}[1]{\mathrm{Tr}\left[#1\right]}

\usepackage{slashed}
\usepackage{enumitem}

\usepackage{authblk}

\theoremstyle{plain}
\newtheorem{theorem}{Theorem}[section]
\newtheorem{remark}[theorem]{Remark}
\newtheorem{prop}[theorem]{Proposition}
\newtheorem{definition}[theorem]{Definition}
\newtheorem{ex}[theorem]{Example}
\newtheorem{lemma}[theorem]{Lemma}

\newtheorem{result}[theorem]{Result}

\usepackage[normalem]{ulem}

\usepackage[pdfencoding=auto, psdextra]{hyperref}
\hypersetup{
    colorlinks=true, 
    linktoc=all,     
    linkcolor=blue,  
}

\newcommand{\pr}[1]{\mathrm{Pr}\left[#1\right]}
\usetikzlibrary{backgrounds, positioning,calc}
\newcommand{\diagdots}[3][-25]{%
  \rotatebox{#1}{\makebox[0pt]{\makebox[#2]{\xleaders\hbox{$\cdot$\hskip#3}\hfill\kern0pt}}}%
}

\begin{document}
\title{Single-qubit loss-tolerant quantum position verification protocol secure against entangled attackers}

\author[1,2]{Lloren\c{c} Escol\`a-Farr\`as}
\author[1,2]{Florian Speelman}
\newcommand{\lle}[1]{{\color{blue}#1}}
\newcommand{\fs}[1]{{\textcolor{red}{[Florian: #1]}}}
\newcommand{\re}[1]{{\textcolor{green}{[Referee: #1]}}}

\affil[1]{QuSoft, CWI Amsterdam, Science Park 123, 1098 XG Amsterdam, The Netherlands}
\affil[2]{Multiscale Networked Systems (MNS),  Informatics Institute, University of Amsterdam, Science Park 904,  1098 XH Amsterdam, The Netherlands }

\renewcommand\Affilfont{\itshape\small}
\maketitle
\begin{abstract}
Protocols for quantum position verification (QPV) which combine classical and quantum information are insecure in the presence of loss. We study the exact loss-tolerance of the most popular protocol for QPV, which is based on BB84 states, and generalizations of this protocol.
By bounding the winning probabilities of a variant of the monogamy-of-entanglement game using semidefinite programming (SDP), we find tight bounds for the relation between loss and error for these extended non-local games.
    
These new bounds enable the usage of QPV protocols using more-realistic experimental parameters.
We show how these results transfer to the variant protocol which combines $n$ bits of classical information with a single qubit, thereby exhibiting a protocol secure against a linear amount of entanglement (in the classical information $n$) even in the presence of a moderate amount of photon loss. Moreover, this protocol stays secure even if the photon encoding the qubit travels arbitrarily slow in an optical fiber.
We also extend this analysis to the case of more than two bases, showing even stronger loss-tolerance for that case. 
    
Finally, since our semi-definite program bounds a monogamy-of-entanglement game, we describe how they can also be applied to improve the analysis of one-sided device-independent QKD protocols.
\end{abstract}

\section{Introduction}
Position-based cryptography (PBC), initially introduced by Chandran, Goyal,
Moriarty and Ostrovsky \cite{OriginalPositionBasedCryptChandran2009}, aims to allow a party to use its geographical location as a credential to implement various cryptographic protocols.
An important building block for PBC is so-called Position Verification (PV), where an untrusted prover \prover{} wants to convince a set of verifiers $V_0,\dots,V_k$ that she is at a certain position $r$. 
In \cite{OriginalPositionBasedCryptChandran2009} it was proven that no secure classical protocol for position verification can exist, since there exists a general attack based on copying classical information.
Adrian Kent first studied PV protocols that use quantum information in 2002, originally named \emph{quantum tagging} \cite{PatentKentANdOthers,OriginalQPV_Kent2011} and currently called Quantum Position Verification (QPV) in the literature.
Due to the no-cloning theorem \cite{Wootters1982NoCloning}, attacks based on simply copying information do not transfer to the quantum setting.
Nevertheless, unconditionally secure QPV was proven to be impossible, with \cite{Buhrman_2014} showing that any QPV could be attacked with double-exponential entanglement, which was later improved to exponential entanglement \cite{Beigi_2011}.

On the other side, it is possible to prove a linear lower bound for the amount of entanglement~\cite{Beigi_2011,TomamichelMonogamyGame2013}---an exponential gap. 
Still, the extreme inefficiency of the general attack left the possibility open that QPV could be shown secure if we consider attackers that are bounded in some (realistic) way.
Ideally that would mean a protocol which requires an exponential amount of entanglement to attack, but any provable gap between the hardness of executing the protocol versus breaking the protocol could be interesting.
This possibility created the opportunity for many interesting follow-up work, showing security under model assumptions such as finding clever polynomial attacks to proposed protocols~\cite{Lau_2011,Chakraborty_2015,speelman2016instantaneous,dolev2019constraining,dolev2022non,gonzales2019bounds,cree2022code}  and analyzing security in other models, such as the random oracle model~\cite{Unruh_2014_QPV_random_oracle,liu2021beating,gao2016quantum}.

Particularly interesting here is the recent work by Liu, Liu, and Qian~\cite{liu2021beating}, which surprisingly shows that classical communication suffices to construct a secure protocol---immediately solving the issue of transmission loss.
Despite this protocol being a great theoretical breakthrough, there are some downsides which make it unappealing to implement in the near future.
A minor theoretical downside is that this proposed protocol is secure only under computational assumptions, and additionally requires the random oracle model to be secure against any entanglement. The larger practical downside is that the honest prover in \cite{liu2021beating}'s protocol requires a large quantum computer to execute the protocol's steps---instead of manipulating and measuring single qubits as in most other  protocols we're considering.

One of the simplest and best-studied QPV protocols, which constitutes the basis of this work, is based on BB84 states, which we will denote \QPVBB. This protocol was initially introduced in \cite{OriginalQPV_Kent2011} and later on proved secure against unentangled adversaries in \cite{Buhrman_2014}, with parallel repetition shown in \cite{TomamichelMonogamyGame2013}\footnote{The analysis was tightened in \cite{https://doi.org/10.48550/arxiv.1504.07171}, however this was proven in a slightly-weaker security model where the attackers share a round of \emph{classical} communication instead of \emph{quantum} communication, so the analyses are not directly comparable.}. The protocol, explained in detail below, consists of one verifier ($V_0$) sending a BB84 state to the prover and the other verifier ($V_1$) sending a classical bit describing in which basis the prover has to measure, either the computational basis or the Hadamard basis. The prover then has to broadcast the measurement outcome to both verifiers, with all communication happening at the speed of light. This protocol can be attacked by attackers that share a single EPR pair.
However, by encoding the basis information in several longer classical strings $x,y\in\{0,1\}^n$ that have to be combined in some way (i.e.\ the basis is given by $f(x,y)$ for some function $f$) it is possible to construct a single-qubit protocol, \QPVBBf, for which the \emph{quantum} effort to break it grows with the amount of \emph{classical} information that an honest party would use.
This type of extension was already considered by Kent, Munro, and Spiller~\cite{OriginalQPV_Kent2011} for a slightly different protocol, and analyzed more in-depth by Buhrman, Fehr, Schaffner, and Speelman~\cite{Buhrman_2013}, with a linear lower bound shown by 
Bluhm, Christandl, and Speelman~\cite{bluhm2022single}\footnote{This linear lower bound is not necessarily tight---the best currently-known attack on \QPVBBf~ requires exponential entanglement in the number of classical bits for most functions $f$.}.(Also see \cite{junge2022geometry} for a related lower bound for a slightly different protocol which combines quantum and classical communication.)
The ability to enhance security by adding classical information makes the functional version of $\mathrm{QPV_{BB84} }$  an appealing candidate for future use.

However, applying QPV experimentally encounters implementation problems, of which two are large enough that  they force us to redesign our protocols. Whereas the transmission of classical information without loss at the speed of light is technologically feasible, e.g.\ via radio waves, the quantum counterpart faces obstacles.
Firstly, most QPV protocols require the quantum information to be transmitted at the speed of light in vacuum, but the speed of light in optical fibers, a medium which would be necessary for many practical scenarios, is significantly lower than in vacuum.  
Secondly, a sizable fraction of photons is lost in transmission in practice. In practice, if telecom wavelength ($\sim$1550nm) single-photon sources are used, and the photons are sent through optical fibers with a loss of $0.15$dB/km~\cite{Cao_2019}. For this loss problem, we can distinguish two recent approaches.
The first of which is to create protocols which are secure against any amount of loss, which we can call \emph{fully loss-tolerant protocols}. This type of protocol was first introduced by Lim, Xu, Siopsis, Chitambar, Evans, and Qi~\cite{lim2016loss}, based on ideas from device-independent QKD. New examples and further analyses of fully loss-tolerant protocols were given by Allerstorfer, Buhrman, Speelman, and Verduyn Lunel \cite{SWAP_protocol_Rene_et_all,allerstorfer2022role}.
These protocols could be excellent realistic candidates for an implementation of QPV, but in the longer term they have two shortcomings: they are not secure against much entanglement---\cite{allerstorfer2022role} for instance show that if security against unbounded loss is required, this is unavoidable---and they require fast transmission of \emph{quantum} information.

In this work we therefore advance another approach, which involves bounding the exact combination of loss rate and error rate that an attacker can achieve, thereby constructing what we may call \emph{partially loss-tolerant protocols}.
The first published example of this is given by Qi and Siopsis~\cite{OtherExtentionBB84Qi_Siopsis2015}, who propose extending the $\mathrm{QPV_{BB84} }$ protocol to more bases to give some loss-tolerance, a proposal which was independently made by Buhrman, Schaffner, Speelman, and Zbinden (available as \cite[Chapter 5]{FlorianThesis}).

One might wonder: because it's desirable to create a protocol which can tolerate a certain level of \emph{measurement error}, and we are only secure against limited loss anyway, why it is not possible to see photon loss as merely another source of error?
The answer is that it is very beneficial to treat these parameters separately because the numbers involved are entirely different. The basic $\mathrm{QPV_{BB84} }$ protocol has a perfect attack, which is possible by claiming a `loss' 50\% of the time. However, the best non-lossy attack has error 0.15.
This means that an experimental set-up which has a loss of 49\% could conceivably enable secure QPV, even though that clearly could never be good enough if the lost photons are counted as error. 
For different protocols, this difference in allowed parameters can be even larger.

In this paper, we study the security of the \QPVBB~ protocol in the lossy case in the No Pre-shared Entanglement (No-PE) model~\cite{Buhrman_2014}---that is, for attackers that do not pre-share entanglement prior to the execution of the protocol, but that are allowed to perform local operations and a single simultaneous round of quantum communication (LOBQC)~\cite{gonzales2019bounds}. Security in this model is motivated by the fact that generating and distributing high-quality entanglement over long distances remains technologically challenging, so the No-PE model captures currently practical limitations.
The best bounds for the non-lossy version of this protocol in the No-PE model come through reducing them to a monogamy-of-entanglement game~\cite{TomamichelMonogamyGame2013}---we extend this game to include a third `loss' response.
Via Semidefinite Programming (SDP), we can then tightly bound the game's winning probability, and therefore \QPVBB's  security under photon loss, showing that a naive mixture of the extremal strategies turns out to give the optimal combination of error and response rate for the attackers\footnote{Note that in our current work we solve the semidefinite programs for a complete range of experimental error probabilities, thereby obtaining an exhaustive characterization of the protocols we consider.
For an experimental implementation only a small number of SDPs would have to be solved, namely the one that corresponds to a small interval around the error and loss present in the experimental set-up, necessitating much less computation time than used in preparing our results.}.
This is a three-player game involving two players, who have complete freedom, and a referee who performs a fixed measurement for every input. {It is well-known how to analyze two-party scenarios using SDP relaxations, but extending these methods to three-party scenarios presents an obstacle. To obtain SDP bounds, we combine the NPA hierarchy \cite{NPA2008} with extra inequalities we derive from the relation between the referee's measurements---the way these inequalities enable us to tackle a three-player problem using SDPs might be of independent interest.}

Importantly, we are able to also show that these results can be adapted to show bounds for the lossy version of \QPVBBf.
We do this by defining a new relaxation of our earlier SDP, which holds for the unentangled case, and show that the numerical bounds for this new SDP can be used to reprove a key lemma in \cite{bluhm2022single} for the case of our protocol.
For instance, if the measurement error is very low, our results show that the single-qubit protocol \QPVBBf~ remains secure against attackers that share an amount of entanglement that is sublinear in the amount of classical bits, even when only, e.g.\ $51\%$ of the photons arrive at the honest party.
The security proof even goes through in the case that  the transmission of quantum information is slow, as long as  classical information can be transmitted fast. 

Moreover, we extend the  \QPVBB~ protocol where $V_0$ and $V_1$ agree on a qubit encoded in $m$ possible different bases over the whole Bloch sphere, a similar extension was proposed in \cite{OtherExtentionBB84Qi_Siopsis2015}, see below for the differences.
Via SDP characterization, although tightness is not guaranteed for $m\geq 3$ in our results, we show that the new protocol becomes more resilient against photon loss when $m$ increases. Again, we are able to show this also holds for the extension to the protocol which lets the basis be determined by more-complicated classical information \QPVBBetafm.
For instance, for $m=5$, if the measurement error is very small the protocol remains secure against attackers sharing less than roughly $\frac{n}{2} $ EPR pairs, even if almost $70\%$ of the photons are lost for the honest parties. 

The main result of this paper can be summarized as follows (see Theorems \ref{theorem q>=n/2-5 for BB84 eta} and \ref{theorem q>=n/2-5 for m-BB84 eta} for a formal version):
\newpage
\begin{theorem} \emph{(Informal)}. The \QPVBBetaf{} protocol is secure against attackers who pre-share a linear amount (in the size of the classical information) of qubits even if photon loss is almost 50\% ($\eta>1/2$). Moreover, the family of its extensions \QPVBBetafm{} achieves security for smaller constant fraction of transmission rate (for $\eta>\frac{1}{m}$).

\end{theorem}

Our results also provide improved upper bounds for the probability of winning some concrete monogamy-of-entanglement (MoE) games \cite{TomamichelMonogamyGame2013}, which encode attacks of  QPV protocols. Finally, we apply our technique proofs to prove security of one-sided device-independent quantum key distribution (DIQKD) BB84 \cite{BB84} for a single round $n=1$, see below, with security under photon loss. However, the interesting case is the asymptotic behavior for arbitrary $n$, which we leave as an open question. 

\paragraph{Comparison to earlier work.}
Our work is directly related to the questions asked by Qi and Siopsis~\cite{OtherExtentionBB84Qi_Siopsis2015}.
The most important difference between that work and ours, is that we are considering a more general security model.
The attack analysis in \cite{OtherExtentionBB84Qi_Siopsis2015} is based upon an assumption that attackers would immediately measure an incoming qubit using a projective measurement, and then communicate classically.
By reducing to (extensions of) monogamy-of-entanglement games, we instead capture \emph{any} quantum action of the attackers, including a round of quantum communication---it's shown in \cite{allerstorfer2022role} that this stronger security model can in fact make a difference in the security.
Critically, this also means that our method of analysis extends to the case where a classical function is combined with a single qubit, and our results make it possible to also prove entanglement bounds, instead of just bounds against unentangled protocols. 
Finally, our results can be reapplied in other settings that can be described by a monogamy-of-entanglement game.

There are two other recent papers by Allerstorfer, Buhrman, Speelman, and Verduyn Lunel~\cite{SWAP_protocol_Rene_et_all,allerstorfer2022role} that study the role of loss in quantum position verification. Those works provide many new results on the capabilities and limitations of fully loss-tolerant protocols. 

Our work builds upon the (unpublished) results of 
Buhrman, Schaffner, Speelman, and Zbinden \cite[Chapter 5]{FlorianThesis}. That work took the initial step to bound QPV protocols in the lossy case using SDPs. However, its analysis was incomplete and numerically significantly less tight, e.g.\ its bounds on $\mathrm{QPV_{BB84} }$ were a factor two worse than the current work. Additionally, we significantly extend its applicability by extending the technique to also show  entanglement bounds for the \QPVBBf~ protocol \cite{Buhrman_2013,bluhm2022single}.

\section{The \texorpdfstring{\QPVBBeta}{QPV-BB84-eta} protocol}\label{section 2 basis no entanglement}
Proposed QPV protocols rely on both relativistic constraints in a $d$-dimensional Minkowski space-time $M^{(d,1)}$ and the laws of quantum mechanics. In the literature, e.g.~\cite{OriginalQPV_Kent2011,Buhrman_2013},  the case $d=1$ is mostly found, i.e.\ verifying the position of \prover{} in a line, since it makes the analysis easier and the main ideas generalize to higher dimensions.
The most general setting for a 1-dimensional QPV protocol is the following: two verifiers $V_0$ and $V_1$, placed on the left and right of \prover{}, send quantum and classical messages to \prover{} at the speed of light, and she has to pass a challenge and reply correctly to them at the speed of light as well.
The verifiers are assumed to have perfect synchronized clocks and if any of them receives a wrong answer or the timing does not correspond with the time it would have taken for light to travel back from the honest prover, they abort the protocol.
Moreover, the time consumed by the prover to perform the challenge (usually a quantum measurement or a classical computation, see below for a detailed example) is assumed to be negligible. 

One of the protocols that has been studied the most is the $\mathrm{QPV_{BB84}}$ protocol, see below, originally introduced in \cite{OriginalQPV_Kent2011}. Here we introduce a variation of it where we consider that the quantum information sent through the quantum channel between $V_0$ and \prover{} can be lost. In practice, a representative example of this will be photon loss in optical fibers.  We also consider that an honest party is assumed to have error rate $p_{err}$, and thus will also respond with a wrong answer sometimes. This error can arise, for example, either from measurement errors or from noise in the quantum channel where the qubit is sent through. We define one round of the lossy-BB84 QPV protocol as follows:

\begin{definition}\label{def:lossyBB84}
Let $\eta$ denote the transmission rate of the qubits sent from $V_0$ to \prover{}. We define one round of the lossy-BB84 QPV protocol, denoted by \QPVBBeta, as follows:
\begin{enumerate}
    \item $V_0$ and $V_1$ secretly agree on  random bits $v,x\in\{0,1\}$. Then, $V_0$ prepares the qubit state $\ket\phi=H^x\ket{v}\in\{\ket0,\ket1,\ket+\,\ket-\}$, where $H$ denotes the Hadamard transformation.
    \item $V_0$ and $V_1$ send $\ket\phi$ and $x$ to \prover{}, respectively, with both signals propagating at the speed of light in vacuum. The two verifiers coordinate so that the qubit and the classical bit arrive to \prover{} simultaneously.
    \item Immediately, \prover{} measures the received qubit in the basis $x$ and broadcasts her outcome, either 0 or 1, to $V_0$ and $V_1$. If she did not receive the qubit, i.e.\  the photon was lost, she sends $\perp$. Therefore, the possible answers from \prover{} are $v_{\textsf{P}}\in\{0,1,\perp\}$. 
    \item If 
    \begin{enumerate} \item $V_0$ and $V_1$ receive their respective answers at the time corresponding with the claimed location, and they are equal, i.e.\  both receive the same $v_{\textsf{P}}$, then, if
        \begin{itemize}
            \item $v_{\textsf{P}}=v$, the verifiers record `\textsc{Correct}', denoted by `\textsc{c}', 
            \item $v_{\textsf{P}}=1-v$, the verifiers record `\textsc{Wrong}', denoted by `\textsc{w}', 
            \item $v_{\textsf{P}}=\perp$, the verifiers record `\textsc{No photon}',  denoted by `$\perp$', 
        \end{itemize}
        \item 
        otherwise, they record `\textsc{Abort}',  denoted by `$\lightning$', and abort the protocol rejecting the location.
    \end{enumerate}
\end{enumerate}
\end{definition}

See Fig.~\ref{fig:QPV_BB84} for a schematic of the original \QPVBBeta{} protocol. By setting $\eta=1$ and $p_{\mathrm{err}}=0$, one recovers the original \QPVBB{} protocol.

\begin{figure}[htbp]
    \centering
    \begin{tikzpicture}[node distance=3cm, auto]
    \node (A) {$V_0$};
    \node [above=0.5cm of A, yshift=-6pt] (line0) {$|$};
    \node [left=1cm of A] {};
    \node [right=of A] (B) {\prover{}};
    \node [above=0.5cm of B, yshift=-5pt] (P0) {$|$};
    \node [above=0.5cm of B, yshift=-15pt] (P00) {};
    \node [right=of B] (C) {$V_1$};
    \node [above=0.5cm of C, yshift=-7pt, xshift=-2pt] (line0) {$|$};
    \node [right=1cm of C] {};
    \node [below=of A] (D) {};
    \node [below=of B] (E) {};
    \node [below=of C] (F) {};
    \node [below=of D] (G) {};
    \node [below=of E] (H) {};
    \node [below=of F] (I) {};
    \node [left= 6cm of E] (J) {};
    \node [below= 3cm of J] (K) {};
    \node [above= 3cm of J] (L) {};
    
    \node [above=0.5cm of A] (posV00) {};
    \node [left=1cm of posV00] (posV0) {};
    \node [above=0.5cm of C] (posV11) {};
    \node [right=1cm of posV11] (posV1) {};
    \draw [->] (posV0) -- (posV1) node[midway,yshift=5pt] {position};

    \draw [->, transform canvas={xshift=0pt, yshift = 0 pt}, quantum] (A) -- (E) node[midway] (x) {} ;
    \draw [->] (C) -- (E);
    \draw [->] (E) -- (I) node[midway] (q) {$v_{\textsf{P}}\in\{0,1,\perp\}$};
    \draw [->] (E) -- (G);

    \draw [->] (L) -- (K) node[midway] {time};

    \node[left=1cm of x, transform canvas={xshift=+ 2pt, yshift = +2 pt}] {$H^x\ket{v}$};
    \node[right = 3.3cm of x, transform canvas={xshift=+ 2pt, yshift = +2 pt}] {$x \in \{0,1\}$};
    \node[left = 3.3cm of q, xshift=-5pt] {$v_{\textsf{P}}\in\{0,1,\perp\}$};
\end{tikzpicture}
\caption{ Schematic representation of steps~2\ and~3\ of the \QPVBBeta{} protocol, where straight lines represent classical information and undulated lines represent quantum information.  
}
\label{fig:QPV_BB84}
\end{figure}

In order to \emph{accept} or \emph{reject} the location, the verifiers run $r$ sequential rounds of \QPVBBeta{}. Let $r_{\textsc{C}}$, $r_{\textsc{W}}$, $r_{\perp}$ and $r_{\lightning}$ denote the number of times that the verifiers output `\textsc{Correct}', `\textsc{Wrong}', `\textsc{No photon}' and `\textsc{Abort}', respectively,  after the $r$ rounds. Given transmission rate $\eta$ and  error probability $p_{err}$, the verifiers expect 
\begin{equation}\label{eq succesful protocol}\begin{split}
   r_{\textsc{C}}&\approx r \eta (1-p_{err}), \hspace{1.5cm} r_{\textsc{W}}\approx r \eta p_{err},\\
    r_{\perp}&\approx r (1- \eta), \hspace{2.2cm} r_{\lightning}=0.
\end{split}
\end{equation}

The symbols $\approx$ are due to the fact that in an implementation, the above values would only be reached for $r\to\infty$.  To make a decision based on the observed outcomes, one must first define what it means for them to be sufficiently ``close” to the expected values in~\eqref{eq succesful protocol}. To this end, we introduce a binary test  (Definition~\ref{def:testBB84}) that determines whether the prover’s location should be \emph{accepted} or \emph{rejected} based on the received data.

Assume that after $r$ sequential repetitions, the verifiers did not receive any `\textsc{Abort}' answers, otherwise they \emph{reject} the location. Notice that the verifiers remain strict with $r_{\lightning}$, since an honest party would never send a different answer to $V_0$ that to $V_1$, and using current technology, her answers will arrive on time, and therefore a single `$\lightning$' answer will suffice to reject the location of the prover.   Denote by $\textbf{a}_i\in\{\textsc{c},\perp,\textsc{w}\}$ whether the answer they recorded in the round $i$ was `\textsc{Correct}', `\textsc{No photon}', or `\textsc{Wrong}'. Consider the payoff function 
    $T_i(\textbf{a}_i)=\sin^2\frac{\pi}{8}\textbf{1}_{\textsc{c}}(\textbf{a}_i)-\sin^2\frac{\pi}{8}\textbf{1}_{\perp}(\textbf{a}_i)-\cos^2\frac{\pi}{8}\textbf{1}_{\textsc{w}}(\textbf{a}_i)$,
for every round $i$ of the protocol.  Let 
\begin{equation}\label{eq:score_BB84}
    \Gamma_r=\sum_{i=1}^r T_i(\textbf{a}_i),
\end{equation}
be the total \emph{score} after $r$ rounds. Fix a parameter $\varepsilon_{\textrm{h}} > 0$, which determines the confidence level of the test that we introduce. That is, it will succeed with probability at least $1 - \varepsilon_{\textrm{h}}$ when interacting with an honest prover. 

\begin{definition}\label{def:testBB84} Let $\varepsilon_{\textrm{h}}>0$. For the \QPVBBeta~protocol executed sequentially $r$ times, we define the acceptance test $\mathsf{T}^{r\mathrm{BB84}}_{\varepsilon_{\textrm{h}}}$, also referred to as the decision criterion, as follows:     the verifiers \emph{accept} the prover’s location if
\begin{equation}\label{eq:test-BB84-qpv}
    \Gamma_r \geq r (\alpha-\delta),
\end{equation}
where $\delta=\sqrt{\frac{\cos^4\frac{\pi}{8}\ln(1/\varepsilon_{\textrm{h}})}{r}}$. 
Otherwise, they \emph{reject}. 
\end{definition}

We will next see that the test $\mathsf{T}^{r\mathrm{BB84}}_{\varepsilon_{\textrm{h}}}$ is \emph{complete}: an honest implementation of the protocol will be accepted, except with negligible probability. For an honest prover ($hp$), we have that, for every round  $i$, 
\begin{equation}
    \mathbb E[T_i^{hp}] = \sin^2\frac{\pi}{8}\eta(1-p_{err})-\sin^2\frac{\pi}{8}(1-\eta)-\cos^2\frac{\pi}{8}\eta p_{err}=:\alpha(\eta,p_{err}),
\end{equation}
and therefore, $\mathbb E [\Gamma_r^{hp}]=r\alpha$.  For simplicity, we will assume the dependence on $\eta$ and $p_{err}$ in $\alpha$ implicit. Then, by Hoeffding's inequality~\cite{hoeffding1963}, see Lemma~\ref{lemma:Hoeffding}, an honest prover will be accepted except with small probability at most $\varepsilon_{\textrm{h}}$. 

\begin{lemma} 
\label{lemma:Hoeffding}Let $T_1,\ldots,T_r$ be independent bounded random variables with $T_i\in[x_a,x_b]$, for all $i\in\{1,\ldots,r\}$, with $-\infty< x_a\leq x_b<\infty$. Then, for all $\delta\geq 0$, the following holds:
\begin{equation}
    \pr{\frac{1}{r}\sum_{i=1}^r(T_i-\mathbb E[T_i])\geq \delta}\leq e^{-\frac{2 r \delta^2 }{(x_b-x_a)^2}}.
\end{equation}
\end{lemma}

In addition, we will see that for any attackers in the No-PE model---which imposes only the restriction that no entanglement is shared prior to the execution of the protocol in each round---the test $\mathsf{T}^{r\mathrm{BB84}}_{\varepsilon_{\textrm{h}}}$ is \emph{sound}: any attackers in this model will be rejected with high probability---in particular, we will show exponentially high probability (in $r$).

\begin{remark} \label{remark constructing test}   This test is engineered from the analysis of the correlations attainable by attackers in Section~\ref{section 2 basis no entanglement}, ensuring that they are rejected except with negligible probability; see the proof of Theorem~\ref{thm sequential repetition no-entanglement} for details of its construction.
\end{remark}

We say that a protocol is \emph{secure} in a given model if it admits a test that is both complete and sound (for any adversaries acting according to the model). Next, we will show that $\mathsf{T}^{r\mathrm{BB84}}_{\varepsilon_{\textrm{h}}}$ fulfills both conditions for a certain range of transmission rate $\eta$ and qubit error $p_{err}$ and thus showing that \QPVBBeta{} is secure.

In order to prove security for \QPVBBeta, we will first show that, in a round, for a range of values of $\eta$ and $p_{\mathrm{err}}$, any adversaries in the No Pre-shared Entanglement (No-PE) model are unable to reproduce the outcome probabilities specified in~\eqref{eq succesful protocol}. This provides intuitive evidence that attackers cannot mimic the behavior of an honest prover, and that their actions can be statistically distinguished. We will later show that the test $\mathsf{T}^{r\mathrm{BB84}}_{\varepsilon_{\textrm{h}}}$ achieves this.

For the security analysis, we will consider the \emph{purified} version of \QPVBB, which is equivalent to it. In this version, instead of $V_0$ sending BB84 states, $V_0$ prepares an EPR pair $\ket{\Phi^+}_{VP}$ and sends the register $P$ to the prover and keeps register $V$. At a later moment, $V_0$ performs the measurement  $\{H^{x}\ketbra{v}{v}_VH^{x}\}_{v\in\{0,1\}}$ in his register $V$. In this way, the verifiers delay the choice of basis in which the qubit is encoded, which, in contrast to the above \emph{prepare-and-measure} version, will make any attack independent of the state sent by $V_0$.

Whereas it is well-known that adversaries sharing an unbounded amount of entanglement can always successfully break any QPV protocol \cite{Buhrman_2014}, the proof of the security of the \QPVBB{} protocol under attackers that do not pre-share entanglement \cite{Buhrman_2014} opened a branch of study, motivated by adversarial models that restrict attackers in a more realistic way.  The most generic attack to \QPVBBeta{} (purified version) in the No-PE model consists of:
\begin{enumerate}
\item 
Alices and Bob prepare an arbitrary quantum state $\sigma_{A_0}$, and Alice holds it. Since the attackers do not pre-share entanglement, any quantum operation that Bob could later perform as a function of $x$ can be included in Alice's operation (see e.g.~\cite{Buhrman_2014,TomamichelMonogamyGame2013}).

\item 
Alice intercepts the qubit register $P$, and applies an arbitrary quantum channel $\mathcal E_{PA_0\rightarrow AB}$ to it, and to $\sigma_{A_0}$. The subscript $PA_0\rightarrow AB$ indicates that the map has input and output registers $PA_0$ and $AB$, respectively. Let $\rho_{V\!AB}$ be the resulting state, that is,
\begin{equation}
    {\mathbb I_V\otimes \mathcal{E}_{PA_0\rightarrow AB}(\ketbra{\Phi^+}{\Phi^+}_{VP}\otimes\sigma_{A_0})=\rho_{V\!AB}}. 
\end{equation}
Alice possesses registers $A$ and $B$, and $V_0$ holds~$V$.   On the other side, Bob intercepts $x$, and copies it.
\item 
Alice keeps register $A$ of $\rho$ and sends register $B$ to Bob. Bob keeps a copy of $x$ and sends a copy to Alice. 
\item 
After one round of simultaneous communication, each party performs a POVM $\{A^x_a\}_{a\in\{0,1,\perp\}}$ and $\{B^x_b\}_{b\in\{0,1,\perp\}}$, on registers $A$ and $B$ of the state $\rho$, and they send  answers $a,b$, respectively, to their corresponding closest verifier. 
\end{enumerate}
See Fig.~\ref{fig:attack_lossyBB84} for a schematic representation of a generic attack to \QPVBBeta{} in the No-PE model. The tuple $\mathbf{S}_\eta:=\{\rho,A^x_a,B^x_b\}_{x,a,b}$ will be called a \emph{strategy} for~\QPVBBeta.

\begin{figure}[ht]
    \centering
    \begin{tikzpicture}[node distance=3cm, auto]
    \node (V0) {$V_0$};   
    \node [above=0.5cm of V0, xshift=-2pt, yshift=-7pt] (P0) {$|$};
    \node[right=4.5cm of P0, yshift=-1pt, xshift=3pt]{$|$}; 
    
    \node [left=2cm of V0](t0){};
    \node [below=7.5cm of t0](t1){};
    \draw [->] (t0) -- (t1) node[midway] {time};

    \node [right=1cm of V0] (A) {Alice};
    \node [right=1.5cm of A, xshift=19pt] (Pos) {};
    \node [above=0.5cm of A, yshift=-6pt] (P01) {$|$};
    
    \node [right= 5cm  of A] (B) {Bob};
    \node [above=0.5cm of B, yshift=-6pt] (P02) {$|$};
    
    \node [right= 2.5cm  of A] (P) {};
    \node [right=1cm of B] (V1) {$V_1$};
    \node [right=1.5cm of P02, xshift=-3pt,yshift=0pt] (P03) {$|$};
    
    \node [below =0.6cm of V0](below_V0){};
    \node [below =0.5cm of V1](below_V1){};

    \node [below =6cm of V0](V_ABans){};
    \node[below =0.5cm of V0](V00){};
    \node[below =0.1cm of V00](V001){};

    

    \node[below =1.2cm of V00, xshift=-10pt](V0001){};

    \node[right=1.2cm of V0001](registerA){$A$};
    \node[right=3.8cm of V0001](registerB){$B$};
    \node[right=5.6cm of V0001](copyz){$x$};
    \node[right=8.4cm of V0001](copyz2){$x$};

    \node [below =7.1cm of V0, xshift=24pt](V0_ans){};
    \node [below =7.1cm of V1, xshift=-24pt](V1_ans){};
    
    \node [below =0.2cm of P](middle){};
    \node [above =0.25cm of middle](middle0){};

    \node [below =0.1cm of middle](middle0){};
    \node [below =0.5cm of A](A_intercepts_Q){};
    \node [below =0.5cm of B](B_intercepts_Q){};
    \node [below =0.5cm of A](A_intercepts_x){$\mathbb I_V\otimes \mathcal{E}(\ketbra{\Phi^+}{\Phi^+}_{VP}\otimes\sigma_{A_0})=\rho_{V\!AB}$};
    \node [right =0.1cm of A_intercepts_x, xshift=60pt](operation){};
    \node [below =0.5cm of B](B_intercepts_y){$x$};

    \node [above=0.5cm of A] (posV00) {};
    \node [left=2.5cm of posV00] (posV0) {};
    \node [above=0.5cm of B] (posV11) {};
    \node [right=2.5cm of posV11] (posV1) {};
    \draw [->] (posV0) -- (posV1) node[midway, yshift=5pt] {position};

    \node [below =5cm of A](A_commits){};
    \node [below =5cm of B](B_commits){};
    \node [below =6cm of A](A_answers){$\{A^{x}_a\}_a$};
    \node [below =0.1cm of A_answers, ](A_POVM){};
    \node [below =6cm of B](B_answers){$\{B^{x}_b\}_b$};
    \node [below =0.1cm of B_answers](B_POVM){};

    \node [left=0.3cm of A](leftA){};
    \node[below=5.6cm of leftA](cA){};
    \node [left=0.2cm of A](leftA2){};
    \node[below=6.6cm of leftA2](aA){};

    \node [right=0.3cm of B](leftB){};
    \node[below=5.6cm of leftB](cB){};
    \node [right=0.2cm of B](leftB2){};
    \node[below=6.6cm of leftB2](bB){};
    
    \node [below =6.2cm of V0](V0bis){};
    \node [below =6.2cm of V1](V1bis){};
    \node [below =7.2cm of V0](V0bis1){};
    \node [below =7.2cm of V1](V1bis1){};

    \draw [->, transform canvas={xshift=0pt, yshift = 0 pt}, quantum] (A_intercepts_x) -- (B_answers) node[midway] (x) {};

    \draw [->, transform canvas={xshift=0pt, yshift = 0 pt}] (B_intercepts_y) -- (A_answers) node[midway] (x) {};

    \draw [->, quantum] (A_intercepts_x) -- (A_answers) {};

    \draw [->] (B_intercepts_y) -- (B_answers) {};
    \draw [->] (A_answers) -- (V0_ans) {};
    \draw [->] (B_answers) -- (V1_ans) {};


    \node [right =2.5cm of A_answers](middle2){};
    \node [below =5.2cm of middle, yshift=-16pt](middle3){$\rho_{V\!AB}$};

    \begin{scope}[on background layer]
    \fill[gray!30, opacity=0.45, rounded corners=2pt] 
        ($(A_answers) + (-2.5,0.25)$) rectangle ($(B_answers) + (0.7,-0.25)$);
\end{scope}

\end{tikzpicture}
\caption{Schematic representation of a generic attack to \QPVBBeta{} in the No-PE model. Straight lines represent classical information, and undulated lines represent quantum information. The gray-shaded region represents the tripartite quantum state $\rho_{VAB}$.}
\label{fig:attack_lossyBB84}
\end{figure}

To illustrate how photon loss can be exploited by adversaries, consider the following example. Suppose the verifiers expect that an honest prover would have a transmission rate of at most~$\eta \leq \frac{1}{2}$, meaning that half or more of the photons are lost in transit. In contrast, since the adversaries may position themselves closer to the verifiers and intercept most of the transmitted photons, we give them the extra power of intercepting all of them. The attack proceeds as follows: Alice selects a random bit $\Tilde{x} \in \{0,1\}$, measures the received qubit in the $\Tilde{x}$ basis, and forwards the outcome and $\Tilde{x}$ to Bob. After a single round of simultaneous communication, both attackers know whether the guess was correct. If $\Tilde{x} = x$, they return the measurement outcome; otherwise, they claim photon loss, i.e.\  output~$\perp$. Since the basis guess is correct with probability~$1/2$, this strategy succeeds with probability~$1/2$---matching the expected success rate of an honest prover under $\eta = \frac{1}{2}$. We refer to this strategy as the \emph{guessing strategy}, denoted~$\mathbf{S}_{\mathrm{guess}}$.

The probabilities that the verifiers, after the attackers' actions (for the random variable $\textsc{V}_{\textsc{AB}}$) record `\textsc{Correct}', `\textsc{Wrong}, `\textsc{No photon}', and `\textsc{Abort}' (different answers)  are thus, respectively, given by 
\begin{align}
q_{\textsc{c}}:=\pr{\textsc{V}_\textsc{AB}=\textsc{c}}&=\frac{1}{2}\sum_{a,x\in \{0,1\}}\tr{\rho V^x_a\otimes A^x_a\otimes B^x_a},\\
      q_{\textsc{w}}:=\pr{\textsc{V}_\textsc{AB}=\textsc{w}}&=\frac{1}{2}\sum_{a,x\in \{0,1\}}\tr{\rho V^x_a\otimes A^x_{1-a}\otimes B^x_{1-a}},\\
     q_{\perp}:=\pr{\textsc{V}_\textsc{AB}=\perp}&=\frac{1}{2}\sum_{x\in\{0,1\}}\tr{\rho \mathbb I_{V}\otimes A^x_{\perp}\otimes B^x_{\perp}},\\
      q_{\lightning}:=\pr{\textsc{V}_\textsc{AB}=\lightning}&=\frac{1}{2}\sum_{a\neq b\in \{0,1,\perp\},x\in\{0,1\}}\tr{\rho \mathbb{I}\otimes A^x_a\otimes B^x_b}\label{eq abort}.
\end{align}
In every round $i\in\{1,\ldots,r\}$ of the protocol (executed sequentially $r$ times), the attackers will pick a strategy $\mathbf S^i_\eta$ that can depend on the previous rounds. A strategy $\mathbf S^i_\eta$ will induce a probability vector $\textbf{q}^i=(q_{\textsc{c}}^i,q_{\perp}^i,q_{\textsc{w}}^i,q_{\lightning}^i)$.  In an honest implementation, 
\begin{equation}\label{eq succesful 1 round attack}\begin{split}
   &\pr{\textsc{V}_\textsc{AB}=\textsc{c}} =\eta (1-p_{err})=:p_{\textsc{c}}, \hspace{1cm} \pr{\textsc{V}_\textsc{AB}=\textsc{w}}=\eta p_{err}=:p_{\textsc{w}},\\
    &\pr{\textsc{V}_\textsc{AB}=\perp} = 1- \eta=:p_{\perp}, \hspace{2cm}  \pr{\textsc{V}_\textsc{AB}=\lightning}=0=:p_{\lightning}. 
\end{split}
\end{equation}
This defines a probability vector $\textbf{p}_{hp}=(p_{\textsc{c}},p_{\perp},p_{\textsc{w}},p_{\lightning})$ that the honest prover would ideally reproduce. An attack is \emph{successful} if the verifiers cannot distinguish if their data came from the distribution $\textbf{p}_{hp}...\textbf{p}_{hp}$ ($r$ times) or from $\textbf{q}^1...\textbf{q}^r$.  As mentioned above, in Section~\ref{section:sequential rep BB84-eta} we provide a test to distinguish between these two cases based on the received data.

\subsection{Exact loss-tolerance of \texorpdfstring{\QPVBBeta}{QPV-eta}}

A single round attack to \QPVBBeta~ for $\eta=1$ can be identified with a so-called monogamy-of-entanglement (MoE) game, introduced by Tomamichel, Fehr, Kaniewski and Wehner in \cite{TomamichelMonogamyGame2013}, formalized below in Definition~\ref{def MoE game} and generalized by Johnston, Mittal, Russo, and Watrous~\cite{Extended_non-local_games_andMoE_games}. The authors of \cite{TomamichelMonogamyGame2013} showed that the optimal probability that the attackers are correct in one round of the \QPVBB~protocol is $\cos^2(\frac{\pi}{8})$. Moreover, they show strong parallel repetition, i.e.\  if \QPVBB~is executed $m$ times in parallel, the optimal probability that the attackers are answer correct is $(\cos^2(\frac{\pi}{8}))^m$. Here we consider an \emph{extension} of a MoE game, which we will call \emph{lossy} MoE game, that will capture a round attack of \QPVBBeta---specifically step~4---and extensions of it, see Section~\ref{section m basis no entanglement}.

\begin{definition}\label{def MoE game} Let $\mathcal{X}$ and  $\mathcal{V}$ be finite non-empty alphabets. Let $V^{x}_{v}$ be POVMs of the same finite dimension for all ${(x,v)\in\mathcal{X}\times\mathcal{V}}$, and let $\mathcal M:=\{V^x_v\}_{x,v}$. A \emph{lossy} monogamy-of-entanglement game with parameter ${\eta\in[0,1]}$,   played by a referee, with associated Hilbert space $\mathcal H_R$, and two collaborative parties Alice and Bob, denoted by 
\begin{equation}
    \textsf{G}_\eta:=(\eta,\mathcal M),
\end{equation}
is described as follows: 
\begin{enumerate}
    \item Alice and Bob, with associated Hilbert spaces $\mathcal{H}_A$ and $\mathcal{H}_B$, respectively, prepare a quantum state $\rho_{RAB}\in\mathcal S(\mathcal{H}_R\otimes\mathcal{H}_A\otimes\mathcal{H}_B)$. 
    \item They send register $R$ to the referee, holding on registers $A$ and $B$, respectively. The two parties are no longer allowed to communicate. 
    \item The referee chooses $x\in\mathcal{X}$ uniformly at random and measures register $R$ using $\{V_v^x\}_v$ to obtain the measurement outcome $v$. Then, he announces $x$ to Alice and Bob.
    \item The collaborative parties make a guess for $v$ and they win the game if and only if both either guess $v$ correctly or both answer $\perp$ (with probability $1-\eta$).     In order to obtain the answers, Alice and Bob perform POVMs $\{A_a^x\}_{a\in\mathcal{V}\cup\{\perp\}}$ and $\{B_a^x\}_{a\in\mathcal{V}\cup\{\perp\}}$ on their local registers, respectively. The tuple $\mathbf{S}_\eta:=\{\rho_{RAB},A_v^x,B_v^x\}_{v\in\mathcal{V}\cup\{\perp\},x\in\mathcal{X}}$ will be called a \emph{strategy} for $\mathcal{G}_\eta$. 
\end{enumerate}
The lossy constraint, i.e.\  answering $\perp$ with probability $1-\eta$, is given by  
\begin{equation}\label{eq:lossy_constraint}
    \frac{1}{\abs{\mathcal{X}}}\sum_{x\in\mathcal{X}}\tr{\rho_{RAB}\mathbb{I}_V\otimes A^x_{\perp}\otimes B^x_{\perp}}=1-\eta.
\end{equation}
\end{definition}

\noindent See Fig.~\ref{fig:lossyMoEgame} for a schematic representation of a lossy MoE game. A monogamy-of-entanglement game $\mathsf{G}$ as introduced in \cite{TomamichelMonogamyGame2013} is recovered by setting $\eta=1$, i.e.\  $\mathsf{G}=\mathsf{G}_{\eta=1}$. The winning probability of a lossy MoE game $\mathsf{G}_\eta$, given a strategy ${\mathbf{S}_\eta=\{\rho_{RAB},A_v^x,B_v^x\}_{v\in\mathcal{V}\cup\{\perp\},x\in\mathcal{X}}}$, is given by 
\begin{equation}\label{eq:winnin_lossy_MOE}
    \omega(\mathsf{G}_\eta,\mathbf{S}_\eta)=\frac{1}{\abs{\mathcal Z}}\sum_{v,x}\tr{\rho_{RAB}V^x_v\otimes A^x_v\otimes B^x_v}.
\end{equation}
The optimal winning probability is given by the supremum of \eqref{eq:winnin_lossy_MOE} over all possible strategies, i.e.\  
\begin{equation}
\omega(\mathsf{G}_\eta):=\sup_{\mathbf{S}_\eta}\omega(\mathsf{G}_\eta,\mathbf{S}_\eta).
\end{equation}

\begin{figure}[htbp]
    \centering
    \begin{tikzpicture}[node distance=1cm, auto]
    \node (xV) {$x$};
    \node [right=3cm of xV] (emptyx) {};
    \node [right=0.4cm of emptyx] (x1) {$x$};
    \node [right=3cm of x1] (emptyx1) {};
    \node [right=0.4cm of emptyx1] (x2) {$x$};
    
    \node [below=0.9 of xV] (belowxV){};
    \node [below=0.3cm of xV] (belowxVbefore){};
    \node [left=0.0 of belowxV] (leftbelowxV){};
    \node [above=0.1 of belowxV, yshift=-5pt] (aboveleftbelowxV){$R$};
    \node [below=0.4 of aboveleftbelowxV] (aboveleftbelowxV_meas){};
    \node [below=0.3 of aboveleftbelowxV,yshift=5pt] (aboveleftbelowxV_V){$\{V_v^x\}_v$};
    \node [below=0.3cm of xV] (belowxVbefore){};
    \node [below=1.6cm of xV] (Vans1){};
    \node [below=2.2cm of xV] (Vans2){};
    \node [below=2.3cm of xV] (Vans3){$v$};
    \draw [->] (Vans1) -- (Vans2);

    \node [below=of emptyx] (belowemptyx){};
    \node [below=0.9cm of x1] (belowx1){};
    \node [left=0.0 of belowx1] (leftbelowx1){};
    \node [below=0.3 of x1,yshift=-2pt] (belowemptyx1){$A$};
    \node [below=0.3 of belowemptyx1,yshift=4pt] (MeasAlice){$\{A_a^x\}_a$};
    \node [above=0.1 of leftbelowx1] (aboveleftbelowx1){};
    \node [below=0.3cm of x1] (belowx1before){};
    \node [below=1.6cm of x1] (Aans1){};
    \node [below=2.2cm of x1] (Aans2){};
    \node [below=2.3cm of x1] (Aans3){$a$};
    \draw [->] (Aans1) -- (Aans2);

    \node [below=0.9cm of  x2] (belowx2){};
    \node [left=0.5 of belowx2] (leftbelowx2){};
    \node [below=0.3 of x2,yshift=-2pt] (belowemptyx2){$B$};
    \node [below=0.3 of belowemptyx2,yshift=4pt] (MeasBob){$\{B_b^x\}_b$};
    \node [above=0.1 of leftbelowx2] (aboveleftbelowx2){};
    \node [below=0.3cm of x2] (belowx2before){};
    \node [below=1.6cm of x2] (Bans1){};
    \node [below=2.2cm of x2] (Bans2){};
    \node [below=2.3cm of x2,yshift=4pt] (Bans3){$b$};
    \draw [->] (Bans1) -- (Bans2);


     \node [below=0.6cm of belowxV] (AlicebelowxV){};
     \node [below=0.6cm of belowx1] (Alicebelowx1){};
     
     
     \node [below=1.7cm of belowxV] (2belowxV){};
     \node [below=1.7cm of belowx2] (2belowx2){};

     \node[mark size=2.5pt,color=black] at (belowxV) {\pgfuseplotmark{*}};
     \node[mark size=2.5pt,color=black] at (belowx1) {\pgfuseplotmark{*}};
     \node[mark size=2.5pt,color=black] at (belowx2) {\pgfuseplotmark{*}};

     \node [right=1.8cm of belowxV] (rho){};
     \node [above=0.0cm of rho, yshift=-12pt] (rho1){$\rho_{RAB}$};

     \begin{scope}[on background layer]
    \fill[gray!30, opacity=0.45, rounded corners=2pt] 
        ($(belowxV) + (-0.5,0.25)$) rectangle ($(belowx2) + (0.5,-0.25)$);
    \end{scope}
     
     \newcommand\Square[1]{+(-#1,-#1) rectangle +(#1,#1)}
     \draw (belowxV) \Square{20pt} ; 
     \draw (belowx1) \Square{20pt} ;
     \draw (belowx2) \Square{20pt} ;

     \draw [->] (xV) -- (belowxVbefore);
     \draw [->] (x1) -- (belowx1before);
     \draw [->] (x2) -- (belowx2before);


\end{tikzpicture}
\caption{Schematic representation of a lossy monogamy-of-entanglement game. The gray-shaded region represents the tripartite quantum state $\rho_{RAB}$ prepared by Alice and Bob and shared amongst the three parties. The referee, Alice and Bob are denoted by $R$, $A$ and $B$, respectively.  }
\label{fig:lossyMoEgame}
\end{figure}

As shown in~\cite{TomamichelMonogamyGame2013}, any strategy can be purified in the sense that, by enlarging the Hilbert spaces if necessary, one may assume $\rho_{RAB} = \ketbra{\psi}{\psi}_{RAB}$ for some pure state $\ket{\psi} \in \mathcal{H}_R \otimes \mathcal{H}_A \otimes \mathcal{H}_B$, and that the local measurements $\{A_a^x\}_a$ and $\{B_b^x\}_b$ are projective for all $x \in \mathcal{X}$. From now on, we will assume our strategies are of this purified form.

Next, we describe the $\mathsf{G}^{BB84}$ game, originally introduced in ~\cite{TomamichelMonogamyGame2013}, which was used by the authors to show security of \QPVBB~in the No-PE model. We will later use its lossy version to show security of \QPVBBeta.

\begin{ex} \label{ex:BB84MoEGame} The \emph{BB84 monogamy-of-entanglement game} is described as follows. Alice and Bob prepare a quantum state and send a qubit from it to the referee, who chooses uniformly at random to measure the qubit either in the computational or the Hadamard basis. Upon knowing the choice of basis, the task of Alice and Bob is to guess the measurement outcome. Using the above terminology, the game is given by
\begin{equation}
    \mathsf{G}^{BB84}=\left(\eta=1,\mathcal{M}\right),
\end{equation}
where $\mathcal{M}=\{V_0^{x},V_1^{x}\}_{x\in\{0,1\}},$ with 
\begin{equation}\label{eq V_0^x, V_1^x}
    V_0^{x}=H^x\ketbra{0}{0}H^x\text{ and }V_1^{x}=H^x\ketbra{1}{1}H^x,
\end{equation}
where $H$ is the Hadamard transformation. Varying $\eta\in[0,1]$ defines the lossy BB84 MoE game, denoted by $\mathsf{G}^{BB84}_\eta$. 
\end{ex}

The $\mathsf{G}^{BB84}$ game can be associated with an attack to the \QPVBB{} protocol in the sense that having a strategy to break the protocol in the No-PE model implies having a strategy for the MoE game \cite[Section 5]{TomamichelMonogamyGame2013}, and therefore
\begin{equation}
    \pr{\textsc{V}_\textsc{AB}=\textsc{c}}\leq \omega(\mathsf{G}^{BB84}).
\end{equation}
In a similar manner, it follows that having a  strategy $\mathbf{S}^{BB84}_{\eta}$ for $\mathsf{G}^{BB84}_\eta$ implies having a No-PE strategy for \QPVBBeta. The idea is that in step 4\ in the attack described above, the attackers start with a tripartite state shared among the verifier, Alice and Bob and their task is to correctly guess the measurement outcome of the measurement which is performed on the verifier's register.  In the \QPVBB~protocol case, verifier $V_0$ plays the role of the referee with associated Hilbert space $\mathcal{H}_V=\mathbb{C}^2$, with $\mathcal{X}=\{0,1\}$ and $\mathcal{V}=\{0,1\}$. In the purified version of \QPVBBeta{}, $V_0$ performs the measurement $\mathcal{M}=\{V_0^{x}=H^x\ketbra{0}{0}H^x,V_1^{x}=H^x\ketbra{1}{1}H^x\}_{x\in\{0,1\}}$,  and the two collaborative parties, who correspond to the attackers, want to break the protocol by guessing the verifier's outcome.  We refer the reader to step~4 of the attack on \QPVBBeta{}, illustrated in Fig.~\ref{fig:attack_lossyBB84} and note that it corresponds to a lossy MoE game\footnote{Although not represented in the figure, $V_0$ performs the measurement $\{V^x_v\}_v$ in his local register $V$ of the state $\rho_{VAB}$.}, depicted in Fig.~\ref{fig:lossyMoEgame}. The strategy $\mathbf{S}_{\textsc{tfkw}}=\{\ketbra{\psi}{\psi},A_a^x=\delta_{a0},B_a^x=\delta_{a0}\}$, where $\ket{\psi}_V=\cos\frac{\pi}{8}\ket{0}_V+\sin\frac{\pi}{8}\ket{1}_V$, gives the optimal probability of winning the $\mathsf{G}^{BB84}$ game \cite{TomamichelMonogamyGame2013} (see discussion below) and thus the optimal probability of being correct attacking the \QPVBB~protocol is upper bounded by
\begin{equation}
    \pr{\textsc{V}_\textsc{AB}=\textsc{c}}=
    \frac{1}{2}\sum_{a,x}\tr{\ketbra{\psi}{\psi}V^x_a\otimes A^x_a\otimes B^x_a}=\cos^2\frac{\pi}{8}.
\end{equation}
This strategy also gives 
\begin{equation}
     \pr{\textsc{V}_\textsc{AB}=\textsc{w}}=\sin^2\frac{\pi}{8}, \hspace{1cm}  \pr {\textsc{V}_\textsc{AB}=\lightning}=0. 
\end{equation}
Comparing these probabilities with \eqref{eq succesful 1 round attack}, and considering that $\eta=1$, the attackers could successfully attack one round of the \QPVBB~protocol if $p_{err}\geq\sin^2\frac{\pi}{8}\simeq0.15$. In terms of an attack to \QPVBB, the strategy $\mathbf{S}_{\textsc{tfkw}}$ comes from the attack described as follows: Alice intercepts the state sent by $V_0$ and measures it in the \emph{Breidbart} basis $\{\cos\frac{\pi}{8}\ket{0}+\sin\frac{\pi}{8}\ket1, \sin\frac{\pi}{8}\ket{0}-\cos\frac{\pi}{8}\ket1\}$---associated with $0$ and $1$, respectively---i.e.\  a projective measurement onto the state that has maximum overlap with $\ket{0}$ and $\ket+$. Then, Alice broadcasts the outcome $a$, and both attackers answer $a$ to their respective closest verifier.

Notice that if in a attack the attackers actually answer, meaning that they do not respond `$\perp$', we have
\begin{equation}\begin{split}
    \pr{\textsc{V}_\textsc{AB}=\textsc{c}\mid \textsc{V}_\textsc{AB}\neq\perp}+\pr{\textsc{V}_\textsc{AB}=\textsc{w}\mid \textsc{V}_\textsc{AB}\neq\perp}+\pr{\textsc{V}_\textsc{AB}=\lightning\mid \textsc{V}_\textsc{AB}\neq\perp}=1.
\end{split}
\end{equation}

\noindent In fact, since for QPV we impose $\pr{\textsc{V}_\textsc{AB}=\lightning}=0$, the above expression reduces to
\begin{equation}\label{eq pcorrec+pwrong if no no photon=1}
    \pr{\textsc{V}_\textsc{AB}=\textsc{c}\mid \textsc{V}_\textsc{AB}\neq\perp}+\pr{\textsc{V}_\textsc{AB}=\textsc{w}\mid \textsc{V}_\textsc{AB}\neq\perp}=1.
\end{equation}
We define the probability of winning, $p_{win}$ as the maximum probability of being correct conditioned on answering, i.e.\  $p_{win}:=\max\pr{\textsc{V}_\textsc{AB}=\textsc{c}\mid \textsc{V}_\textsc{AB}\neq\perp}$, which has the interpretation of the normalized (over the conclusive answers) optimal probability of answering `\textsc{Correct}'. Showing that $p_{win}$ has a constant gap below one would imply that, over the conclusive rounds, the attackers cannot be correct as many times as they want, and if $p_{win}$ is below the expected value of an honest prover, i.e.\  $\eta(1-p_{err})$, the attackers will not be able to mimic her behavior. On the other hand, for our security approach, we will consider the probability that the attackers can actually play (answer)---not answering `$\perp$'---in the $\mathsf{G}^{BB84}_\eta$ game, i.e.\ given a strategy $\mathbf S_\eta=\{\ket\psi,A^z_a,B^z_b\}_{z,a,b}$, the probability that they answer, $p_{ans}$,
\begin{equation}
\label{eq pans BB84}
p_{ans}:=\pr{\textsc{V}_\textsc{AB}=\textsc{c}}+\pr{\textsc{V}_\textsc{AB}=\textsc{w}}=\frac{1}{2}\sum_{a,x\in \{0,1\}}\bra{\psi}A_a^{x} B_a^{x}\ket{\psi},
\end{equation}
where we used the following simplified notation: when clear from the context, tensor products, identities and $\psi$ will be omitted, e.g.\ $\bra{\psi}V_0^x\otimes A_1^{x}\otimes B_1^{x}\ket{\psi}=\expectedbraket{V_0^xA_1^{x}B_1^{x}}$. Moreover, if attackers want to mimic an honest prover, they have to be consistent with the error $p_{err}$, that is, 
\begin{equation}
\label{eq condition perr by def 1}
    \frac{\expectedbraket{V_{0}^{x}A_{1}^{x}B_{1}^{x}}}{\expectedbraket{V_{0}^{x}(A_{0}^{x}+A_{1}^{x})(B_{0}^{x}+B_{1}^{x})}}\leq p_{err},  \hspace{1cm}\frac{\expectedbraket{V_{1}^{x}A_{0}^{x}B_{0}^{x}}}{\expectedbraket{V_{1}^{x}(A_{0}^{x}+A_{1}^{x})(B_{0}^{x}+B_{1}^{x})}}\leq p_{err},
\end{equation}
where we impose that the error rate for both outputs 0 and 1 is upper bounded by the same amount for all inputs $x$.

Notice that if $p_{ans}=1$ while satisfying \eqref{eq condition perr by def 1}, the attackers can always attack the protocol without being caught. Using \eqref{eq pans BB84}, the security of the protocol can be regarded as the maximum probability that the attackers can respond without being caught, and the protocol will be proven to be secure if the attackers cannot reproduce $p_{ans}\geq \eta$ for a given $p_{err}$, we formalize this idea in the following definition.

\begin{definition}
We define the security region $SR$ of the \QPVBBeta~ protocol as  the set of pairs $(p_{err},p_{ans})\in [0,1]\times [0,1]$ for which no strategy $\mathbf{S}^{\eta}_{MoE}$ (and thus no No-PE strategy) exists that breaks the \QPVBBeta~ protocol with the corresponding error and response rate.

A subset of $SR$ will be denoted by $SSR$. We define the attackable region $AR$ as the complementary set of the $SR$. A subset of $AR$ will be denoted as $SAR$. 
\end{definition}

Therefore, our interest relies on maximizing expression \eqref{eq pans BB84} over all the strategies  $\mathbf{S}^{\eta}_{MoE}$ to break the \QPVBBeta~ protocol. However, unlike the set of probabilities achievable by classical physics, the set of probabilities attainable by quantum mechanics, $\mathcal{Q}$, has uncountably many extremal points, see e.g.\ ~\cite{Bell_non_locality_report}, and therefore it makes the optimization problem a tough task. On the positive side, in \cite{NPA2008}, Navascués, Pironio and Acín (NPA) introduced a recursive way to construct subsets ${\mathcal{Q}_{\ell}\supset\mathcal{Q}_{\ell+1}\supset\mathcal{Q}}$ for all $\ell\in\mathbb{N}$  with the property that each of them can be tested using SDP and are such that $\cap_{\ell\in\mathbb{N}}\mathcal{Q}_{\ell}=\mathcal{Q}_{co}$, where $\mathcal{Q}_{co}\supset \mathcal Q$ is the set of probabilities obtained by Alice and Bob performing commuting measurements on a joint Hilbert space instead of tensor product measurements. For finite-dimensional Hilbert spaces, both sets are equivalent.  

For all $a,b\in\{0,1,\perp\}$ and all $x,x'\in\{0,1\}$, the elements $\expectedbraket{A_a^{x}B_{b}^{x'}}$ will appear in the maximization problem solvable via SDP, and they are bounded by linear constraints given by $\mathcal{Q}_{\ell}$, see Appendix~\ref{appendix non-logal games and NPA hierarchy}. In addition to these constraints, we impose the additional linear constraints derived from \QPVBBeta, i.e.\ since in the protocol the verifiers abort if they receive different messages, from \eqref{eq abort},
\begin{equation}\label{eq restriction error epsilon BB84}\expectedbraket{A_a^{x}B_{b}^{x}}=0 \hspace{5mm}\forall a\neq b\in \{0,1,\perp\}, \forall x \in \{0,1\},\end{equation}
and the prover subject to a measurement error $p_{err}$, see Proposition~\ref{prop ineq perr for QPV BB84}.

\begin{prop} \label{prop ineq perr for QPV BB84} Let $a,b\in\{0,1\}$. For all $x,x'\in\{0,1\}$, the terms $\expectedbraket{A_a^{x}B_{b}^{x'}}$ can be bounded by $p_{err}$ by the following inequality:
\begin{equation}
    \label{eq perr constraint sum ab BB84}
    \sum_{ab}(2-\norm{V_a^x+V_b^{x'}})\expectedbraket{A_a^{x}B_{b}^{x'}}\leq p_{err}\sum_{a}(\expectedbraket{A_a^{x}B_{a}^{x}}+\expectedbraket{A_a^{x'}B_{a}^{x'}}).
\end{equation}
\end{prop}
The proof Proposition~\ref{prop ineq perr for QPV BB84} is a particular case of the proof of Proposition~\ref{prop ineq perr for QPV}. \\
The value of $p_{ans}$ in \eqref{eq pans BB84} can be therefore upper bounded by the SDP problem:
\begin{equation}
\boxed{
\label{eq upperbound p_ans all constraints BB84}
\begin{split}
    &\max \frac{1}{2}\sum_{x,a\in\{0,1\}}\expectedbraket{A_a^{x}B_a^{x}};\\
    &\textrm{subject to: the linear constraints for } \mathbf{S}^{\eta}_{MoE} \in \mathcal{Q}_\ell,\\& \textrm{ \hspace{15mm} and equations \eqref{eq restriction error epsilon BB84} and \eqref{eq perr constraint sum ab BB84}}. 
\end{split}}
\end{equation}

Where, abusing notation, we denoted $\mathbf{S}^{\eta}_{MoE}\in \mathcal{Q}_\ell$ meaning that the probabilities obtained from $\mathbf{S}^{\eta}_{MoE}$ belong to the set $\mathcal{Q}_{\ell}$. Fig.~\ref{Fig SR bb84} shows the solution of the SDP~\eqref{eq upperbound p_ans all constraints BB84} for different values of $p_{err}$ for the first and second level of the NPA hierarchy using the Ncpol2sdpa package \cite{Wittek_2015_Ncpol2Sdpa} in Python.
The values above the solution for any given $p_{err}$ represent points where does not exist an attack such that $p_{ans}\geq\eta$ and therefore correspond to $SSR$, the area represented in light blue. The results plotted in Fig.~\ref{Fig SR bb84} coincide with the tight bound of the winning probability of the MoE game attacking the $\mathrm{QPV_{BB84} }$ protocol, since $p_{ans}$ reaches 1 for $p_{err}=0.1464\simeq 1-\cos^2(\pi/8)$.\\

\begin{prop} \label{prop monotonicity} The function $p_{ans}(p_{err})$ for $p_{err}\in[0,1]$ obtained by the solution of \eqref{eq upperbound p_ans all constraints BB84} is monotonically increasing, i.e.\ if $p_{err}^0\leq p_{err}^1$, then $p_{ans}(p_{err}^0)\leq p_{ans}(p_{err}^1)$. 
\end{prop}

\begin{proof}
It follows from the fact that $p_{ans}(p_{err}^1)$ is obtained by an SDP which relaxation of the restrictions of the SDP providing  $p_{ans}(p_{err}^0)$.
\end{proof}

Informally, Proposition~\ref{prop monotonicity} assures that between two numerical solutions for different $p_{err}$ there are no `abrupt jumps', more specifically, in Fig.~\ref{Fig SR bb84}, any solution between two plotted points cannot be grater than the point in the right.

Consider the strategy $\mathbf{S}_{MoE}^{mix}|_{p}$ given by the probabilistic mixture of playing the strategy $\mathbf{S}_{\textsc{tfkw}}$ with probability $p$ and $\mathbf{S}_{guess}$ with probability $1-p$, conditioned on answering. As long as $p_{ans}<1$, for each $p$, this mixture gives a unique pair of $(p_{err},p_{ans})$ (for $p_{ans}=1$, take the minimum $p_{err}$), and we equivalently denote $\mathbf{S}_{MoE}^{mix}|_{p}$ by the corresponding $(p_{err},p_{ans})$ as $\mathbf{S}_{MoE}^{mix}|_{(p_{err},p_{ans})}$. The values of $p_{ans}$ obtained by this strategy, see continuous line in Fig.~\ref{Fig SR bb84}, provide a region where the protocol is attackable, i.e.~a $SAR$.  

Since the SSR obtained from the second level of the NPA hierarchy and the $SAR$ obtained from $\mathbf{S}_{EoM}^{mix}|_p$ are such that $SSR\cup SAR=[0,1]\times[0,1]$, up to infinitesimal precision, it means that they correspond to $SR$ and $AR$, respectively, i.e.\ the solutions of the SDP \eqref{eq upperbound p_ans all constraints BB84} for $\ell=2$ converge to the quantum value and are tight. This means that Fig.~\ref{Fig SR bb84} represents a full characterization of the security of the \QPVBBeta~protocol under photon loss with attackers that do not pre-share entanglement, and the light blue region encodes all the points $(p_{err},\eta)$ where the protocol is secure. The result is summarized as follows:

\begin{result} In the No-PE model, if attackers answer with probability $\eta$ and never respond inconsistent answers, the optimal probability that they answer `\textsc{Correct}' in a round of \QPVBBeta~for $\eta\in[\frac{1}{2},1]$ is given by
\begin{equation}
\max\Pr[\textsc{V}_\textsc{AB}=\textsc{c}]=\cos^2\big(\frac{\pi}{8}\big)\eta+\sin^2\big(\frac{\pi}{8}\big)(1-\eta).
\end{equation}
\end{result}

\begin{figure}[H]
\label{Fig SR bb84}
\centering
\includegraphics[width=105mm]{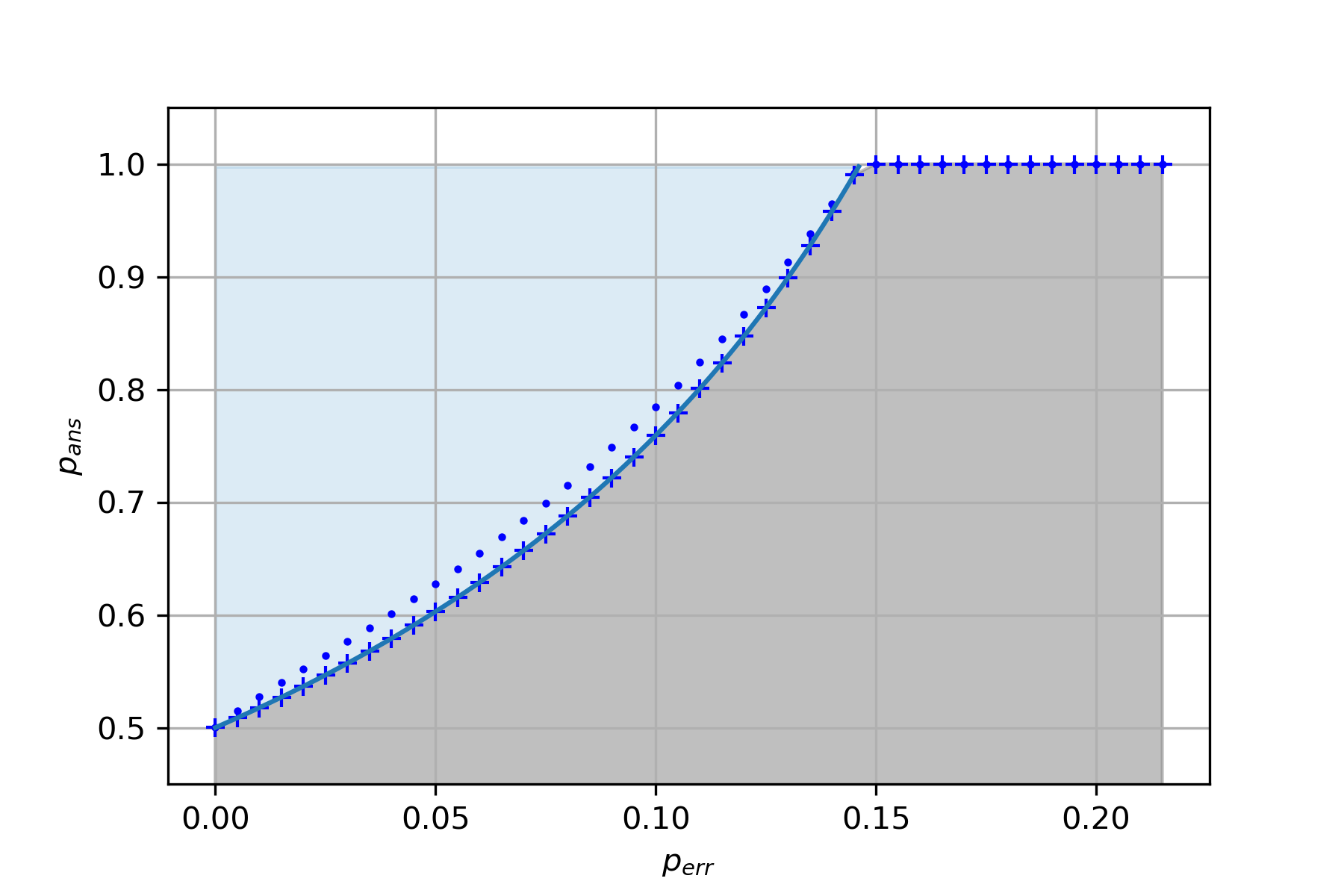}
\caption{Solutions of the first, $\ell=1$, (blue dots) and second level, $\ell=2$, (blue pluses) of the NPA hierarchy for the SDP \eqref{eq upperbound p_ans all constraints BB84}. The light blue and the gray area correspond to $SR$ and $AR$, respectively. The continuous line represents $\mathbf{S}_{MoE}^{mix}|_p$.}
\end{figure}

\subsection{Sequential repetition of \texorpdfstring{\QPVBBeta}{eta QPV BB84}}\label{section:sequential rep BB84-eta}

After $r$ sequential repetitions of \QPVBBeta, the verifiers have to accept or reject the location of the prover. In order to study the security of the sequential repetition, we introduce the following notation and concepts. $\mathbb E$ denotes the expected value and $\textbf{1}_*(\textbf{a})=1$ if $*=a$ and $0$ otherwise is the indicator function.
For a random variable $X$, taking values on a finite set $\mathfrak{X}=\{\texttt{x}_1,...,\texttt{x}_d\}$, a probability distribution $p$ 
is specified by  $p_{\texttt{x}_i}=\Pr[X=\texttt{x}_i], \texttt{x}_i\in \mathfrak{X},$ and $p$ can be represented by a probability vector $\textbf{p}=(p_{\texttt{x}_1},...,p_{\texttt{x}_d})$. The set of all probability distributions $\textbf{p}$ over $\mathfrak X$ is $\Delta_{d-1}=\{\textbf{p}\in\mathbb R^d\mid \sum_{\texttt{x}_i\in \mathfrak{X}}p_{\texttt{x}_i}=1, p_i\geq0\}$, which is known as the probability simplex, and it is a ($d-1$)-dimensional manifold.

Recall that in every round of the protocol, attackers will pick a strategy $\mathbf S^i_\eta$ that can depend on the previous rounds.  Assume that the verifiers did not receive any `\textsc{Abort}' answers, otherwise, they reject the location. Let $\Gamma_r^{att}$ denote the total score that the attackers ($att$) get, defined in \eqref{eq:score_BB84}. In the next theorem, we show that attackers in the No-PE model will fail the test $\mathsf{T}^{r\mathrm{BB84}}_{\varepsilon_{\textrm{h}}}$ with exponentially high probability.

\begin{theorem}\label{thm sequential repetition no-entanglement} Consider the $r$ sequential repetition of \QPVBBeta. Let $\varepsilon_{\textrm{h}}>0$, $\eta$ and $p_{err}$ be such that $\alpha-\delta>0$, with  $\delta=(\cos^4\frac{\pi}{8}\ln(1/\varepsilon_{\textrm{h}})/r)^{1/2}$. Then, any sequential strategy to attack \QPVBBeta~ in the No-PE model fulfills that $\mathbb E [\Gamma_r^{att}]\leq 0$. Moreover, the probability that the attackers are accepted in the $\mathsf{T}^{r\mathrm{BB84}}_{\varepsilon_{\textrm{h}}}$ test is exponentially small: 
\begin{equation}
\pr{\Gamma_r^{att}\geq r(\alpha-\delta)}\leq e^{-r (\alpha-\delta)^2/2}.
\end{equation}
\end{theorem}

Theorem~\ref{thm sequential repetition no-entanglement} shows that there exists a test which is both complete and sound, and thus, the \QPVBBeta{} is secure for the corresponding $\eta$ and $p_{err}$. The existence of this test implies that after $r$ rounds, attackers will be \emph{caught} with exponentially high probability.

The points $(\eta,p_{err})$ such that $\alpha>0$ correspond to the blue region in Fig.~\ref{Fig SR bb84} and also below the black dots in Fig.~\ref{Fig simplex}---the above $\alpha-\delta$ corresponds to a small shift that can be made small increasing the number of sequential repetitions. The proof of Theorem~\ref{thm sequential repetition no-entanglement} is a particular case of the proof of Theorem~\ref{theorem seq repetition entangled}.

\begin{figure}[h]
\centering
\includegraphics[width=85mm]{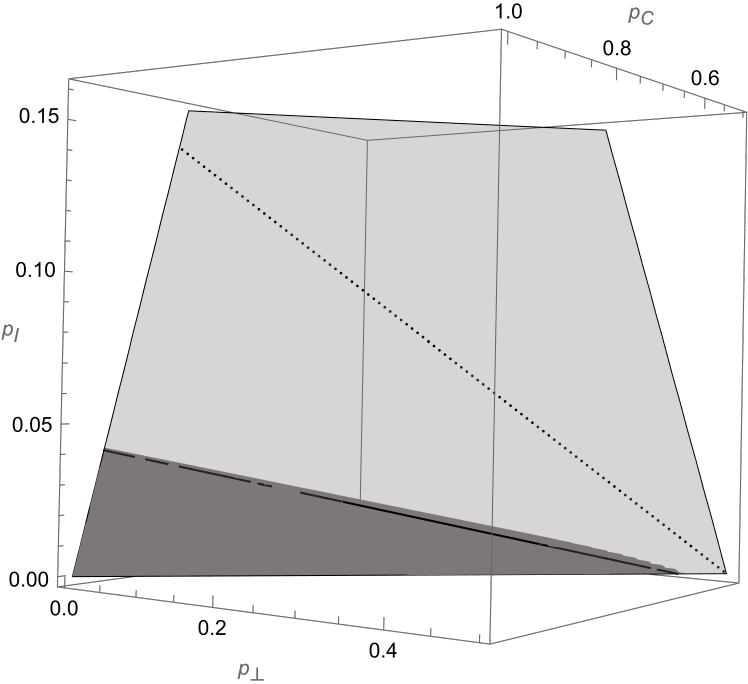}
\caption{Probability simplex $\Delta_2$ for probabilities taking values on $\mathfrak{X}=\{\textsc{c},\perp,\textsc{i}\}$, where \textsc{i} (incorrect) denotes either \textsc{w} or $\lightning$. Black dots correspond to numerical solutions of \eqref{eq upperbound p_ans all constraints BB84} for $\ell=2$. The dark region is the set of probabilities that Theorem \ref{theorem q>=n/2-5 for BB84 eta} excludes and the black straight line is the intersection between $\Delta_2$ and the plane $\gamma_{\textsc{c}}p_{\textsc{c}}-\gamma_{\perp}p_{\perp}-\gamma_{\textsc{i}}p_{\textsc{i}}=0$.}
\label{Fig simplex}
\end{figure}

\section{The \texorpdfstring{\QPVBBetaf}{eta-f-QPV BB84} protocol and its security under entangled attackers}\label{section 2 basis and entanglement}

In the previous section, we have shown security for the \QPVBBeta~protocol. Our argument was restricted to attackers who do not pre-share entanglement, and it is well-known \cite{OriginalQPV_Kent2011} that such a protocol is broken with a single EPR pair. Recent work by Bluhm, Christandl and Speelman~\cite{bluhm2022single}, building on work by Buhrman, Fehr, Schaffner, and Speelman \cite{Buhrman_2013}, has shown that when adding classical information to \QPVBB, the \QPVBBf~ protocol (see below), the protocol remains secure if the attackers hold less than $n/2-5$ qubits, making it secure against entangled attackers that hold an entangled state of smaller dimension.\footnote{Note that this is only a lower bound. The actual best attack known requires an amount of entanglement that is \emph{exponential} in the number of classical bits---it is very much possible that the protocol is better than proven.}
Recall that an attack on a QPV protocol can conceptually be split up into two rounds---the first in which the attackers hold part of the input and their pre-shared entangled state, and the second round in which the attackers have to respond to the closest verifier. 
The key to showing their result is by considering two possible joint states held by the attackers as a result of their first-round actions---one for an input where they have to measure the input qubit in the computational basis and one where they have to measure the qubit in the Hadamard basis. Then, these two states already have to be `far apart' (even though Alice's part does not depend on Bob's input yet, and vice-versa), which by a counting argument makes it possible to obtain a bound on the size of the pre-shared state.
In the spirit of applying an analogous counting argument to show security for the lossy case, we use the results of Section \ref{section 2 basis no entanglement} to prove a lemma (Lemma \ref{lemma D>=eta Delta}) that allows us to apply a similar counting argument.

We define one round of the lossy-function-BB84 protocol as follows:

\begin{definition}\label{def qpv bb84 f}
Let $n\in\mathbb{N}$, and consider a $2n$-bit boolean function $f:\{0,1\}^n \times \{0,1\}^n \to \{0,1\}$. We define one round of the lossy-function-BB84 protocol, denoted by \QPVBBetaf, as follows: 
\begin{enumerate}
    \item $V_0${} and $V_1${} secretly agree on random bit strings $x,y\in\{0,1\}^n$ and a bit $v\in\{0,1\}$. Then, $V_0${} prepares the qubit state $H^{f(x,y)}\ket{v}\in\{\ket0,\ket1,\ket+\,\ket-\}$.
    \item $V_0${} sends the qubit  $H^{f(x,y)}\ket{v}$, and $x$ to \prover{}, whereas  $V_1${} sends $y$ to \prover{}, coordinating their times so that the messages arrive simultaneously to \prover{}. The classical information is required to travel
at the speed of light, whereas the quantum information can be sent arbitrarily slow.
    \item Immediately, \prover{} computes $f(x,y)$, measures the received qubit in the basis $f(x,y)$, and broadcasts her outcome, either 0 or 1, to $V_0${} and $V_1$. If she did not receive the qubit, i.e.\  the photon was lost, she sends $\perp$. Therefore, the possible answers from \prover{} are $v_{\textsf{P}}\in\{0,1,\perp\}$. 
    \item  If 
    \begin{enumerate} \item $V_0${} and $V_1${} receive their respective answers at the time corresponding to the claimed location, and they are equal, i.e.\  both receive the same $v_{\textsf{P}}$, then, if
        \begin{itemize}
            \item $v_{\textsf{P}}=v$, the verifiers record `\textsc{Correct}', denoted by `\textsc{c}', 
            \item $v_{\textsf{P}}=1-v$, the verifiers record `\textsc{Wrong}', denoted by `\textsc{w}', 
            \item $v_{\textsf{P}}=\perp$, the verifiers record `\textsc{No photon}',  denoted by `$\perp$', 
        \end{itemize}
        \item 
        otherwise, they record `\textsc{Abort}',  denoted by `$\lightning$', and abort the protocol rejecting the location.
    \end{enumerate}
\end{enumerate}
\end{definition}

\begin{figure}[htbp]
    \centering
    \begin{tikzpicture}[node distance=3cm, auto]
    \node (A) {$V_0$};
    \node [left=1cm of A] {};
    \node [right=of A] (B) {\prover{}};
    \node [right=of B] (C) {$V_1$};
    \node [right=1cm of C] {};
    \node [below=of A] (D) {};
    \node [below=of B] (E) {}; 
    \node [above=0.5cm of A] (posV00) {};
    \node [left=1cm of posV00] (posV0) {};
    \node [above=0.5cm of C] (posV11) {};
    \node [right=1cm of posV11] (posV1) {};
    \draw [->] (posV0) -- (posV1) node[midway] {position};

    \node [above=-0.1cm of A] (N) {};
    \node [right=0.8cm of N] {$\diagdots[120]{2.5em}{0.1em}$};
    
    \node [right=0.8cm of A] (M) {};
    
    \node [below=of C] (F) {};
    \node [below=of D] (G) {};
    \node [below=of E] (H) {};
    \node [below=of F] (I) {};
    \node [left= 6cm of E] (J) {};
    \node [below= 3cm of J] (K) {};
    \node [above= 3cm of J] (L) {};

    \draw [->, transform canvas={xshift=0pt, yshift=0pt}, quantum] (M) -- (E) node[midway] (x) {} ;
    \draw [->] (A) -- (E);
    \draw [->] (C) -- (E);
    \draw [->] (E) -- (I) node[midway] (q) {$v_{\textsf{P}} \in \{0,1,\perp\}$};
    \draw [->] (E) -- (G);

    \draw [->] (L) -- (K) node[midway] {time};

    \node[left=0.1cm of x, yshift=13pt, xshift=45pt, transform canvas={xshift=+ 2pt, yshift = +2 pt}] {$H^{f(x,y)}\ket{v}$};
    \node[left=1.5cm of x, transform canvas={xshift=+ 2pt, yshift = +2 pt}] {$x\in\{0,1\}^n$};
    \node[right = 3.3cm of x, transform canvas={xshift=+ 2pt, yshift = +2 pt}] {$y \in \{0,1\}^n$};
    \node[left = 3.3cm of q] {$v_{\textsf{P}} \in \{0,1,\perp\}$};
\end{tikzpicture}
\caption{Schematic representation of the \QPVBBetaf~protocol, where straight lines represent classical information and undulated lines represent quantum information.}
\label{fig:protocol-BB84f}
\end{figure}

See Fig.~\ref{fig:protocol-BB84f} a schematic representation of \QPVBBetaf.  The \QPVBBf~protocol  corresponds to \QPVBBetaf~ for $\eta=1$ and $p_{err}=0$. In the same way as for \QPVBBeta, after $r$ rounds, the verifiers expect the answers to satisfy  \eqref{eq succesful protocol}.

Similarly to \QPVBBeta{} in Section~\ref{section:sequential rep BB84-eta}, here, we present a binary test to either \emph{accept} or \emph{reject} the location based on the observed data.  Let $\textbf{a}_i\in\{\textsc{c},\perp,\textsc{I}\}$ denote whether the answer that the verifiers recorded in the round $i$ was `\textsc{Correct}', `\textsc{No photon}',  and `\textsc{Wrong}' (or `\textsc{Abort}')\footnote{The symbol `\textsc{i}' caputres the \emph{incorrect} answers `\textsc{w}' and `$\lightning$'.}, respectively. Consider the payoff function  
\begin{equation}T^{ent}_i(\textbf{a}_i)=\gamma_{\textsc{c}}\textbf{1}_{\textsc{c}}(\textbf{a}_i)-\gamma_{\perp}\textbf{1}_{\perp}(\textbf{a}_i)-\gamma_{\textsc{i}}\textbf{1}_{\textsc{i}}(\textbf{a}_i),    
\end{equation}
where $(\gamma_{\textsc{c}},\gamma_{\perp},\gamma_{\textsc{i}})=\frac{1}{\sqrt{488625947}}(943,1107,22057)$ and the superscript $ent$ is due to that this payoff will be used to securely verify the location even if attackers who pre-share entanglement intend to break \QPVBBetaf. Let 
\begin{equation}
    \Gamma_r^{ent}=\sum_{i=1}^r T_i^{ent}(\textbf{a}_i),
\end{equation}
be the total \emph{score} after $r$ rounds. We next introduce the following acceptance test with confidence level $\varepsilon_{\textrm{h}}$.

\begin{definition} Let $\varepsilon_{\textrm{h}}>0$. For the \QPVBBetaf~protocol executed sequentially $r$ times, we define the acceptance test $\mathsf{T}^{r f-\mathrm{BB84}}_{\varepsilon_{\textrm{h}}}$, also referred to as the decision criterion, as follows:     the verifiers \emph{accept} the prover's location if
\begin{equation}\label{eq:test-BB84-qpvent}
    \Gamma_r^{ent} \geq r (\alpha^{ent}-\delta),
\end{equation}
where $\delta=\sqrt{\frac{4\ln(1/\varepsilon_{\textrm{h}})}{r}}$. 
Otherwise, they \emph{reject}. 
\end{definition}
For an honest prover ($hp$), for every $i$, 
\begin{equation}
    \mathbb E[T_i^{ent,hp}] = \gamma_{\textsc{c}}\eta(1-p_{err})-\gamma_{\perp}(1-\eta)-\gamma_{\textsc{i}}\eta p_{err}:\alpha^{ent}(\eta,p_{err}),
\end{equation}
and therefore, $\mathbb E [\Gamma_r^{ent,hp}]=r\alpha^{ent}$ for simplicity we will assume the dependence on $\eta$ and $p_{err}$ implicit. By Hoeffding's inequality, see Lemma~\ref{lemma:Hoeffding}, with $T_i^{ent}\in[-\frac{22057}{\sqrt{488625947}},\frac{943}{\sqrt{488625947}}]\subset[-1,1]$, an honest prover will be accepted except with small probability at most $\varepsilon_{\textrm{h}}$, and therefore the test $\mathsf{T}^{r f-\mathrm{BB84}}_{\varepsilon_{\textrm{h}}}$ is \emph{complete}. We will show that, any attackers who pre-share a linear amount of qubits in~$n$, after $r$ sequential executions of \QPVBBetaf{}, will fail the test with exponentially high probability, and thus showing \emph{soundness}. Showing both completeness and soundness will prove that \QPVBBetaf{} is \emph{secure}. 

Similarly as $\mathsf{T}^{r\mathrm{BB84}}_{\varepsilon_{\textrm{h}}}$, this test is engineered from the analysis of the correlations attainable by attackers in Section~\ref{section 2 basis and entanglement}, ensuring that they are rejected except with negligible probability; see the proof of Theorem~\ref{theorem seq repetition entangled} for details of its construction.

As introduced in Section~\ref{section 2 basis no entanglement}, for the security analysis, we will consider the purified version of \QPVBBetaf, which is equivalent to it, and the difference relies on, $V_0${} preparing an  EPR pair $\ket{\Phi^+}_{VP}$, keeping a qubit register $V$, and sending the register $P$ to the prover. In a later moment, the $V_0${} performs the measurement  $\{V^{f(x,y)}_{v}=H^{f(x,y)}\ketbra{v}{v}_VH^{f(x,y)}\}_{v\in\{0,1\}}$ in his local register $V_0$. 

Consider a general attack to the \QPVBBetaf~protocol, where Alice and Bob take the same role as in Section~\ref{section 2 basis no entanglement}, but in addition, prior to the execution of the protocol, they can pre-share entanglement. In the most general attack to \QPVBBetaf, Alice and Bob proceed as follows:

\begin{enumerate} \item Alice intercepts the qubit state with register $P$ sent by $V_0${}, and applies an arbitrary quantum operation to it and to a local register that she possesses, possibly entangling them. She keeps part of the resulting state, and sends the rest to Bob. Since the qubit $P$  can be sent arbitrarily slow by $V_0${} (the verifiers only time the classical information), this happens before Alice and Bob can intercept $x$ and $y$. 

\item Alice intercepts $x$ and Bob intercepts $y$. At this stage, Alice, Bob, and $V_0${} share a quantum state $\ket{\varphi}$, make a partition and let $q$ be the number of qubits that Alice and Bob each hold, recall that $V_0${} holds a qubit with register $V$, and thus the three parties share a quantum state $\ket{\varphi}$ of $2q+1$ qubits.  Alice and Bob apply a unitary $U_{A_\text{k}A_\text{c}}^{x}$ and $V_{B_\text{k}B_\text{c}}^{y}$ on their local registers $A_\text{k}A_\text{c}=:A$ and $B_\text{k}B_\text{c}=:B$, where k and c denote the registers that will be kept and communicated, respectively. Due to the Stinespring dilation, we consider unitary operations instead of quantum channels. They end up with the quantum state ${\ket{\psi_{xy}}=\mathbb{I}_{V}\otimes U_{A_\text{k}A_\text{c}}^x\otimes V_{B_\text{k}B_\text{c}}^y\ket\varphi}$.
\item Alice sends register $A_\text{c}$ and $x$ to Bob (and keeps register $A_\text{k}$), and Bob sends register $B_\text{c}$ and $y$ to Alice (and keeps register $B_\text{k})$. 

\item Alice and Bob perform POVMs  $\{ A^{xy}_{a}\}_{a\in\{0,1,\perp\}}$ and $\{ B^{xy}_{b}\}_{b\in\{0,1,\perp\}}$ on their local registers $A_\text{k}B_\text{c}=:A'$ and $B_\text{k}A_\text{c}=:B'$, and answer their outcomes $a$ and $b$ to their closest verifier, respectively. 
\end{enumerate}

\begin{figure}[ht]
    \centering
    \begin{tikzpicture}[node distance=3cm, auto]
   \node (V0) {$V_0$};   
    \node [left=2cm of V0](t0){};
    \node [below=7.3cm of t0](t1){};
    \draw [->] (t0) -- (t1) node[midway] {time};

    \node [right=1cm of V0] (A) {Alice};
    \node [right= 5cm  of A] (B) {Bob};
    \node [right= 2.5cm  of A] (P) {};
    \node [right=1cm of B] (V1) {$V_1$};
    \node [below =0.6cm of V0](below_V0){};
    \node [below =0.5cm of V1](below_V1){};

    \node [below =6.15cm of V0](V_ABans){};

    \node [below =7.1cm of V0, xshift=22pt](V0_ans){};
    \node [below =7.1cm of V1, xshift=-22pt](V1_ans){};
    
    \node [below =0.2cm of P](middle){$\ket\varphi_{VAB}$};
    \node [above =0.25cm of middle](middle0){};

    \node [below =0.1cm of middle](middle0){$\ket{\psi_{xy}}_{VAB}$};
    \node [below =0.5cm of A](A_intercepts_Q){};
    \node [below =0.5cm of B](B_intercepts_Q){};
    \node [below =0.5cm of A](A_intercepts_x){$U^x$};
    \node [below =0.5cm of B](B_intercepts_y){$V^y$};

    \node [above=0.5cm of A] (posV00) {};
    \node [left=2.5cm of posV00] (posV0) {};
    \node [above=0.5cm of B] (posV11) {};
    \node [right=2.5cm of posV11] (posV1) {};
    \draw [->] (posV0) -- (posV1) node[midway] {position};

    \node [below =5cm of A](A_commits){};
    \node [below =5cm of B](B_commits){};
    \node [below =6cm of A](A_answers){$\{A^{xy}_a\}_a$};
    \node [below =0.1cm of A_answers, ](A_POVM){};
    \node [below =6cm of B](B_answers){$\{B^{xy}_b\}_b$};
    \node [below =0.1cm of B_answers](B_POVM){};

    \node [left=0.3cm of A](leftA){};
    \node[below=5.6cm of leftA](cA){};
    \node [left=0.2cm of A](leftA2){};
    \node[below=6.6cm of leftA2](aA){};

    \node [right=0.3cm of B](leftB){};
    \node[below=5.6cm of leftB](cB){};
    \node [right=0.2cm of B](leftB2){};
    \node[below=6.6cm of leftB2](bB){};
    
    \node [below =6.2cm of V0](V0bis){};
    \node [below =6.2cm of V1](V1bis){};
    \node [below =7.2cm of V0](V0bis1){};
    \node [below =7.2cm of V1](V1bis1){};

    \draw [->, transform canvas={xshift=0pt, yshift = 0 pt}, quantum] (A_intercepts_x) -- (B_answers) node[midway] (x) {};

    \draw [->, transform canvas={xshift=0pt, yshift = 0 pt}, quantum] (B_intercepts_y) -- (A_answers) node[midway] (x) {};

    \draw [->, quantum] (A_intercepts_x) -- (A_answers) {};
    \draw [->, quantum] (B_intercepts_y) -- (B_answers) {};
    \draw [->] (A_answers) -- (V0_ans) {};
    \draw [->] (B_answers) -- (V1_ans) {};

    \node [right =2.5cm of A_answers](middle2){};
    \node [below =4.9cm of middle,yshift=-9pt](middle3){$\ket{\psi_{xy}}_{VA'B'}$};

\begin{scope}[on background layer]
    \fill[gray!30, opacity=0.45, rounded corners=2pt] 
        ($(A_intercepts_x) + (-2.5,0.7)$) rectangle ($(B_intercepts_y) + (0.7,+0.2)$);
\end{scope}

\begin{scope}[on background layer]
    \fill[gray!30, opacity=0.45, rounded corners=2pt] 
        ($(A_intercepts_x) + (-2.5,-0.7)$) rectangle ($(B_intercepts_y) + (0.7,-0.2)$);
\end{scope}

\begin{scope}[on background layer]
    \fill[gray!30, opacity=0.45, rounded corners=2pt] 
        ($(A_answers) + (-2.5,0.25)$) rectangle ($(B_answers) + (0.7,-0.25)$);
\end{scope}
    
\end{tikzpicture}
\caption{Schematic representation of a general attack on \QPVBBetaf. Straight lines represent classical information, and undulated lines represent quantum information, including $x$ and $y$. The gray-shaded regions represent the corresponding tripartite quantum states in them.}
\label{fig:attack-parallel_repBB84}
\end{figure}

See Fig.~\ref{fig:attack-parallel_repBB84} for a schematic representation of the general attack to \QPVBBetaf. The tuple $\mathbf S_{\eta}^{f-\text{BB84}}:=\{\ket\varphi,U^x,V^y,\{A^{xy}_a\}_a,\{B^{xy}_b\}_b\}_{x,y}$ will be called a $q$-qubit strategy for \QPVBBetaf. Then, given the strategy $\mathbf S_{\eta}^{f-BB84}$, the  probabilities that the verifiers, after the attackers' actions (for the random variable $\textsc{V}_{\textsc{AB}}$, as denoted above) record \textsc{Correct} (\textsc{c}), \textsc{Wrong} (\textsc{w}), \textsc{No photon} ($\perp$), and \textsc{Abort} ($\lightning$)  are, respectively, given by 
\begin{equation}\label{eq succesful 1 round attack with x and y}
\begin{split}
    q_{\textsc{c}}:=\pr{\textsc{V}_\textsc{AB}=\textsc{c}}&=\frac{1}{2^{2n}}\sum_{a\in \{0,1\},x,y\in\{0,1\}^n}\tr{\ketbra{\psi_{xy}}{\psi_{xy}} V^{f(x,y)}_a\otimes A^{xy}_a\otimes B^{xy}_a},\\
     q_{\textsc{w}}:=\pr{\textsc{V}_\textsc{AB}=\textsc{w}}&=\frac{1}{2^{2n}}\sum_{a\in \{0,1\},x,y\in\{0,1\}^n}\tr{\ketbra{\psi_{xy}}{\psi_{xy}} V^{f(x,y)}_a\otimes A^{xy}_{1-a}\otimes B^{xy}_{1-a}},\\
      q_{\perp}\pr{\textsc{V}_\textsc{AB}=\perp}&=\frac{1}{2^{2n}}\sum_{x,y\in\{0,1\}^n}\tr{\ketbra{\psi_{xy}}{\psi_{xy}} \mathbb{I}_V\otimes A^{xy}_{\perp}\otimes B^{xy}_{\perp}},\\
        q_{\lightning}:=\pr{\textsc{V}_\textsc{AB}=\lightning}&=\frac{1}{2^{2n}}\sum_{a\neq b\in \{0,1,\perp\},x,y\in\{0,1\}^n}\tr{\ketbra{\psi_{xy}}{\psi_{xy}} \mathbb{I}_V\otimes A^{xy}_a\otimes B^{xy}_b}.
\end{split}\end{equation}

\subsection{Security of \texorpdfstring{\QPVBBetaf}{f-eta-BB84 QPV}  }

Our goal is to show that if the number of qubits $q$ that the attackers hold at the beginning of the protocol is linear, then, given that they do not respond `$\perp$', their probability of being correct is strictly less than the corresponding probability of the honest prover. To this end, we define a relaxation of the condition of being correct, and we consider $q$-qubit strategies which have a high chance that the verifiers record \textsc{Correct} at the end of the protocol. More specifically, we will define a set of quantum states that are `good' for a given fixed input, conditioned on actually playing. The first definition, which is an extension of Definition 4.1 in \cite{bluhm2022single}, considers single round attacks that are `good' for $l\leq 2^{2n}$ pairs of $x,y$. The reason to do so is that the attackers could be wrong for pairs that might be asked with exponentially small probability. 
\begin{definition} Let $\varepsilon\geq0$ 
and $l \in \mathbb N$. A  $q$-qubit strategy\[\mathbf S_{\eta}^{f-BB84}=\{\ket\varphi,U^x,V^y,\{A^{xy}_a\}_a,\{B^{xy}_b\}_b\}_{x,y}\] for \QPVBBetaf~is \emph{$(\varepsilon$,$l$)-perfect} if there exists a set $\mathcal L$ with $\abs{\mathcal L}\geq l$, such that  the attackers
\begin{enumerate}
    \item `respond' with probability $\eta$ on these pairs:
\begin{equation}
    \tr{\ketbra{\psi_{xy}}{\psi_{xy}} \mathbb{I}_V\otimes A^{xy}_{\perp}\otimes B^{xy}_{\perp}}=1-\eta \text{ }\forall(x,y)\in\mathcal L,\\
\end{equation}
\item and, up to $\varepsilon$, they are \textsc{Correct} at least with the same probability as an honest prover:
\begin{equation}\label{eq:l-good strategy}
\sum_{a\in \{0,1\}}\tr{\ketbra{\psi_{xy}}{\psi_{xy}} V^{f(x,y)}_a\otimes A^{xy}_a\otimes B^{xy}_a}\geq \eta\big((1-p_{err})-\varepsilon\big) \text{ }\forall(x,y)\in\mathcal L.
\end{equation}
\end{enumerate}
\end{definition}

\begin{definition}\label{def S_i^e}
Let Alice and Bob hold arbitrary registers $A$ and $B$, respectively, and let $\varepsilon \geq 0$. Let $V$ be a qubit register. For $i\in\{0,1\}$, we define the set $\mathcal{S}_i^{\varepsilon}$ as
\begin{equation}\label{eq def S_i^epsilon}
\begin{split}
    \mathcal{S}_i^{\varepsilon}:=\{&\ket{\psi}_{VAB} \mid \exists \textrm{ POVMs } \{A^{i}_a\}_a \textrm{ and } \{B^{i}_b\}_b \textrm{ acting on } A \text{ and }B, \\& \text{ respectively}\textrm{ such that } \eqref{eq condition S_i^epsilon}\textrm{ and }\eqref{eq attacker's response rate eta}\textrm{ are fulfilled}  \},
    \end{split}
\end{equation}
where
\begin{equation}\label{eq condition S_i^epsilon}
    \sum_{a\in \{0,1\}}\tr{\ketbra{\psi}{\psi} V^{i}_a\otimes A^{i}_a\otimes B^{i}_a}\geq \eta\big((1-p_{err})-\varepsilon\big), 
\end{equation}
and 
\begin{equation}\label{eq attacker's response rate eta}
   \tr{\ketbra{\psi}{\psi} \mathbb{I}_V\otimes A^{i}_{\perp}\otimes B^{i}_{\perp}}=1-\eta.
\end{equation}
\end{definition}

Given input $i\in\{0,1\}$ and a state $\ket{\phi}$, fulfilling responding $a=b\neq\perp$ with probability $\eta$ and $\perp$ with probability $1-\eta$ for every input, and never responding $a\neq b$, the maximum probability of being correct for such input is given by 

\begin{equation}
\label{eq prob 0 eta}
\begin{split}
    p_{\phi}^{i,\eta}:=
    \frac{1}{\eta} \max_{\substack{\{A_a^i,B_a^i\}_{a\in\{0,1,\perp\}}\\\textrm{with }\tr{\ketbra{\phi}{\phi} A_{a}^i B_{b}^i}=0, \forall a\neq b\in\{0,1,\perp\}, \\\textrm{and }\tr{\ketbra{\phi}{\phi} A_{\perp}^i B_{\perp}^i}=1-\eta}}    \sum_{a\in\{0,1\}}\tr{\ketbra{\phi}{\phi}V_a^i A_a^i  B_a^i},
    \end{split}
\end{equation}
where $(i)$ indicates `on the specific input $i$', `$\mid \ket{\phi}$' indicates that the probabilities we obtained with the state $\ket{\phi}$,  $V_a^i$ are as in \eqref{eq V_0^x, V_1^x} and $\{A_a^i,B_a^i\}_{a\in\{0,1,\perp\}}$ are POVMs.
As a consequence of the SDP \eqref{eq upperbound p_ans all constraints BB84}, considering that, from \eqref{eq pcorrec+pwrong if no no photon=1} and $\Pr[\textsc{V}_\textsc{AB}=\textsc{Wrong}\mid 
\textsc{V}_\textsc{AB}\neq \textsc{No photon}]\leq \Pr[\textsc{V}_\textsc{AB}=\textsc{Wrong}]\leq p_{err}$, we have that $p_{win}+p_{err}\leq1$, and thus we find that there exists a function $w:[0,1]\rightarrow[1-\cos^2(\frac{\pi}{8}),1]$ such that for all states $\ket{\phi}$, regardless of their dimension,  upper bounds the performance of the attackers:
\begin{equation}\label{eq ineq SDP fct f}
    \frac{1}{2}\left(p_{\phi}^{0,\eta}+p_{\phi}^{1,\eta}\right)\leq w(\eta).
\end{equation}

Moreover, consider the following relaxation of \eqref{eq prob 0 eta}, where the restrictions are such that the attackers respond with different answers with probability $\xi$ and  
and have a response rate in the interval $[(1-\eta)-\xi,(1-\eta)+\xi]$,
\begin{equation} 
\label{eq prob 0 eta with error}
    \Tilde{p}_{\phi}^{i,\eta,\xi}=\frac{1}{\eta}\max_{\substack{\{A_a^i,B_a^i\}_{a\in\{0,1,\perp\}}\\\textrm{with }\tr{\ketbra{\phi}{\phi} A_{a}^i B_{b}^i}\leq \xi, \forall a\neq b\in\{0,1,\perp\}, \\\textrm{and }(1-\eta)-\xi\leq\tr{\ketbra{\phi}{\phi} A_{\perp}^i B_{\perp}^i}\leq(1-\eta)+\xi}}  \sum_{a\in\{0,1\}}\tr{\ketbra{\phi}{\phi}V_a^i A_a^i  B_a^i}.
\end{equation}

On the other hand, let $\xi=0.001$, and consider the relaxation of \eqref{eq upperbound p_ans all constraints BB84} consisting of replacing \eqref{eq restriction error epsilon BB84} by ${\expectedbraket{A^x_aB^x_b}\leq\xi}$ $\forall a\neq b\in\{0,1,\perp\}$, $\forall x\in\{0,1\}$ and \eqref{eq perr constraint sum ab BB84} by $ \sum_{ab}(2-\norm{V_a^x+V_b^{x'}})\expectedbraket{A_a^{x}B_{b}^{x'}}\leq p_{err}(4\xi+\sum_{a}\expectedbraket{A_a^{x}B_{a}^{x}}+\expectedbraket{A_a^{x'}B_{a}^{x'}})+8\xi$, where the latter inequality is obtained analogously to \eqref{eq perr constraint sum ab BB84} by bounding the terms $\expectedbraket{A^x_aB^x_b}\leq\xi$, for all $a\neq b$. This implies that there exists a function $\Tilde{w}^{\xi}:[0,1]\rightarrow(1-\cos^2(\frac{\pi}{8})+\xi,1]$, obtained by the relaxation of the SDP (and allowing extra $\xi$ for the response rate), such that for  all states $\ket{\phi}$, regardless of their dimension, upper bounds the performance of the attackers who are allowed to respond different answers with probability $\xi$ and have a response rate in the interval $[(1-\eta)-\xi,(1-\eta)+\xi]$:
\begin{equation}\label{eq ineq SDP with errors}
    \frac{1}{2}\left(\Tilde{p}_{\phi}^{0,\eta,\xi}+\Tilde{p}_{\phi}^{1,\eta,\xi}\right)\leq \Tilde{w}^{\xi}(\eta),
\end{equation}
and $\Tilde{w}^{\xi}(\eta)$ is such that $w(\eta)\leq \Tilde{w}^{\xi}(\eta)$. This inequality is due to the fact that the latter is obtained by a relaxation of the constraints of the SDP of the former.

Due to the fact that $p_{win}+p_{err}\leq1$, the plot in Fig.~\ref{Fig SR bb84} can be represented in terms of the winning probability $p_{win}$, see Fig.~\ref{Fig pwin bounded bu w and Tilde w}. The plotted points in Fig.~\ref{Fig pwin bounded bu w and Tilde w} represent a numerical approximation of the functions $w(\eta)$ and $\Tilde{w}^{\xi}(\eta)$.  

\begin{figure}[h]
\label{Fig pwin bounded bu w and Tilde w}
\centering
\includegraphics[width=100mm]{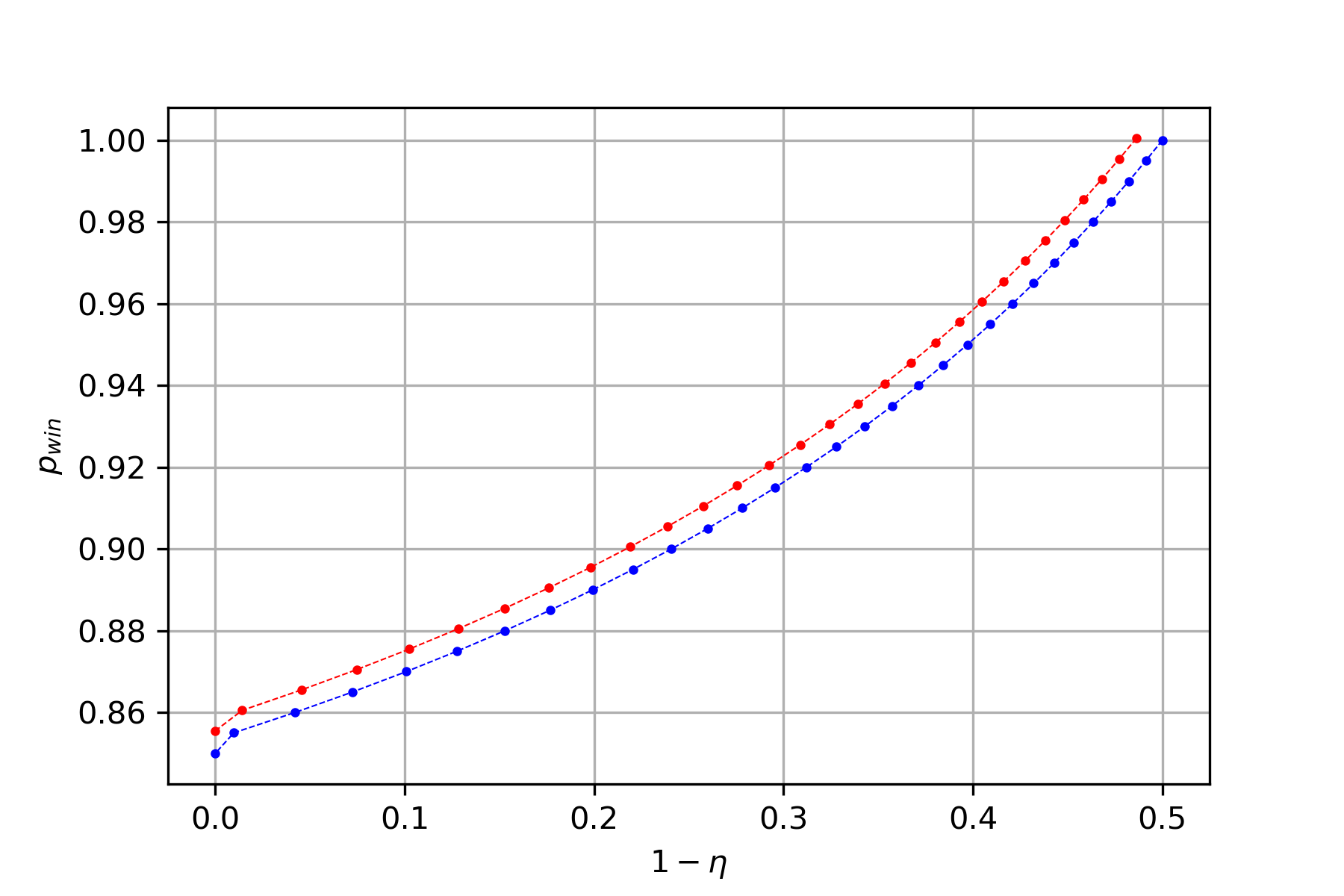}
\caption{Upper bounds of the winning probability given by \eqref{eq upperbound p_ans all constraints BB84} (blue dots), (equivalent representation of the blue pluses in Fig.~\ref{Fig SR bb84}), which corresponds to a numerical representation of the function $w(\eta)$. Red dots correspond to a numerical representation of  the function $\Tilde w^{\xi}(\eta)$, which is obtained by adding $\xi=0.001$ to the relaxation of \eqref{eq upperbound p_ans all constraints BB84} where the attackers are allowed to make errors with probability $\xi$. The continuous interpolation between values  is meant for a better viewing of the plot.}
\end{figure}

Now, we prove that the difference between the probabilities obtained by two quantum states projected into the same space is upper bounded by their trace distance. We then use this result to show that if two quantum states can be used to successfully attack around of the protocol with high probability with the POVMs  $\{A^{xy}_a\} \textrm{ and } \{B^{xy}_b\}$  for input $0$ and $1$ of a $q-$qubit strategy for \QPVBBetaf, respectively, these two states have to differ by at least a certain amount. These results are formalized in the next proposition and lemma.

\begin{prop}
Let  $\ket{\psi}$ and $\ket{\varphi}$ be two quantum states of (the same) arbitrary dimension, and let $\mathcal{D}(\ket{\psi},\ket{\varphi})$ denote their trace distance. Then, for every projector $\Pi$,
\begin{equation}\label{eq tr different states upper bounded by their distance}
    \abs{\tr{(\ketbra{\psi}{\psi}-\ketbra{\varphi}{\varphi})\Pi}}\leq  \mathcal{D}(\ket{\psi},\ket{\varphi}).
\end{equation}
\end{prop}

\begin{proof} There exist $Q$ and $S$ positive operators with orthogonal support \cite[Chapter 9]{NielsenChuang} such that 
\begin{equation}
    \ketbra{\psi}{\psi}-\ketbra{\varphi}{\varphi}=Q-S, \textrm{ and }  \mathcal{D}(\ket{\psi},\ket{\varphi})=\tr{Q}=\tr{S}.
\end{equation}
Then,
\begin{equation}
\begin{split}
    \tr{(\ketbra{\psi}{\psi}-\ketbra{\varphi}{\varphi}\Pi)}&=\tr{(Q-S)\Pi}= \tr{Q\Pi}-\tr{S\Pi}\leq \tr{Q\Pi}\\&\leq \tr{Q}\|\Pi\|= \mathcal{D}(\ket{\psi},\ket{\varphi}),
    \end{split}
\end{equation}
where we used that $S$ is positive definite and $\|\Pi\|=1$.
\end{proof}

\begin{lemma}\label{lemma D>=eta Delta}
Let $\ket{\psi}$ and $\ket{\varphi}$ be such that $p_{\psi}^{1,\eta}\geq \Tilde{w}^{\xi}(\eta)+\Delta$ and $p_{\varphi}^{0,\eta}\geq \Tilde{w}^{\xi}(\eta)+\Delta$, for some $\Delta>0$, which, due to Definition \ref{def S_i^e}, $\ket{\psi}\in\mathcal S_0^{\varepsilon}$ and $\ket{\varphi}\in\mathcal S_1^{\varepsilon}$, for $\varepsilon=1-(\Tilde{w}^{\xi}(\eta)+\Delta)$. 
Then,
\begin{equation}
    \mathcal{D}(\ket{\psi},\ket{\varphi})\geq \eta \Delta.
\end{equation}
\end{lemma}

Notice that the hypothesis of Lemma \ref{lemma D>=eta Delta} imply that $p_{\psi}^{1,\eta}, p_{\varphi}^{0,\eta}\geq \Tilde{w}^{\xi}(\eta)+\Delta>w(\eta)$ and thus these two states perform better in inputs $1$ and $0$, respectively, than any state would perform on average on both inputs. The greater is $\Delta$, the better they can perform. 
\begin{proof} Let $\xi=\mathcal{D}(\ket{\psi},\ket{\varphi})$ and $\psi=\ketbra{\psi}{\psi}$, $\varphi=\ketbra{\varphi}{\varphi}$. Subtracting and adding $\varphi$ to $\psi$ in Equation~\eqref{eq prob 0 eta} for $i=1$,
\begin{equation}\label{eq inequality p ptilde and delta}
\begin{split}
    \eta p_{\psi}^{1,\eta}&=\max_{\substack{\{A_a^1,B_a^1\}_{a\in\{0,1,\perp\}}\\\textrm{with }\tr{\psi A_{a}^1 B_{b}^1}=0, \forall a\neq b\in\{0,1,\perp\}, \\\textrm{and }\tr{\psi A_{\perp}^1 B_{\perp}^1}=1-\eta}}    \sum_{a\in\{0,1\}}\tr{(\psi-\varphi+\varphi)V_a^1 A_a^1  B_a^1}\\&\leq 2\xi+\max_{\substack{\{A_a^1,B_a^1\}_{a\in\{0,1,\perp\}}\\\textrm{with }\tr{\psi A_{a}^1 B_{b}^1}=0, \forall a\neq b\in\{0,1,\perp\}, \\\textrm{and }\tr{\psi A_{\perp}^1 B_{\perp}^1}=1-\eta}}    \sum_{a\in\{0,1\}}\tr{\varphi V_a^1 A_a^1  B_a^1}\\&\leq 2\xi+
    \max_{\substack{\{A_a^1,B_a^1\}_{a\in\{0,1,\perp\}}\\\textrm{with }\tr{\varphi A_{a}^1 B_{b}^1}\leq \xi, \forall a\neq b\in\{0,1,\perp\}, \\\textrm{and }(1-\eta)-\xi\leq\tr{\varphi A_{\perp}^1 B_{\perp}^1}\leq(1-\eta)+\xi}}  \sum_{a\in\{0,1\}}\tr{\varphi V_a^1 A_a^1  B_a^1}=2\xi +\eta \Tilde{p}_{\varphi}^{1,\eta,\xi},
    \end{split}
\end{equation}
where the first bound by $2\xi$ comes from \eqref{eq tr different states upper bounded by their distance} and we used that, because of \eqref{eq tr different states upper bounded by their distance}, the condition $\tr{\psi A_{a}^1 B_{b}^1}=0, \forall a\neq b\in\{0,1,\perp\}$ implies $\tr{\varphi A_{a}^1 B_{b}^1}\leq \xi, \forall a\neq b\in\{0,1,\perp\}$ and condition $\tr{\psi A_{\perp}^1 B_{\perp}^1}=1-\eta$ implies  $(1-\eta)-\xi\leq\tr{\varphi A_{\perp}^1 B_{\perp}^1}\leq(1-\eta)+\xi$.\\
Combining \eqref{eq inequality p ptilde and delta}, the hypothesis $p_{\psi}^{1,\eta}\geq \Tilde{w}^{\xi}(\eta)+\Delta$ and Equation \eqref{eq ineq SDP with errors}, we have
\begin{equation}\label{equation ineq p tilde}
    \Tilde{p}_{\varphi}^{0,\eta,\xi}\leq \Tilde{w}^{\xi}(\eta)-\Delta+\frac{2\xi}{\eta}.
\end{equation}
On the other hand, since $\Tilde{p}_{\varphi}^{0,\eta,\xi}$ is obtained by relaxing the restrictions of $p_{\varphi}^{0,\eta}$, we have that $\Tilde{p}_{\varphi}^{0,\eta,\xi}\geq p_{\varphi}^{0,\eta}$ and, by hypothesis,  $p_{\varphi}^{0,\eta}\geq \Tilde{w}^{\xi}(\eta)+\Delta$. These, together with \eqref{equation ineq p tilde}, lead to $\xi\geq \eta \Delta$.
\end{proof}

Notice that Lemma \ref{lemma D>=eta Delta} implies that Alice and Bob in some sense have to decide what strategy they follow before they communicate. Consequently, if the dimension of the state they share is small enough, a classical description of the first part of their strategy yields a compression of $f$.
The notion of the following definition captures this classical compression.

 \begin{definition}\cite{bluhm2022single}
 \label{def:classical-rounding}
Let $q$, $k$, $n \in \mathbb N$, $\varepsilon > 0$. Then, 
\begin{equation*}
g:\{0,1\}^{3k} \to \{0,1\}    
\end{equation*}
is an \emph{$(\varepsilon, q)$-classical rounding} of size $k$ if for all $f:\{0,1\}^{2n} \to \{0,1\}$, for all states $\ket{\psi}$ on $2q+1$ qubits, for all $l \in \{1, \ldots, 2^{2n}\}$ and for all $(\varepsilon, l)$-perfect $q$-qubit strategies for \QPVBBf, there are functions $f_A:\{0,1\}^n \to \{0,1\}^{k}$, $f_B:\{0,1\}^n \to \{0,1\}^{k}$ and $\lambda \in \{0,1\}^{k}$ such that $g (f_A(x), f_B(y), \lambda) =f(x,y)$ on at least $l$ pairs $(x,y)$.
\end{definition}

\begin{lemma} \label{lemma size delta-net} \cite{book_prob_banach_spaces}
Let $|||*|||$ be any norm on $\mathbb{R}^{n_0}$, for $n_0\in\mathbb{N}$. There is a $\delta$-net $S$ of the unit sphere of $(\mathbb{R}^{n_0},|||*|||)$ of cardinality at most $(1+2/\delta)^{n_0}$.
\end{lemma}

\begin{lemma}  Let $\Delta> 0$, and let $0\leq\varepsilon\leq\varepsilon_0$, where $\varepsilon_0$ is such that $\ket{\psi_i}\in\mathcal{S}_i^{\varepsilon}$ for $i\in\{0,1\}$ implies $\mathcal{D}(\ket{\psi_0},\ket{\psi_1})\geq\eta\Delta$. Then there is an $(\varepsilon,q)$-classical rounding of size $k=\log(\lceil\frac{4}{2^{\frac{2}{3}} (\eta \Delta +2)^{\frac{1}{3}}-2}\rceil)2^{2q+2}$.  
\end{lemma}

\begin{proof} Sketch (see Lemma \ref{lemma classical rounding size} for a detailed proof of the generalized version). Consider a $\delta$-net in Euclidean norm for the set of pure state on $2q+1$ qubits, where the net has cardinality at most $2^k$. Following the proof of Lemma \ref{lemma classical rounding size} analogously, we have that $\delta$ is such that
$
    3\delta+3\delta^2+\delta^3<\eta\Delta/2,
$
which holds for $\delta<(2+\eta\Delta)^{\frac{1}{3}}/2^{\frac{1}{3}}-1$. By Lemma \ref{lemma size delta-net}, we obtain the size $k=\log(\lceil\frac{4}{2^{\frac{2}{3}} (\eta \Delta +2)^{\frac{1}{3}}-2}\rceil)2^{2q+2}$. The remaining part of the proof is analogous to the proof of Lemma \ref{lemma classical rounding size}.
\end{proof}

\begin{lemma}
\label{lem:stable_E4}
Let $\Delta=0.013$, $\eta\in(0.53,1]$, $\varepsilon \in [0,1]$ $n$, $k$, $q \in \mathbb N$, $n \geq 10$. Moreover, fix an $(\varepsilon, q)$-classical rounding $g$ of size $k$ with  
   $k=\log(\lceil\frac{4}{2^{\frac{2}{3}} (\eta \Delta +2)^{\frac{1}{3}}-2}\rceil)2^{2q+2}$. Let $
q \leq \frac{1}{2}n-5$. Then, a uniformly random $f: \{0,1\}^{2n} \to \{0,1\}$ fulfills the following with probability at least $1 - 2^{-2^{n}}$:
 For any $f_A:\{0,1\}^n \to \{0,1\}^{k}$, $f_B:\{0,1\}^n \to \{0,1\}^{k}$, $\lambda \in \{0,1\}^{k}$, the equality $g(f_A(x), f_B(y), \lambda) = f(x,y)$ holds on less than $3/4$ of all pairs $(x,y)$.
\end{lemma}

\begin{proof}  Sketch (see below Lemma \ref{lem:stable_E4} for a detailed proof of the generalized version). We want to estimate the probability that for a randomly chosen $f$, we can find $f_A$ and $f_B$ such that the corresponding function $g$ is such that $\Pr_{x,y}[f(x,y)= g(f_A(x),f_B(y),\lambda)]\geq3/4$. In a similar manner as in \eqref{eq bounding probability of g}, we have that
\begin{equation}
     \Pr[f:\exists f_A,f_B,\lambda \textrm{ s.t. } \Pr_{x,y}[f(x,y) = g(f_A(x),f_B(y),\lambda)]]\leq
2^{ (2^{n+1}+1)k} 2^{2^{2n}h(1/4)} 2^{-2^{2n}},
\end{equation}
where $h$ denotes the binary entropy function. If $q\leq n/2-5$ and $k=\log(\lceil\frac{4}{2^{\frac{2}{3}} (\eta \Delta +2)^{\frac{1}{3}}-2}\rceil)2^{2q+2}$, with $\Delta=0.013$, for $\eta\in(0.53,1]$, the above expression is strictly upper bounded by $2^{-2^n}$.
\end{proof}

Lemma \ref{lem:stable_E4} shows that if the dimension of the initial state that Alice and Bob hold is small enough, any $(\varepsilon,3/4\cdot2^{2n})$-perfect $q$-qubit strategy needs a number of qubits which is linear in $n$. This leads to our first main theorem:

\begin{theorem}\label{theorem q>=n/2-5 for BB84 eta}
Consider the most general attack to a round of the \QPVBBetaf~protocol for a transmission rate ${\eta\in(0.53,1]}$. Let $\Delta=0.013$. If the attackers respond with probability $\eta$ and control at most $q$ qubits at the beginning of the protocol, and $q$ is such that
\begin{equation}
    q\leq \frac{n}{2}-5,
\end{equation}
then, for any $q$-qubit strategy for \QPVBBetaf, the probability that the attackers answer \textsc{Correct} is at most
\begin{equation}\label{eq bound attackers theorem BB84}
    \pr{\textsc{V}_\textsc{AB}=\textsc{c}} \leq\eta\big(1-\frac{1}{4}[1-(\Tilde{w}^{\epsilon}(\eta)+\Delta)]\big).
\end{equation}
\end{theorem}

The proof of Theorem~\ref{theorem q>=n/2-5 for BB84 eta} is a particular case of the proof of Theorem~\ref{theorem q>=n/2-5 for m-BB84 eta}.  See Fig.~\ref{figure bounds theorem bb84} for a representation of the bound \eqref{eq bound attackers theorem BB84} and see Fig.~\ref{Fig simplex} for the set of probabilities in $\Delta_2$ excluded by this theorem. Notice that if  $p_{err}<1-\frac{1}{4}[1-(\Tilde{w}^{\epsilon}(\eta)+\Delta)]$, then the probability that the attackers' answer \textsc{Correct} is stricly below the corresponding probability for an honest prover. 

\begin{figure}[h]
\centering
\includegraphics[width=105mm]{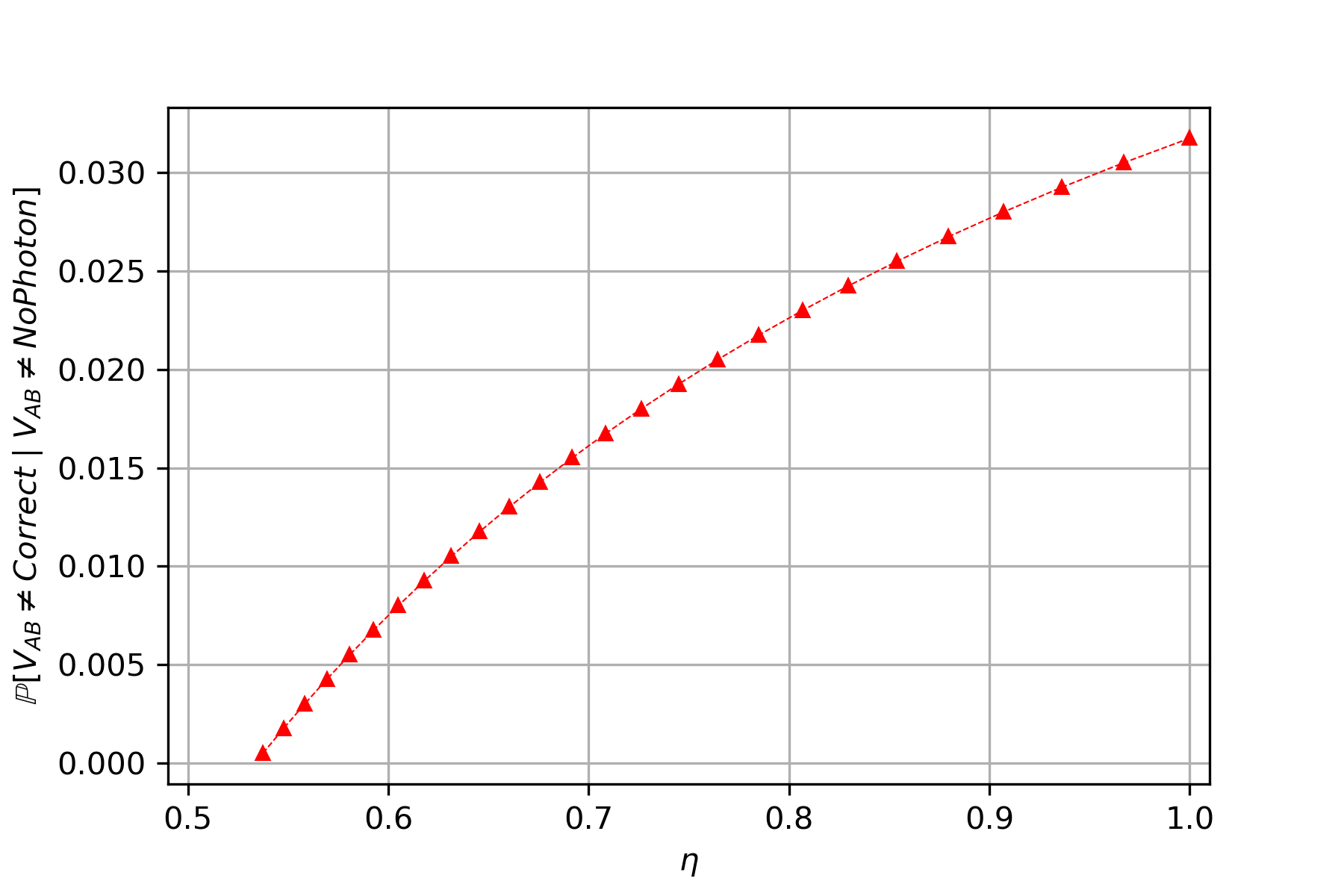}
\caption{Lower bounds (red triangles pointing up) given by Theorem \ref{theorem q>=n/2-5 for BB84 eta} on the probabilities that the attackers are not correct given that they answer.}
\label{figure bounds theorem bb84}
\end{figure}

\subsection{Sequential repetition of \texorpdfstring{\QPVBBetaf}{f,eta-BB84 QPV}}\label{sec:sequential rep f-lossy-BB84}

We will next show that, under the conditions of Theorem~\ref{theorem q>=n/2-5 for BB84 eta}, attackers will reproduce a score $\Gamma^{ent, att}_r$ that will fail the $\mathsf{T}^{r f-\mathrm{BB84}}_{\varepsilon_{\textrm{h}}}$ with exponentially high probability.  Let $\mathcal{A}\subseteq \Delta_2$ be the set of probabilities $\textbf{q}=(q_{\textsc{c}},q_{\perp},q_{\textsc{i}})$ that the attackers can reproduce, where $q_{\textsc{i}}:=q_{\textsc{w}}+q_{\lightning}$. Notice that since for $\textbf{q}_1,\textbf{q}_2\in\mathcal{A}$, the attackers can play any convex combination of them, the set $\mathcal{A}$ is convex. Let $\Delta=0.013$ and, for simplicity, denote $\mathfrak{b}_{\eta}:=\frac{1}{4}\big(1-(\Tilde{w}^{\epsilon}(\eta)+\Delta)$. Let $\mathcal{R}\subset\Delta_2$ be the region of the probability simplex defined as $\mathcal{R}:=\cup_{\eta\in(0.53,1]}\mathcal{R}_{\eta}\subset \Delta_2$,  where the sets $\mathcal{R}_{\eta}$ are defined as
\begin{equation}
\begin{split}
      \mathcal{R}_{\eta}=\big\{ (p_\textsc{c},p_{\perp},p_{\textsc{i}})\in \Delta_2 \mid 
    p_\textsc{c}&> \max\{1-p_{\perp}+\frac{\mathfrak{b}_{\eta}(1-2p_{\perp})}{1-2(\eta-\epsilon)},\frac{\eta-\epsilon-\cos^2(\frac{\pi}{8})-\mathfrak{b}_{\eta}}{1-(\eta-\epsilon)}p_{\perp}+\cos^2(\frac{\pi}{8})\}
    \\ \wedge p_{\textsc{i}}&< \min\{\frac{\mathfrak{b}_{\eta}(1-2p_{\perp})}{2(\eta-\epsilon)-1},
    \frac{\mathfrak{b}_{\eta}-\sin^2(\frac{\pi}{8})}{1-(\eta-\epsilon)}p_{\perp}+\sin^2(\frac{\pi}{8})\}
    \big \}.
\end{split}
\end{equation}

\begin{lemma}\label{lemma probability simplex} Under the conditions of  Theorem~\ref{theorem q>=n/2-5 for BB84 eta}, 
any $\textbf{q}$ that the attackers can reproduce is such that $\textbf{q}\notin\mathcal{R}$, i.e.\  $\mathcal{R}\cap\mathcal{A}=\emptyset$, and thus
\begin{equation}
    \mathcal{A}\subsetneq \mathcal{R}^c.
\end{equation}
\end{lemma}

\begin{proof}
Let $\textbf{q}=(q_{\textsc{c}},q_{\perp},q_{\textsc{i}})$ be the probability distribution corresponding to an attack from Alice and Bob and define $\mathcal{A}$ as the set of all probabilities  If they are allowed to answer  with probability $\eta\pm\epsilon$, then
\begin{equation}
    \textbf{q}=(q_{\textsc{c}},1-(\eta +\mu),q_{\textsc{i}}), \mu\in[-\epsilon,\epsilon].
\end{equation}
Moreover, by Theorem~\ref{theorem q>=n/2-5 for BB84 eta},
if they control at most $q\leq \frac{n}{2}-5$ qubits at the beginning of the protocol, 
\begin{equation}\label{eq qI >= g(eta)}
    q_{\textsc{i}}\geq \mathfrak{b}_{\eta}, \forall\mu\in[-\epsilon,\epsilon].
\end{equation}
We define the following region 
\begin{equation}
    \mathcal{R}^0_{\eta}=\{(p_\textsc{c},p_{\perp},p_{\textsc{i}})\in \Delta_2 \mid p_{\perp}\in[1-(\eta-\epsilon),1-(\eta+\epsilon)] \wedge p_{\textsc{i}}< \mathfrak{b}_{\eta}\},
\end{equation}
therefore, by \eqref{eq qI >= g(eta)}, $ \textbf{q}\notin\mathcal{R}^0_{\eta}$, and thus
\begin{equation}\label{eq R0 cap A = emptyset}
    \mathcal{R}^0_{\eta}\cap\mathcal{A}=\emptyset.
\end{equation}
Notice that the strategies  $\mathbf S_{\textsc{tfkw}}$ and  $\mathbf S_{guess}$, reproduce  $ \textbf{q}_{\textsc{tfkw}}=(\cos^2(\frac{\pi}{8}),0,1-\cos^2(\frac{\pi}{8}))$ and $ \textbf{q}_{guess}=(\frac{1}{2},\frac{1}{2},0)$ and, respectively, and therefore $\textbf{q}_{\textsc{tfkw}},\textbf{q}_{guess}\in\mathcal{A}$. The attackers can also do the strategies consisting on always responding $\perp$ or always being incorrect (e.g.\ answering inconsistent answers $a\neq b$), which reproduce the probabilities $\textbf{q}=(0,1,0)$ and $\textbf{q}=(0,0,1)$, respectively. Therefore, the convex hull of  these four points is a convex subset of $\mathcal{A}$. 
Consider the straight line $s_1$ defined by the two points $\textbf{q}_{guess}$ and $(\eta-\epsilon-\mathfrak{b}_{\eta},1-(\eta-\epsilon),\mathfrak{b}_{\eta})\in\partial \mathcal{R}^0_{\eta}$, where $\partial$ denotes the boundary, which is given by 
\begin{equation}
    s1: \frac{x-\frac{1}{2}}{\eta-\epsilon-\frac{1}{2}-\mathfrak{b}_{\eta}}=\frac{y-\frac{1}{2}}{\frac{1}{2}-(\eta-\epsilon)}=\frac{z}{\mathfrak{b}_{\eta}}.
\end{equation}
Consider the set $\mathcal{R}^1_{\eta}$, defined as the points in $\Delta_2$ that do not belong to $\mathcal{R}^0_{\eta}$ that are 'below' the straight line $s_1$ and are such that $p_{\perp}\leq 1-(\eta+\epsilon)$, formally defined as
\begin{equation}
\begin{split}
      \mathcal{R}^1_{\eta}=\big\{ (p_\textsc{c},p_{\perp},p_{\textsc{i}})\in \Delta_2 \mid 
   p_\textsc{c}>1-p_{\perp}+\frac{\mathfrak{b}_{\eta}(1-2p_{\perp})}{1-2(\eta-\epsilon)}
    \wedge p_{\textsc{i}}< \frac{\mathfrak{b}_{\eta}(1-2p_{\perp})}{2(\eta-\epsilon)-1}
    \big \}\setminus \mathcal{R}^0_{\eta}.
\end{split}
\end{equation}
We will show by contradiction that $\mathcal{R}^1_{\eta}\cap\mathcal{A}=\emptyset$. Assume that exists  $\textbf{p}_{R}\in\mathcal{R}^1_{\eta}\cap\mathcal{A}$. Since $\mathcal{A}$ is a convex set and $\textbf{q}_{guess}\in\mathcal{A}$, $t \textbf{p}_{R}+(1-t)\textbf{q}_{guess}\in\mathcal{A}$ for all $t\in[0,1]$. By construction of $\mathcal{R}^1_{\eta}$, $\{t \textbf{p}_{R}+(1-t)\textbf{q}_{guess}\mid t\in[0,1]\}\cap\mathcal{R}^0_{\eta}\neq\emptyset$, i.e.\  the straight line connecting $\textbf{p}_{R}$ and $\textbf{q}_{guess}$ intersects $\mathcal{R}^0_{\eta}$. Then, $\exists t_0\in[0,1]$ such that $\textbf{p}_{t_0}:=t_0 \textbf{p}_{R}+(1-t_0)\textbf{q}_{guess}\in\mathcal{R}^0_{\eta}\cap\mathcal{A}$. However, by \eqref{eq R0 cap A = emptyset}, $\mathcal{R}^0_{\eta}\cap\mathcal{A}=\emptyset$. Therefore, $\nexists\textbf{p}_{R}\in\mathcal{R}^1_{\eta}\cap\mathcal{A}$.\\
Similarly, consider the set $\mathcal{R}^2_{\eta}$ consisting of the points in $\Delta_2$ that are such that $p_{\perp}\geq1-(\eta+\epsilon)$ and that are 'below' the straight line given by the points $\textbf{q}_{\textsc{tfkw}}$ and $(\eta-\epsilon-\mathfrak{b}_{\eta},1-(\eta-\epsilon),\mathfrak{b}_{\eta})$, formally defined as
\begin{equation*}
\begin{split}
      \mathcal{R}^2_{\eta}=\big\{(p_\textsc{c},p_{\perp},p_{\textsc{i}})\in \Delta_2 \mid 
    p_\textsc{c}>\frac{\eta-\epsilon-\cos^2(\frac{\pi}{8})-\mathfrak{b}_{\eta}}{1-(\eta-\epsilon)}p_{\perp}+\cos^2(\frac{\pi}{8})\}
   \wedge p_{\textsc{i}}<
    \frac{\mathfrak{b}_{\eta}-\sin^2(\frac{\pi}{8})}{1-(\eta-\epsilon)}p_{\perp}+\sin^2(\frac{\pi}{8})
    \big\}
\end{split}
\end{equation*}
By an analogous convexity argument with the point $\textbf{q}_{\textsc{tfkw}}$, we have that $\mathcal{R}^2_{\eta}\cap\mathcal{A}=\emptyset$.\\ Notice that $\mathcal{R}_{\eta}=\mathcal{R}^0_{\eta}\cup\mathcal{R}^1_{\eta}\cup\mathcal{R}^2_{\eta}$, and therefore, 
$  \mathcal{R}_{\eta}\cap\mathcal{A}=\emptyset$ $\forall \eta\in(0.53,1]$. As a consequence, $ \mathcal{R}\cap\mathcal{A}=\big(\cup_{\eta\in(0.53,1]}\mathcal{R}_{\eta}\big)\cap\mathcal{A}=\emptyset$ and therefore, $\mathcal{A}\subset\mathcal{R}^c$.  
\end{proof}

\begin{theorem}\label{theorem seq repetition entangled} Let $\varepsilon_{\textrm{h}}>0$, and $\eta$ and $p_{err}$ be such that $\alpha^{ent}>0$. Then, any sequential strategy that attackers who pre-share $q\leq n/2-5$ qubits to break \QPVBBetaf~ is such that $\mathbb E [\Gamma_r^{ent,att}]\leq 0$ and, moreover, the probability that they are accepted in the $\mathsf{T}^{r f-\mathrm{BB84}}_{\varepsilon_{\textrm{h}}}$ test is exponentially small:
\begin{equation}\label{eq thm seq repetition entanglement}
    \mathbb P[\Gamma_r^{ent,att}\geq r(\alpha^{ent}-\delta)]\leq e^{-r (\alpha^{ent}-\delta)^2/2}.
\end{equation}
\end{theorem}

Theorem~\ref{theorem seq repetition entangled} shows that there is a way (statistical test) to distinguish an honest prover from attackers with exponentially high probability \, i.e.\  that after $r$ rounds the attackers will be \emph{caught}. The values for which $\alpha^{ent}>0$ correspond to the points below the straight line in Fig.~\ref{Fig simplex}---the above $\alpha^{ent}-\delta>0$ corresponds to a shift, where $\delta$ can be made small by increasing the number of repetitions.

\begin{proof}
    Let $\textbf{p}_1=(959/1000,0,41/1000)$ and $\textbf{p}_2=(27/50,23/50,0)$. The line segment  given by $\textbf{p}_1$ and $\textbf{p}_2$ is a subset of $\mathcal{R}$ and it can be described as the intersection of $\Delta_2$ and the plane $\gamma_{\textsc{c}}p_{\textsc{c}}-\gamma_{\perp}p_{\perp}-\gamma_{\textsc{i}}p_{\textsc{i}}=0$, with the normal vector $\boldsymbol{\gamma}=(\gamma_{\textsc{c}},\gamma_{\perp},\gamma_{\textsc{i}})$, where $\gamma_{\textsc{c}}=\frac{943}{\sqrt{488625947}}\simeq 0.04266$, $\gamma_{\perp}=\frac{1107}{\sqrt{488625947}}\simeq 0.050079$ and $\gamma_{\textsc{c}}=\frac{22057}{\sqrt{488625947}}\simeq 0.99783$. Consider the follow non-empty partition of the simplex $\Delta_2$: 
    \begin{equation}
        \Delta_2^{+}:=\Delta_2\cap\{\textbf{p}\in\Delta_2 \mid \textbf{p} \cdot \boldsymbol{\gamma}>0\}, \hspace{3mm}\textrm{ and }\hspace{3mm} \Delta_2^{-}:=\Delta_2\cap\{\textbf{p}\in\Delta_2 \mid \textbf{p} \cdot \boldsymbol{\gamma}\leq 0\},
    \end{equation}
which is such that $\Delta_2=\Delta_2^{+}\mathbin{\dot{\bigcup}}\Delta_2^{-}$. In particular, $\Delta_2^{+}\subset\mathcal{R}$ and $\mathcal{A}\subset\Delta_2^{-}$. Then we have that for all $\textbf{q}\in\mathcal{A}$,
   \begin{equation}
       q_{\textsc{c}}\gamma_{\textsc{c}}-q_{\perp}\gamma_{\perp}-q_{\textsc{i}}\gamma_{\textsc{i}}\leq0. 
    \end{equation}
The above amount corresponds to the expected value of $T^{ent}$ for a strategy $\textbf{q}\in\mathcal{A}$, and thus we have that for all rounds $i$, the expected value of the attackers ($att$) is such that
\begin{equation}
    \mathbb E_{att} [T^{ent}_i]\leq 0.
\end{equation}
Define $\Gamma^{ent,att}_0=0$. The process $\Gamma=(\Gamma^{ent,att}_r:r\geq 0)$ is a supermartingale relative to the filtration $\mathcal{F}_r$, where $\mathcal{F}_r= \sigma(T_1^{ent},...,T_{r}^{ent})$, and $\sigma$ denotes the $\sigma$-algebra. In fact, 
\begin{equation}
    \mathbb E[\Gamma_r^{att}\mid \mathcal{F}_{r-1}]= \mathbb E[T_{r}^{ent}\mid  \mathcal{F}_{r-1}]+\mathbb E[\Gamma_{r-1}^{att}\mid  \mathcal{F}_{r-1}]\leq \Gamma_{r-1}^{att},
\end{equation}
which is the definition of a supermartingale. The first equality is due to the linearity of the conditional expectation and the inequality is due to the fact that $\mathbb E[T_{r}^{ent}\mid  \mathcal{F}_{r-1}]$ is independent of $\mathcal{F}_{r-1}$ and thus $\mathbb E[T_{r}^{ent}\mid  \mathcal{F}_{r-1}]=\mathbb E_{att} [T_{r}^{ent}]\leq 0$ and $\Gamma_{r-1}^{att}$ is $\mathcal{F}_{r-1}$-measurable.

An immediate application of Azuma's inequality \cite{Azuma1967}
leads to \eqref{eq thm seq repetition entanglement}.\\
The proof of Theorem~\ref{thm sequential repetition no-entanglement} 
is analogous to the above case considering the segment line given by the two points $\textbf{q}_{\textsc{tkfw}}=(\cos^2\frac{\pi}{8},0,\sin^2\frac{\pi}{8},0)$ and $\textbf{q}_{guess}=(\frac{1}{2},\frac{1}{2},0,0)$ and the corresponding partition is such that $\Delta_2^-=\mathcal{A}$ and $\Delta_2^+=\Delta_2\setminus\mathcal{A}$. 
    
\end{proof}

\section{The \texorpdfstring{\QPVBBetam}{QPV-m-eta} protocol}\label{section m basis no entanglement}

In Section \ref{section 2 basis no entanglement} we showed security of  \QPVBBeta~ protocol for unentangled attackers. Nevertheless, the protocol was shown to be secure only for transmission rate $\eta>\frac{1}{2}$, which is still very hard for current technology to achieve.
For this reason, we propose a protocol which generalizes \QPVBBeta~ to more basis settings, for which we can apply similar techniques to prove security in the lossy case. In this section, we generalize the results of Section \ref{section 2 basis and entanglement}, showing security for non-entangled attackers and reaching arbitrary big photon loss.

Independently and around the same time, Buhrman, Schaffner, Speelman, Zbinden \cite[Chapter 5]{FlorianThesis} and Qi and Siopsis \cite{OtherExtentionBB84Qi_Siopsis2015} introduced extensions of the $\mathrm{QPV_{BB84} }$ protocol.
Both are based on allowing the verifiers to choose among more than two different qubit bases, which for the \QPVBB~ protocol corresponded to the computational and the Hadamard basis. The protocol in \cite{FlorianThesis} allows $V_0$ choosing among $m$, for an arbitrary $m\geq 2$, different orthonormal bases in the meridian $\varphi=0$ of the Bloch sphere depending on the angle $\theta\in[0,\pi)$, where these are uniformly distributed, i.e.\ $\theta\in\{\frac{0}{m}\pi,...,\frac{m-1}{m}\pi\}$, and the bases are $\{\ket{0_{\theta}},\ket{1_{\theta}}\}$, where 
\begin{equation}
\label{eq extention BB84 just theta}
\begin{split}
    &\ket{0_{\theta}}:=\cos\frac{\theta}{2}\ket{0}+\sin\frac{\theta}{2}\ket{1},\hspace{1cm}\ket{1_{\theta}}:=\sin\frac{\theta}{2}\ket{0}-\cos\frac{\theta}{2}\ket{1}.
\end{split}
\end{equation}
Recall that the $\mathrm{QPV_{BB84} }$ protocol is recovered taking $m=2$, where $\theta=0$ and $\theta=\frac{\pi}{2}$ correspond to the computational and Hadamard basis, respectively. On the other hand, the extension in \cite{OtherExtentionBB84Qi_Siopsis2015} allows the verifiers to choose among $m$ encoding bases over the whole Bloch sphere, however such an extension only works for $m$ large enough and not all large integers are allowed. 

Here we present a similar extension allowing to choose $m$ random bases over the Bloch sphere for all $m\geq 2$, which works regardless whether $m$ is large or small, and we prove that this translates to better security in case terms of the loss-tolerance of the quantum information in an experimental implementation.
We add the $\varphi$ parameter corresponding to the azimuth angle to the states in \eqref{eq extention BB84 just theta} as a phase $e^{i\varphi}$ in front of $\ket{1}$, in a similar way as in \cite{OtherExtentionBB84Qi_Siopsis2015}.
We do however use a slightly different procedure than \cite{OtherExtentionBB84Qi_Siopsis2015} to compute the precise angles, to make the basis choice more uniform (see below).

\subsection{Discrete uniform choice of basis over the Bloch sphere}\label{section discrete basis chosing}
In order to avoid accumulation of points in the sphere around the poles due to the unit sphere area element $d\Omega=\sin\theta d\theta d\varphi$, a continuous uniform distribution of points can be made by taking  \cite{weissteinDiscretitzationSphere} 
\begin{equation}
    \begin{split}
    &\theta= \cos^{-1}(2u-1), \hspace{1cm} \varphi= 2 \pi v,
    \end{split}
\end{equation}
where $u$ and $v$ are uniformly distributed over the interval $(0,1)$. 
Notice that allowing $\varphi\in[0,2\pi)$ would imply to have duplicate bases (i.e.\ the same basis vectors in different order), thus, $\varphi$ will be restricted to take values in the range $[0,\pi)$. Moreover, in the discrete case we are interested also in the north pole of the sphere ($\theta=0$), corresponding to the computational basis, and therefore in order to include it, discretizing the sphere with $m_{\theta}$ different $\theta$ and with $m_{\varphi}$ different $\varphi$, the 0 must be included in the range of $u$, i.e.\  $u\in\{\frac{0}{m_{\theta}},...,\frac{m_{\theta}-1}{m_{\theta}}\}$. Similarly, in order to have the Hadamard basis (and the bases in between them in the meridian $\varphi=0$), $v\in\{\frac{0}{m_{\varphi}},...,\frac{m_{\varphi}-1}{m_{\varphi}}\}$. Let $\Tilde{u}:=m_{\theta}u$ and $\Tilde{v}:=m_{\varphi}v$, which determine the $m_{\theta}m_{\varphi}$ points of the discretization.  
Let $x:=\Tilde{u}\Tilde{v}\equiv m_{\varphi}\Tilde{u}+\Tilde{v}$, therefore, given $x\in[m_{\theta}m_{\varphi}]$, one can recover $\Tilde{u}=\floor{x/m_{\varphi}}$ and $\Tilde{v}=x \bmod m_{\varphi}$, where $\floor{*}$ stands for the floor function. Notice that this discrete parametrization has $m_{\varphi}$ degenerate points for $\Tilde{u}=m_{\theta}-1$, corresponding to $(\theta=0,\varphi)$, which can be easily removed by $x$ taking values in the range $\{0,...,m-1\}=:[m]$, where $m:=m_{\varphi}(m_{\theta}-1)+1$. Therefore, we can discretize the bases in the Bloch sphere depending on $x$ so that $\forall x \in [m]$,

\begin{equation}
\begin{split}
    \theta(x)=\arccos{(\frac{2}{m_{\theta}}(\floor{\frac{x}{m_{\varphi}}}+1)-1)},\hspace{1cm}\varphi(x)=\pi \frac{x\bmod m_{\varphi}}{m_{\varphi}}.
\end{split}
\end{equation}
Then the protocol is extended allowing the verifiers to choose among the bases $\{\ket{0_{x}},\ket{1_{x}}\}$ for $x \in [m]$, where
\begin{equation}
\label{eq extention BB84 Bloch sphere}
\begin{split}
    &\ket{0_{x}}:=\cos\frac{\theta(x)}{2}\ket{0}+e^{i\varphi(x)}\sin\frac{\theta(x)}{2}\ket{1},\hspace{1cm}
    \ket{1_{x}}:=\sin\frac{\theta(x)}{2}\ket{0}-e^{i\varphi(x)}\cos\frac{\theta(x)}{2}\ket{1}.
\end{split}
\end{equation}
Note that for any $m$, we can discretize the bases in as many ways as divisors $m-1$ has in the following way: one chooses the number of $\theta$ and $\varphi$ as $(m_{\theta},m_{\varphi})=(d_m+1,\frac{m-1}{d_m})$, for each $d_m$ divisor of $m-1$. See Fig. \ref{Fig discretization of basis} for the representation of the $\ket{0_{x}}$ for all $x\in[m]$ in the Bloch sphere for different choices of $m_{\theta}$ and $m_{\varphi}$. The $\ket{1_{x}}$ corresponds to the diametrically opposite point, and therefore representing the kets $\ket{0_{x}}$ determines the $m$ orthonormal bases, e.g.\ the computational basis is associated with the north pole. As examples, the choice  $(m,m_{\theta},m_{\varphi})=(2,2,1)$ corresponds to the computational and Hadamard bases, and the choice $(m,m_{\theta},m_{\varphi})=(3,2,2)$ to  the computational, Hadamard bases and the basis formed by the eigenvectors of the Pauli $Y$ matrix. 

 \begin{figure}[h]
\label{Fig discretization of basis}
\centering
\subfigure[]{\includegraphics[width=50mm]{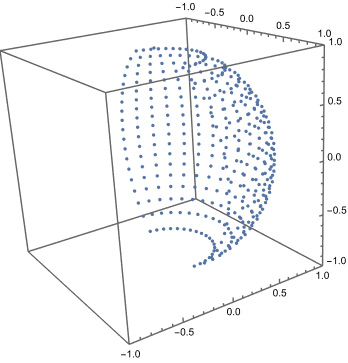}}
\subfigure[]{\includegraphics[width=50mm]{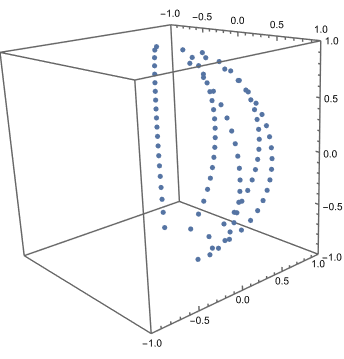}}
\caption{Discretization of $m$ bases in the Bloch sphere, where the points on the surface of the sphere represent the ket $\ket{0_x}$ of the orthonormal basis $\{\ket{0_x},\ket{1_x}\}$ for $x\in[m]$.  (a) $m_{\theta}=m_{\varphi}=20$, (b) $m_{\theta}=10,m_{\varphi}=5$. }
\end{figure}

Based on Section \ref{section discrete basis chosing}, we introduce an extension of the \QPVBBeta~ protocol, which we denote by \QPVBBetam, where $m_{\theta\varphi}$ is the sort notation of $(m,m_{\theta}, m_{\varphi})$.

\begin{definition}\label{def k bases protocol}
Let $\eta$ denote the transmission rate of the qubits sent from $V_0$ to \prover{}, and let $m_{\theta\varphi}$ be as described above. We define one round of the \QPVBBetam~protocol as follows:
\begin{enumerate}
    \item $V_0$ and $V_1$ secretly agree on $v\in\{0,1\}$ and $x\in[k]$. Then, $V_0$ prepares the qubit state $\ket{v_x}$ as in \eqref{eq extention BB84 Bloch sphere}. 
    \item $V_0$ sends the qubit $\ket{v_x}$ to \prover{}, and $V_1$ sends $x$ to \prover{} coordinating their times so that they arrive to \prover{} at the same time.
    \item Immediately, \prover{} measures the received qubit in the basis $\{\ket{0_x},\ket{1_x}\}$ and broadcasts her outcome  to both verifiers. If, due to experimental losses, the qubit did not arrive, \prover{} broadcasts `\textsc{No photon}' with the symbol~$\perp$. Denote the response by $v_{\textsf{P}}$.
    \item  If 
    \begin{enumerate} \item $V_0$ and $V_1$ receive their respective answers at the time corresponding with the claimed location, and they are equal, i.e.\  both receive the same $v_{\textsf{P}}$, then, if
        \begin{itemize}
            \item $v_{\textsf{P}}=v$, the verifiers record `\textsc{Correct}', denoted by `\textsc{c}', 
            \item $v_{\textsf{P}}=1-v$, the verifiers record `\textsc{Wrong}', denoted by `\textsc{w}', 
            \item $v_{\textsf{P}}=\perp$, the verifiers record `\textsc{No photon}',  denoted by `$\perp$', 
        \end{itemize}
        \item         otherwise, they record `\textsc{Abort}',  denoted by `$\lightning$', and abort the protocol rejecting the location.
    \end{enumerate}
\end{enumerate}
\end{definition}
At the end of the protocol, when sequentially run $r$ times, the verifiers check that the answers satisfy \eqref{eq succesful protocol}. Notice that \QPVBBeta~is recovered for $m=2$, with the unique choice of $m_{\theta}=2$ and $m_{\varphi}=1$. Analogous to \QPVBBeta, for each choice of $m_{\theta\varphi}$, from the analysis of the correlations shown below in Section~\ref{section m basis no entanglement}, one can construct a decision test to make the binary decision to either \emph{accept} or \emph{reject}, see proof of Theorem~\ref{theorem seq repetition entangled} for how to construct it. 

\subsection{Security of the  \texorpdfstring{\QPVBBetam}{QPV-m-eta} in the No-PE model}

Analogous to \QPVBBeta{} in Section~\ref{section 2 basis no entanglement}, for the security analysis, we will consider the purified version of \QPVBBetam{}, where, in Definition~\ref{def k bases protocol}, in step~1, $V_0$ instead of preparing the qubit $\ket{0_x}$, prepares the EPR pair $\ket{\Phi^+}$, sends one qubit register to $P$, and keeps the other register. In a later moment, $V_0$ measures his local register in the basis $\{\ket{0_x},\ket{1_x}\}$. 
In an analogy to the \QPVBBeta~protocol, an attack on the \QPVBBetam~protocol can be associated with a MoE (monogamy-of-entanglement) game in the following way. Let $V$ be the register of the qubit of the verifier,  with associated Hilbert space $\mathcal{H}_V=\mathbb{C}^2$, with $\mathcal{X}=[m]$ and $\mathcal{V}=\{0,1\}$. The verifiers perform the collection of measurements
\begin{equation}
    \{V_0^{x},V_1^{x}\}_{x\in[m]},
\end{equation}
where $V_0^{x}=\ketbra{0_x}{0_x}$ and $V_1^{x}=\ketbra{1_{x}}{1_{x}}$. The two collaborating parties in the MoE game correspond to the attackers who want to break the protocol with their guess.
Then, the attackers Alice and Bob, with associated Hilbert spaces $\mathcal{H}_A$ and $\mathcal{H}_B$, respectively, have to win against the verifiers $V$ both giving the same outcome to $V$ or declare photon loss. 
Thus, having a strategy to attack the protocol implies having an strategy for a MoE game. An extension of a strategy for a lossy MoE game, see Section \ref{section 2 basis no entanglement}, naturally generalizes as $\mathbf{S}^{\eta}_{MoE}=\{\ket{\psi},A_{a}^{x},B_{a}^{x}\}_{a\in\{0,1,\perp\},x\in[m]}$.

In \cite{TomamichelMonogamyGame2013} the following upper bound to win a general MoE game is given:
\begin{equation}
    \label{eq upperbound Tomamichel}
    p_{win}\leq \frac{1}{|\mathcal{X}|}+\frac{|\mathcal{X}|-1}{|\mathcal{X}|}\sqrt{\max_{x\neq x'\in\mathcal{X}}\max_{a,a'\in\mathcal{A}}\norm{\sqrt{V_a^x}\sqrt{V_{a'}^{x'}}}^2}.
\end{equation}
The security analysis of the  \QPVBBetam~protocol will be based, in the same way as the \QPVBBeta\\ protocol, on maximizing the probability that the attackers `play' without being caught:
\begin{equation}
\label{eq pans}
    p_{ans}=\frac{1}{m}\sum_{x\in[m],a\in \{0,1\}}\bra{\psi}A_a^{x} B_a^{x}\ket{\psi}. 
\end{equation}
As in the \QPVBBeta~ protocol, the constraints will be the linear constraints implied by $\mathbf{S}^{\eta}_{MoE}\in\mathcal{Q}_{\ell}$, the analogous to \eqref{eq restriction error epsilon BB84}, i.e.,
\begin{equation}\label{eq restriction error epsilon}\expectedbraket{A_a^{x}B_{b}^{x}}=0 \hspace{5mm}\forall a\neq b\in \{0,1,\perp\}, \forall x \in [m],\end{equation}
and the inequalities given in Proposition \ref{prop ineq perr for QPV} bounded by $p_{err}$.

\begin{prop} \label{prop ineq perr for QPV} Let $a,b\in\{0,1\}$, $\alpha_i^a=\braket{i_{x'}}{a_{x}}$ and $\beta_i^b=\braket{i_{x}}{b_{x'}}$ for $i\in\{0,1\}$. The terms $\expectedbraket{A_a^{x}B_{b}^{x'}}$ can be bounded by $p_{err}$ by the two inequalities below:
\begin{equation}
    \label{eq perr constraint sum ab}
    \sum_{ab}(2-\norm{V_a^x+V_b^{x'}})\expectedbraket{A_a^{x}B_{b}^{x'}}\leq p_{err}\sum_{a}(\expectedbraket{A_a^{x}B_{a}^{x}}+\expectedbraket{A_a^{x'}B_{a}^{x'}}),
\end{equation}
\begin{equation}
\label{eq constraint perr 2 simplified}
    \begin{split}
        &\sum_{a,b}\big(4-\norm{(1+\abs{\beta_a^b }^2)V_a^{x}+
        (1+\abs{\alpha_b^a }^2)V_b^{x'}+\beta_0^b\beta_1^{b*}\ketbra{0_{x}}{1_{x}}+\beta_0^{b*}\beta_1^{b}\ketbra{1_{x}}{0_{x}}+
        \alpha_0^a\alpha_1^{a*}\ketbra{0_{x'}}{1_{x'}}\\&+\alpha_0^{a*}\alpha_1^{a}\ketbra{1_{x'}}{0_{x'}}}\big)\expectedbraket{A_a^{x}B_{b}^{x'}}
        \leq
        p_{err}\big((2+\max_{i,j}\abs{\beta_i^j}^2)\sum_a\expectedbraket{A_a^xB_a^x}+(2+\max_{i,j}\abs{\alpha_i^j}^2)\sum_a\expectedbraket{A_a^{x'}B_a^{x'}}\big),
    \end{split}
\end{equation}
\end{prop}

The proof of Proposition \ref{prop ineq perr for QPV}, see Appendix \ref{Appendix proof ineq perr QPV}, relies on combining both expressions in \eqref{eq condition perr by def 1},
using $A_0^x+A_1^x\preceq \mathbb{I}$ and $B_0^{x'}+B_1^{x'}\preceq \mathbb{I}$ and bounding terms by the norm of the sums of the projectors $V_a^x$.\\

Therefore, using the above constraints, $p_{ans}$ can be upper bounded by the SDP problem:

\begin{equation}
\boxed{
\label{eq upperbound p_ans all constraints}
\begin{split}
    &\max \frac{1}{m}\sum_{x}(\expectedbraket{A_0^{x}B_0^{x}}+\expectedbraket{A_1^{x}B_1^{x}});\\
    &\textrm{subject to: the linear constraints for } \mathbf{S}^{\eta}_{MoE} \in \mathcal{Q}_\ell,\\& \textrm{ \hspace{15mm} and equations \eqref{eq restriction error epsilon}, \eqref{eq perr constraint sum ab} and \eqref{eq constraint perr 2 simplified}}. 
\end{split}}
\end{equation}
Fig.~\ref{Fig secure regions 322} shows a $SSR$ for the  $\mathrm{QPV^{\eta}_{(3,2,2)}}$ protocol obtained from the solutions of the SDP \eqref{eq upperbound p_ans all constraints} using the Ncpol2sdpa package \cite{Wittek_2015_Ncpol2Sdpa} in Python, see \cite{Code} for the code. Notice that Fig.~\ref{Fig secure regions 322} shows that the security region for this protocol is greater than the $SR$ of \QPVBBeta, meaning that it is more secure. However, analytical bounds on the best attack (even no loss) for these MoE games are so far only known for the BB84 game, and therefore we can not show tightness of our results beyond the \QPVBBeta~ protocol---a gap between our best upper bounds and lower bounds remains.
Numerical results from \eqref{eq upperbound p_ans all constraints} show that for different arbitrary $m$, $p_{ans}$ for $p_{err}=0$ is upper bounded by $\frac{1}{m}$, which is attainable by the strategy of Alice randomly guessing $x$, measuring in this basis, broadcasting the outcome and answering if she was correct and otherwise claiming no photon.\\ 
\begin{figure}[h]
\centering
\includegraphics[width=95mm]{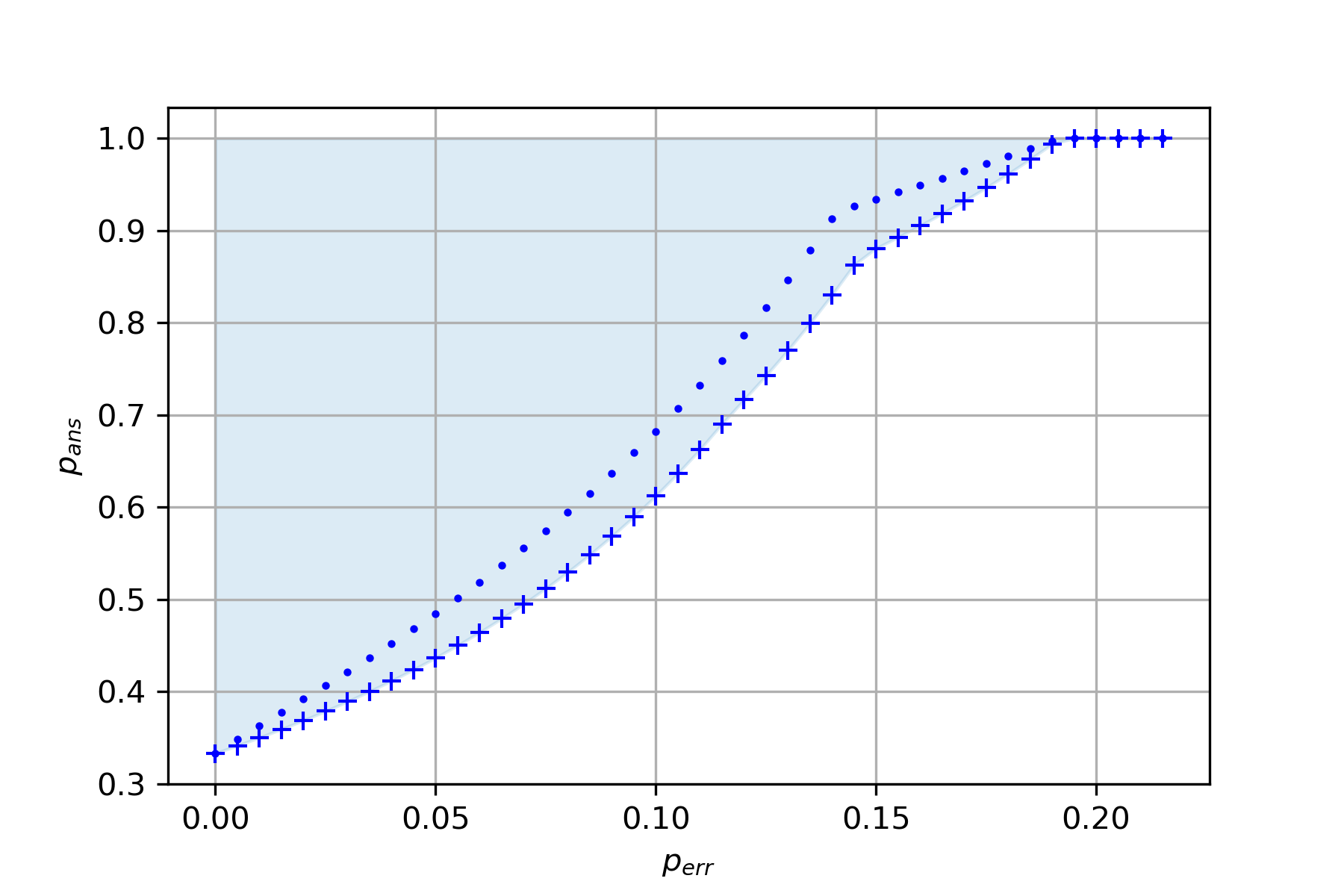}
\caption{Solutions of the first (blue dots) and second level (blue pluses) of the NPA hierarchy for the SDP \eqref{eq upperbound p_ans all constraints} highlighting a $SSR\subseteq SR$ of the $\mathrm{QPV^{\eta}_{(3,2,2)}}$ protocol in light blue. }
\label{Fig secure regions 322}
\end{figure}
As stated above, finding the smallest $p_{err}$ such that $p_{ans}=1$ can be used to upper bound the winning probability $p_{win}$. Fig.~\ref{Fig upper bound Tomamichel and ours} shows the values upper bounding $p_{win}$ with the SDP~\eqref{eq upperbound p_ans all constraints}, showing security of the protocol for different $(m,m_{\theta}, m_{\varphi})$, compared with the upper bound obtained by \eqref{eq upperbound Tomamichel} \cite{TomamichelMonogamyGame2013}, when the attackers always `play', where significant differences between both methods can be appreciated. The security of the sequential repetition of \QPVBBetam~is obtained as in the proof of Theorem~\ref{theorem seq repetition entangled} adjusting the parameters accordingly for the particular choice of $m_{\theta\varphi}$. 

\begin{figure}[h]
\centering
\includegraphics[width=100mm]{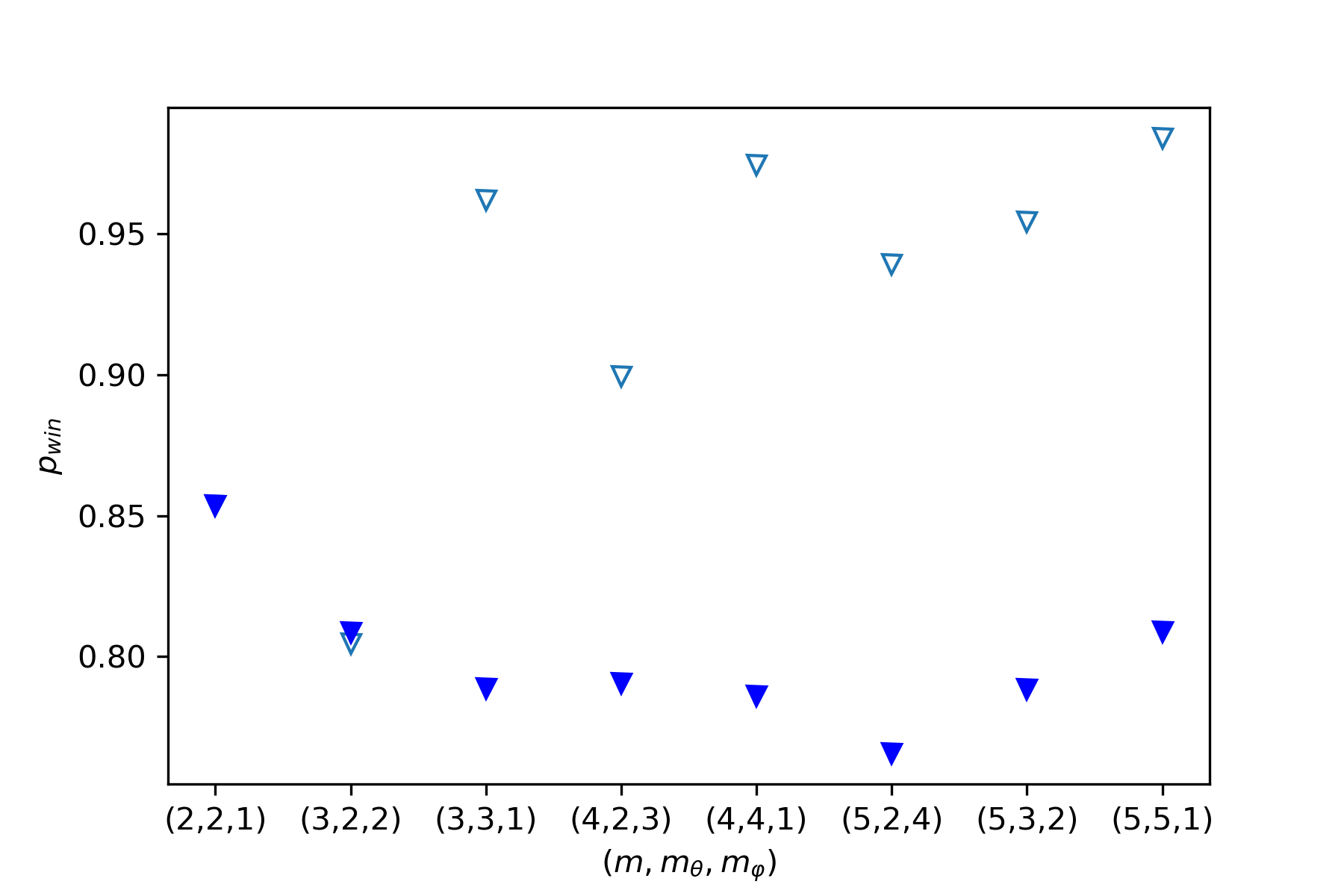}
\caption{Upper bounds of $p_{win}$ using \eqref{eq upperbound Tomamichel} (empty triangles) and the solution of SDP \eqref{eq upperbound p_ans all constraints} (solid triangles) for the \QPVBBetam~protocol for different values of $(m,m_{\theta}, m_{\varphi})$.}
\label{Fig upper bound Tomamichel and ours}
\end{figure}

\newpage\section{The \texorpdfstring{\QPVBBetafm}{QPV-BB84-eta,f} protocol and its security against entangled attackers}\label{section m basis and entanglement}

In Section \ref{section 2 basis and entanglement} we have shown that  \QPVBBetaf~is secure against entangled attackers if they hold a bounded number of qubits, but it is currently non-implementable experimentally due to the fact that it is not secure for $\eta\leq 1/2$. On the other hand, in Section \ref{section m basis no entanglement} we have shown that extending the \QPVBBeta~ protocol to $m$ bases allows for more resistance to photon loss. 
In this section, we use the results from Section~\ref{section m basis no entanglement} to prove Lemma~\ref{lemma average distace eta Delta/2}. This lemma will form the key tool to re-apply the analysis in \cite{bluhm2022single} to our case, and as a consequence we will show that the \QPVBBetafm~protocol is secure against entangled attackers (and more loss-tolerant than \QPVBBetaf). 
The results in this section are proven for arbitrary $m$, however they are based on solving the SDP described in \eqref{eq upperbound p_ans all constraints}, which needs $m$ to be fixed.
Here, we obtain numerical results for the two particular cases $m_{\theta\varphi}=(3,2,2)$ and $m_{\theta\varphi}=(5,3,2)$, but to obtain the results for any $m_{\theta\varphi}$ that would potentially be applied experimentally, one just needs to solve \eqref{eq upperbound p_ans all constraints} and the corresponding relaxation, see below.
Moreover, here we solve the semidefinite programs for a complete range of $p_{err}$, thereby obtaining an exhaustive characterization for the fixed $m_{\theta\varphi}$, but for an experimental implementation it would just be needed to solve the SDPs for the ranges that the experimental set-up requires. 

For some $n\in\mathbb{N}$, consider a $2n$-bit function $f:\{0,1\}^n \times \{0,1\}^n \to [m]$. \emph{One round of the $m_{\theta\varphi}$-basis lossy-function protocol}, denoted by \QPVBBetafm, is described as in Definition \ref{def qpv bb84 f} changing the range of $f$. The corresponding general attack is described as in the attack to \QPVBBetaf, extended by changing the range of the function~$f$, and similarly for the \emph{$q$-qubit strategy}.

With the same reasoning as in Section \ref{section 2 basis and entanglement}, from \eqref{eq upperbound p_ans all constraints} and its corresponding relaxation given by $\expectedbraket{A^x_aB^x_b}\leq\xi$, for all $a\neq b$, we have that there exists functions $w_{m_{\theta\varphi}}(\eta)$ and $\Tilde w_{m_{\theta\varphi}}^{\xi}(\eta)$ such that 
\begin{equation}\label{eq ineq SDP  m basis}
    \frac{1}{m}\sum_{i\in[m]}p_{\phi}^{i,\eta}\leq w_{m_{\theta\varphi}}(\eta),
\end{equation}
and
\begin{equation}\label{eq ineq SDP with errors m basis}
    \frac{1}{m}\sum_{i\in[m]}\Tilde{p}_{\phi}^{i,\eta,\xi}\leq \Tilde{w}^{\xi}_{m_{\theta\varphi}}(\eta),
\end{equation}
that are such that $w_{m_{\theta\varphi}}(\eta)\leq\Tilde{w}_{m_{\theta\varphi}}^{\xi}(\eta)$. The probabilities $p_{\phi}^{i,\eta}$ and $\Tilde{p}_{\phi}^{i,\eta,\xi}$ are as in \eqref{eq prob 0 eta}
and \eqref{eq prob 0 eta with error}, respectively, with $i\in[m]$. See Fig.~\ref{figure upper bound w Tilde w for m=3 and m=5} for a numerical approximation of the functions $w_{m_{\theta\varphi}}(\eta)$ and $\Tilde{w}_{m_{\theta\varphi}}^{\xi}(\eta)$ for different $m_{\theta\varphi}$, see \cite{Code} for the code.

Now, we show a lemma that formalizes the idea that if a set of $m$ quantum states can be used to be correct with high probability in an attack of the protocol, then, their average distance is lower bounded by a certain amount.
This has the interpretation that these states cannot all be simultaneously arbitrarily close. 
This means that exists an $\varepsilon_0$ such that for all $\varepsilon\leq\varepsilon_0$, $\bigcap_{i\in[m]}\mathcal{S}_i^{\varepsilon}=\emptyset$,  where $\mathcal{S}_i^{\varepsilon}$ is defined as in Definition~\ref{def S_i^e} with $i\in[m]$.

 \begin{figure}[h]
\centering
\subfigure[]{\includegraphics[width=74mm]{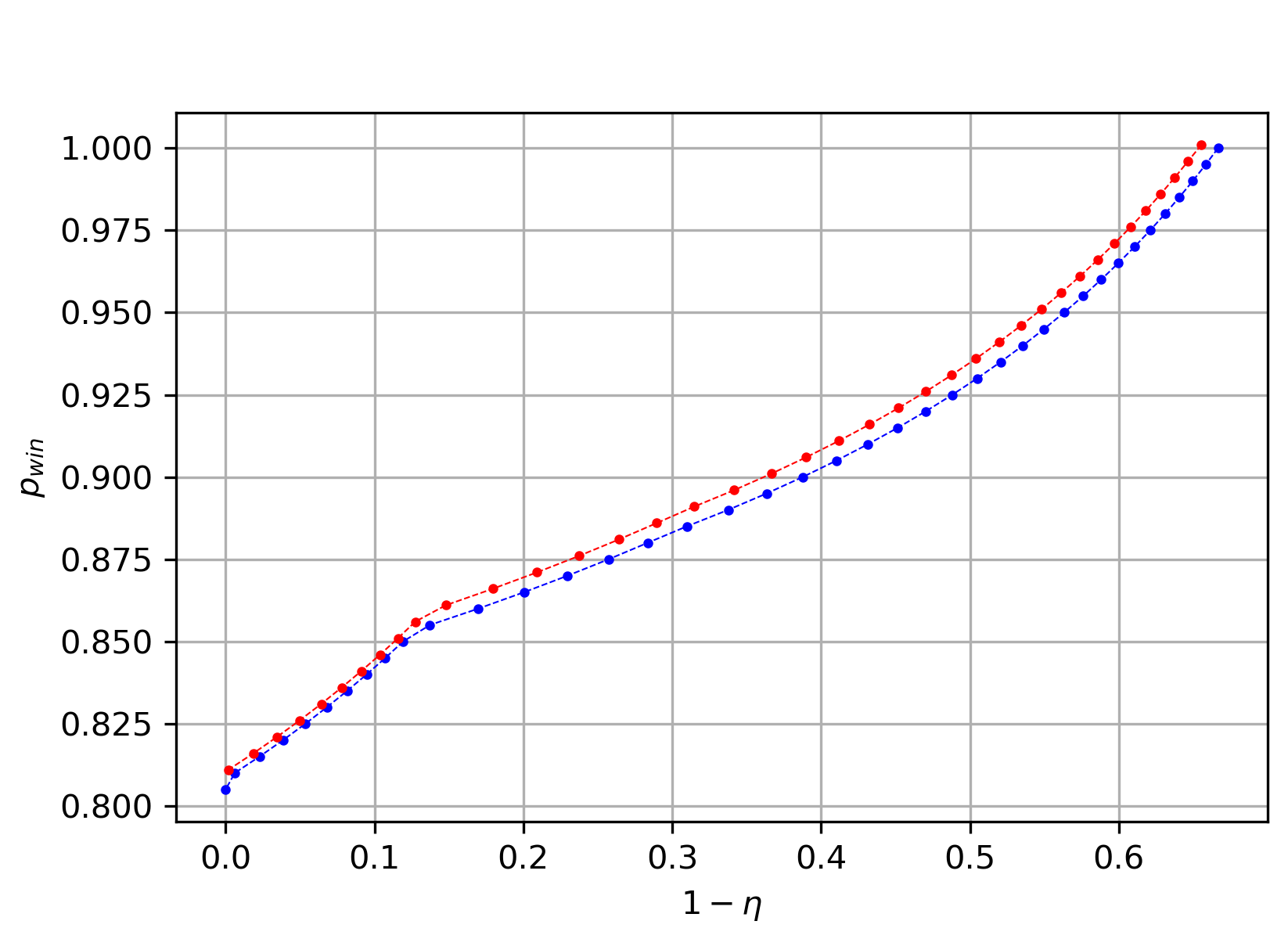}}
\subfigure[]{\includegraphics[width=74mm]{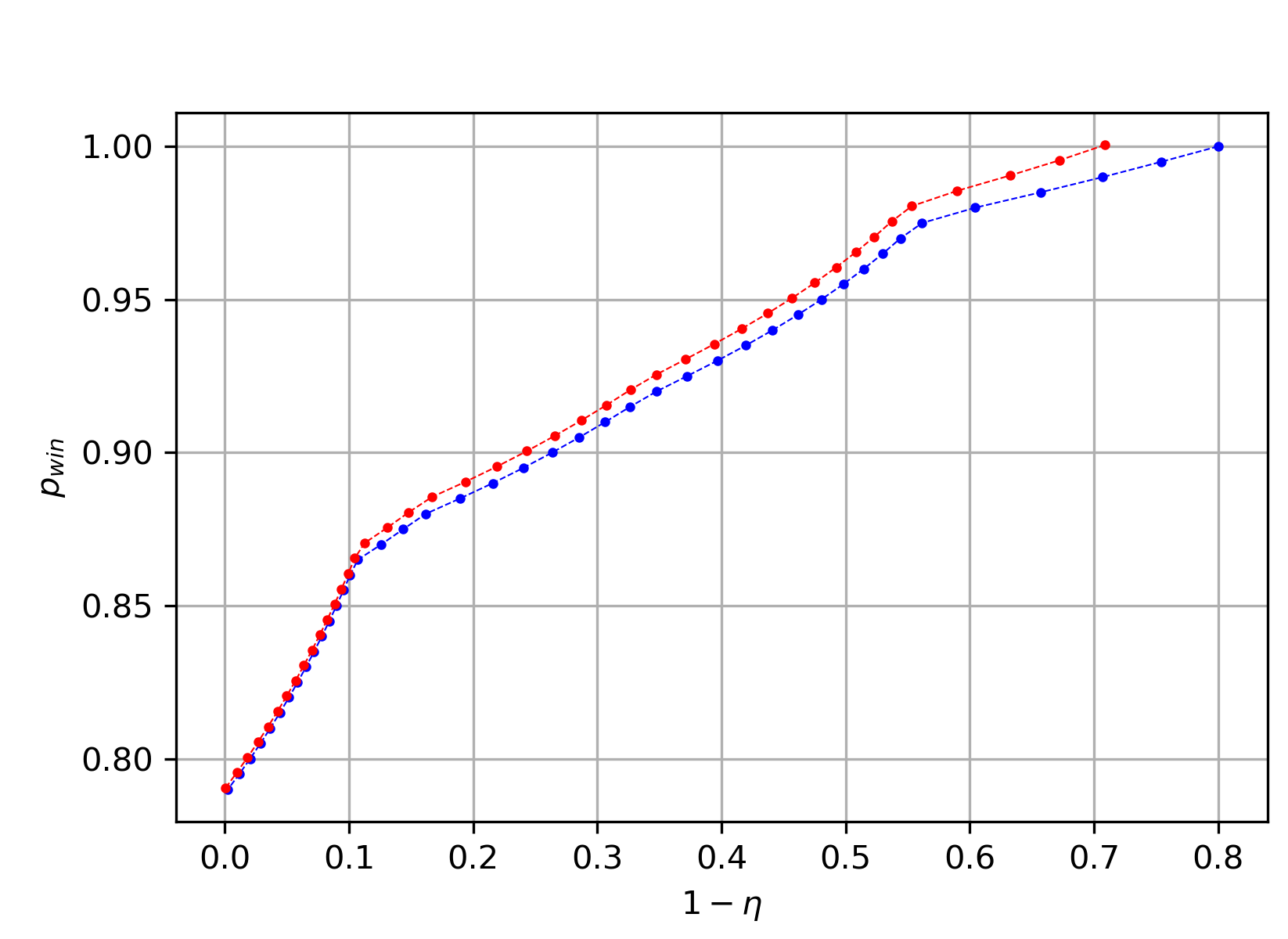}}
\caption{Upper bounds of the winning probability given by \eqref{eq upperbound p_ans all constraints} (blue dots), corresponding to a numerical representation of the function $w_{m_{\theta\varphi}}(\eta)$. Red dots correspond to a numerical representation of  the function $\Tilde w^{\xi}_{m_{\theta\varphi}}(\eta)$, which is obtained by adding $\xi=0.001$ to the relaxation of \eqref{eq upperbound p_ans all constraints} where the attackers are allowed to make errors with probability $\xi$, for (a) $m_{\theta\varphi}=(3,2,2)$, and (b) $m_{\theta\varphi}=(5,3,2)$. The blue dots in (a) are obtained by the level 2 of the NPA hierarchy, and the rest of the values, via the level `1+AB'. The continuous interpolation between values  is meant for a better viewing of the plot.}
\label{figure upper bound w Tilde w for m=3 and m=5}
\end{figure}

\begin{lemma}\label{lemma average distace eta Delta/2}
Let $\ket{\psi_i}$ be such that $p^{i,\eta}_{\psi_i}\geq \Tilde{w}^{\xi}_{m_{\theta\varphi}}(\eta)+\Delta$, for all $i\in [m]$, for some $\Delta>0$, which implies that \ $\ket{\psi_i}\in\mathcal{S}_i^{\varepsilon}$, for $\varepsilon=1-(\Tilde{w}^{\xi}(\eta)+\Delta)$. 
Then, for all $j\in [m]$
\begin{equation}\label{eq in lemma at on average they are far}
    \underset{i\in[m]}{\mathbb{E}}[\mathcal{D}(\ket{\psi_j},\ket{\psi_i})]\geq \frac{\eta \Delta}{2}.
\end{equation}
\end{lemma}
\begin{proof}
Let $\xi_{ij}=\mathcal{D}(\ket{\psi_j},\ket{\psi_i})$ and let $\xi=\max_{ij}\xi_{ij}$. From an analogous application of equation~\eqref{eq inequality p ptilde and delta}, we have that for all $i,j\in[m]$,
\begin{equation}
    p^{i,\eta}_{\psi_i}\leq \frac{2}{\eta}\xi_{ij}+\Tilde{p}^{i,\eta,\xi_{ij}}_{\psi_j}\leq \frac{2}{\eta}\xi_{ij}+\Tilde{p}^{i,\eta,\xi}_{\psi_j},
\end{equation}
where in the second inequality we used that replacing $\xi_{ij}$ by $\xi$ is a relaxation of the restrictions of the maximization \eqref{eq prob 0 eta with error}. Fixing $j$ and summing over $i$,
\begin{equation}\label{eq for m basis p and ptilde sum over i}
    \sum_{i\in[m]}  p^{i,\eta}_{\psi_i}\leq \frac{2}{\eta}\sum_{i\in[m]}\xi_{ij}+\sum_{i\in[m]}\Tilde{p}^{i,\eta,\xi}_{\psi_j}.
\end{equation}
By hypothesis, each term in the left-hand side is lower bounded by $\Tilde{w}^{\xi}_{m_{\theta\varphi}}(\eta)+\Delta$. This, together with \eqref{eq ineq SDP with errors m basis}, lead to \eqref{eq in lemma at on average they are far}. 
\end{proof}

Since Lemma \ref{lemma average distace eta Delta/2} implies that if a set of $m$ states `performs well' on their respective inputs, their average distance with respect to an arbitrary state is at least a certain amount, meaning that there are at least two states that differ by such an amount, it has as consequence that Alice and Bob in some sense have to decide (at least one) strategy not to follow before they communicate. Consequently, if the dimension of the state they share is small enough, a classical description of the first part of their strategy yields a compression of $f$. The notion of the following definition captures this classical compression.

\begin{definition} Let $q,k,n\in\mathbb{N}$, $\varepsilon>0$. Then,
\begin{equation}
    g_{m_{\theta\varphi}}: \{0,1\}^{3k}\rightarrow 2^{[m]}	\setminus [m]
\end{equation}
is an $(\varepsilon,q)$-classical rounding restriction of size $k$ is for all $f:\{0,1\}^{2n}$, for all states $\ket{\phi}$ on $2q+1$ qubits, for all $l\in\{1,...,2^{2n}\}$ and for all $(\varepsilon,l)$-perfect $q$-qubit strategies for $QPV^{f}_{m_{\theta \varphi}}$, there are functions $f_A:\{0,1\}^n\rightarrow\{0,1\}^k$, $f_B:\{0,1\}^n\rightarrow\{0,1\}^k$ and $\lambda\in\{0,1\}^k$ such that $f(x,y)\in g(f_A(x),f_B(y),\lambda)$.
\end{definition}

\begin{lemma}\cite{bluhm2022single}\label{lemma unit vectors trace distance norm}
Let $\ket{x},\ket{y}\in\mathbb{C}^d$, for $d\in\mathbb N$, be two unit vectors. Then, $\mathcal{D}(\ket{x},\ket{y})\leq \| \ket{x}-\ket{y}\|_2$.
\end{lemma}
\begin{lemma}\label{lemma classical rounding size}
Let $\Delta> 0$, and let $0\leq\varepsilon\leq \varepsilon_0$, where $\varepsilon_0$ is such that $\ket{\psi_i}\in\mathcal{S}_i^{\varepsilon}$, implies $\underset{i\in[m]}{\mathbb{E}}[\mathcal{D}(\ket{\psi_j},\ket{\psi_i})]\geq \frac{\eta \Delta}{2}$. Then there exists an $(\varepsilon,q)$-classical rounding restriction of size $k=\log(\lceil\frac{4}{2^{\frac{1}{3}} (\eta \Delta +4)^{\frac{1}{3}}-2}\rceil)2^{2q+2}$. 

\end{lemma}
\begin{proof}
We follow the same techniques as in the proof of Lemma~3.12 in \cite{bluhm2022single}. Let $\delta=\frac{\sqrt[3]{\eta \Delta +4}}{2^{2/3}}-1-\epsilon$, where $\epsilon>0$ is infinitesimally small, and consider $\delta$-nets $\mathcal{N}_S$, $\mathcal{N}_A$ and $\mathcal{N}_B$, where the first is for the set of pure states on $2q+1$ qubits in Euclidean norm and the other nets are for the set of unitaries in dimension $2^q$ in operator norm. They are such that $\abs{\mathcal{N}_S}$, $\abs{\mathcal{N}_A}$, $\abs{\mathcal{N}_B}\leq 2^k$. Let $\ket{\varphi}\in\mathcal{N}_S$, $U_A\in\mathcal{N}_A$, and $U_B\in\mathcal{N}_B$ be the elements with indices $x'\in\{0,1\}^k$, $y'\in\{0,1\}^k$ and $\lambda\in\{0,1\}^k$, respectively. We define $g$ as $g(x,y,\lambda)=\{j \mid U\otimes V\ket{\varphi}\in\mathcal{S}_j^{\varepsilon}\}$. We are going to show that $g$ is an $(\varepsilon,q)$-classical rounding restriction.\\
Let $i$, $j$ be such that $\mathcal{D}(\ket{\psi_i},\ket{\psi_j})\geq \eta\Delta/2$. Let $\ket{\psi}$, $\{U_A^x,U_B^y\}_{xy}$ be from a $q$-qubit strategy for \QPVBBetafm, and choose $\lambda$, $f_A(x)$ and $f_B(y)$ to be the closest elements to $\ket{\psi}$, $U_A^x$ and $U_B^y$, respectively, in their corresponding $\delta$-nets in the Euclidean and operator norm, respectively, (if not unique, make an arbitrary choice) and let $\ket{\varphi},U_A,U_B$ be their corresponding elements. Assume $U_A^x\otimes U_B^y\ket{\psi}\in\mathcal{S}_i^{\varepsilon}$. Then,
\begin{equation}\label{eq distance in the net}
\begin{split}
    \mathcal{D}(&U_A^x\otimes U_B^y\ket{\psi},U_A\otimes U_B\ket{\varphi})\leq \| U_A^x\otimes U_B^y\ket{\psi}-U_A\otimes U_B\ket{\varphi}\|_2\\&\leq \|(U_A+U_A^x-U_A)\otimes(U_B+U_B^y-U_B)(\ket{\varphi}+\ket{\psi}-\ket{\varphi})-U_A\otimes U_B\ket{\varphi}\|_2
    \\&\leq3\delta+3\delta^2 +\delta^3< \frac{\eta \Delta/2}{2},
\end{split}
\end{equation}
where in the first inequality, we have used Lemma \ref{lemma unit vectors trace distance norm}, in the second, we have used the triangle inequality and the inequality $\|X\otimes Y \ket{x}\|_2\leq \|X\|_{\infty}\|Y\|_{\infty}\| \ket{x}\|_2$, together with $\|U_A^x-U_A\|_{\infty},\\{\|U_B^y-U_B\|_{\infty}},\|\ket{\psi}-\ket{\varphi}\|\leq \delta$, and, finally, in the last inequality we used that $\delta<\frac{\sqrt[3]{\eta \Delta +4}}{2^{2/3}}-1$.
Thus, $U_A\otimes U_B\ket{\varphi}$ is closer to $\mathcal{S}_i^{\varepsilon}$ than to $\mathcal{S}_j^{\varepsilon}$.\\
Consider an $(\varepsilon, l)$-perfect strategy for \QPVBBetafm~ and let $(x,y)$ be such that the attackers are caught with probability at most $\varepsilon$ and such that $f(x,y)=i$. In particular, we have that $U_A^x\otimes U_B^y\ket{\psi}\in\mathcal{S}_i^{\varepsilon}$, and because of \eqref{eq distance in the net}, $f(x,y)\in g(f_A(x),f_B(y),\lambda)$. Since there are at least $l$ pairs $(x,y)$ fulfilling it, $f(x,y)\in g(f_A(x),f_B(y),\lambda)$ holds on at least $l$ pairs $(x,y)$ and therefore $g$ is an $(\varepsilon,q)$-classical rounding restriction. The size of $k$ follows from Lemma \ref{lemma size delta-net}.

\end{proof}

Given $m_{\theta\varphi}$, we denote by $\eta_{m_{\theta\varphi}}$ the maximum $\eta$ such that $\Tilde w_{m_{\theta\varphi}}^{\xi}(\eta)+\Delta\leq 1$, e.g.\ for $m_{\theta\varphi}=(2,2,1)$, i.e.\ \QPVBBeta, that corresponds to 0.53. From \eqref{eq ineq SDP with errors m basis} and picking $\Delta=0.009$,  $\eta_{m_{\theta\varphi}}=0.36$ for $m_{\theta\varphi}=(3,2,2)$ and $\eta_{m_{\theta\varphi}}=0.34$ for $m_{\theta\varphi}=(5,3,2)$ (by picking a smaller $\Delta$, the latter gets closer to 0.2). 

\begin{lemma}
\label{lemma holds on less than beta}
Let $\Delta=0.009$, $\eta\in(\eta_{m_{\theta\varphi}},1]$, $\varepsilon \in [0,1]$ $n$, $k$, $q \in \mathbb N$, $n \geq 10$. Moreover, fix an $(\varepsilon, q)$-classical rounding $g$ of size $k$ with $k=\log(\lceil\frac{4}{2^{\frac{1}{3}} (\eta \Delta +4)^{\frac{1}{3}}-2}\rceil)2^{2q+2}$. Let $
q \leq \frac{1}{2}n-5$. Then, a uniformly random $f: \{0,1\}^{2n} \to \{0,1\}$ fulfills the following with probability at least $1 - 2^{-2^{n}}$:
 For any $f_A:\{0,1\}^n \to \{0,1\}^{k}$, $f_B:\{0,1\}^n \to \{0,1\}^{k}$, $\lambda \in \{0,1\}^{k}$,  $ f(x,y)\in g(f_A(x), f_B(y), \lambda) $ holds on less than $1-\beta_{m_{\theta\varphi}}$ of all pairs $(x,y)$, for certain $\beta_{m_{\theta\varphi}}>0$ ($\beta_{m_{\theta\varphi}}=0.15$ for $m_{\theta\varphi}=(3,2,2)$ and $\beta_{m_{\theta\varphi}}=0.13$ for $m_{\theta\varphi}=(5,3,2)$ ).
\end{lemma}

\begin{proof}
 For simplicity, denote $\beta_{m_{\theta\varphi}}$ by $\beta$. We want to estimate the probability that for a randomly chosen $f$, we can find $f_A$ and $f_B$ such that the corresponding function $g$ is such that $\Pr_{x,y}[f(x,y)\in g_{m_{\theta \varphi}(f_A(x),f_B(y),\lambda)}]\geq(1-\beta)$.
\begin{equation}\label{eq bounding probability of g}
\begin{split}
  \Pr&[f:\exists f_A,f_B,\lambda \textrm{ s.t. } \Pr[x,y: f(x,y)\in g_{m_{\theta \varphi}}(f_A(x),f_B(y),\lambda)]\geq(1-\beta)]\\&=\frac{\abs{\{f:\exists f_A,f_B,\lambda \textrm{s.t. }\Pr[x,y: f(x,y)\in g_{m_{\theta \varphi}}(f_A(x),f_B(y),\lambda)]\geq(1-\beta) \}}}{\abs{\{f:\{0,1\}^{2n}\rightarrow[m]\}}}\\& \leq 
  \frac{\abs{\{f:\exists f_A,f_B,\lambda \textrm{ s.t. } \forall x,y,  f(x,y)\in g_{m_{\theta\varphi}}(f_A(x),f_B(y),\lambda) \}}}{m^{2^{2n}}}\sum_{i=0}^{\beta2^{2n}}\binom{2^{2n}}{i}(m-1)^i \\&\leq\frac{1}{m^{2^{2n}}}2^{(2^{n+1}+1)k}(m-1)^{2^{2n}}(m-1)^{\beta2^{2n}}2^{h(\beta)2^{2n}}\\&
  =2^{(h(\beta)-\log m+(1+\beta)\log m-1)2^{2n}+(2^{n+1}+1)k}.
\end{split}
\end{equation}
Where in the first equality we used that $f$ is chosen uniformly at random, in the second step we estimate the numerator by considering a ball in Hamming distance around every function $g$ that cab be expressed suitable by $f_A$, $f_B$, $\lambda$,  and in the third step we bounded $(m-1)^i$ in the sum by $(m-1)^{\beta2^{2n}}$ and we used the inequality $\sum_{l=0}^{\lambda n}\binom{n}{n}\leq 2^{nh(\lambda)}$ for $n\in\mathbb{N}$ and $\lambda\in(0,1/2)$ \cite{book_binary_entropy_inequaity}. For $m_{\theta\varphi}=(3,2,2)$ and $m_{\theta\varphi}=(5,3,2)$, $\Delta=0.009$, $k=\log(\lceil\frac{4}{2^{\frac{1}{3}} (\eta \Delta +4)^{\frac{1}{3}}-2}\rceil)2^{2q+2}$ for $\eta\in(\eta_{m_{\theta\varphi}},1]$ and $q\leq n/2-5$, \eqref{eq bounding probability of g} is strictly upper bounded by  $2^{-2^n}$.
\end{proof}

Lemma \ref{lemma holds on less than beta} has the same interpretation as Lemma \ref{lemma holds on less than beta}. This leads to our second main theorem, which provides a lower bound of the probability that attackers pre-sharing entanglement are caught in a round of the \QPVBBetafm~ protocol.

\begin{theorem}\label{theorem q>=n/2-5 for m-BB84 eta}
Consider the most general attack a round of the \QPVBBetafm~protocol for a transmission rate $\eta\in(\eta_{m_{\theta\varphi}},1]$ and prover's error rate $p_{err}$. Let $\Delta=0.009$. If the attackers respond with probability $\eta$ and  control at most $q$ qubits at the beginning of the protocol, and $q$ is such that
\begin{equation}
    q\leq \frac{n}{2}-5,
\end{equation}
then,
\begin{equation}
    \pr{\textsc{V}_\textsc{AB}=\textsc{c}}\leq\eta(1-\beta_{m_{\theta\varphi}}[1-(\Tilde{w}^{\epsilon}(\eta)+\Delta)]).
\end{equation}
\end{theorem}

\begin{proof}
Let $0\leq\varepsilon\leq\varepsilon_0=1-(\Tilde{w}^{\xi}_{m_{\theta\varphi}}(\eta)+\Delta)$. By Lemma~\ref{lemma classical rounding size} there exists $g_{m_{\theta\varphi}}$ $(\varepsilon,q)$-classical rounding of size $k=\log(\lceil\frac{4}{2^{\frac{1}{3}} (\eta \Delta +4)^{\frac{1}{3}}-2}\rceil)2^{2q+2}$. Fix $f:\{0,1\}^{2n}\to[m]$ such that $f(x,y)\in g(f_A(x),f_B(y),\lambda)$ holds on less than $1-\beta_{m_{\theta\varphi}}$ of all pairs $(x,y)$, for all $f_A,f_B$ and $\lambda$ as defined previously. By Lemma~\ref{lemma holds on less than beta}, a uniformly random $f$ will have this property with probability at least $1-2^{-2^n}$.\\
On the other hand, assume that there is a $(\varepsilon,(1-\beta_{m_{\theta\varphi}})\cdot 2^{2n})$-perfect $q$-qubit strategy for \QPVBBetafm. Then, the corresponding $f_A,f_B,\lambda$ satisfy $f(x,y)\in g_{m_{\theta\varphi}}(f_A(x),f_B(y),\lambda)$ on at least $(1-\beta_{m_{\theta\varphi}})\cdot2^{2n}$ pairs $(x,y)$. This is a contradiction of the choice of $f$. Therefore, with probability at least $1-2^{-2^n}$ the function $f$ is such that there are no $(\varepsilon,(1-\beta_{m_{\theta\varphi}})\cdot2^{2n})$-perfect $q$-qubit strategies for \QPVBBetafm. 
Hence, for every strategy that the attackers can implement, on at least $\beta_{m_{\theta\varphi}}$ of the possible strings $(x,y)$, they will not be correct with probability at least $\varepsilon$. 
\end{proof}

We see then that if $p_{err}<\beta_{m_{\theta\varphi}}[1-(\Tilde{w}^{\epsilon}_{m_{\theta\varphi}}(\eta)+\Delta)]$, attackers who pre-share entanglement cannot answer `\textsc{Correct}' with sufficiently high probability even in the presence of photon loss (more loss-tolerant than the \QPVBBetaf~protocol). 
The security of the sequential repetition of \QPVBBetafm~is obtained as in the proof of Theorem~\ref{theorem seq repetition entangled} adjusting the parameters accordingly for the particular choice of $m_{\theta\varphi}$. 

\section{Application to QKD}
In \cite{TomamichelMonogamyGame2013} security of one-sided device-independent quantum key distribution (DIQKD) BB84 \cite{BB84} was proven using a monogamy-of-entanglement game.
In order to reduce an attack to the protocol to a MoE game, we consider an entanglement-based variant of the original BB84 protocol, which implies security of the latter \cite{BB84withoutBell}. The BB84 entangled version, tolerating an error~$p_{err}$ in Bob's measurement results, is described as follows \cite{TomamichelMonogamyGame2013}:
\begin{enumerate}
    \item Alice prepares $n$ EPR pairs, keeps a half and sends the other half to Bob from each EPR pair. Bob confirms he received them. 
    \item Alice picks a random basis, either computational or Hadamard, to measure each qubit, sends them to Bob and both measure, obtaining $X$ and $Y$, respectively. 
    \item Alice sends a random subset $X_T\subset X$ of size $t$ and sends it to Bob. If the corresponding $Y_T$ has a relative Hamming distance greater than $p_{err}$, they abort.
    \item In order to perform error correction, Alice sends a syndrome 
    $S(X_{\Bar{T}})$ of length $s$ (the leakage) and a random hash function $F:\{0,1\}^{n-t}\rightarrow \{0,1\}^l$, where $l$ is the length of the final key, from a universal family of hash functions to Bob.  
    \item Finally, for privacy amplification, Alice and Bob compute $K=F(X_{T^c})$ and $\hat{K}=F(\hat{X}_{T^c})$, where $\hat{X}_{T^c}$ is the corrected version of $Y_{T^c}$. 
\end{enumerate}

The security proof relies on considering an eavesdropper Eve and an untrusted Bob's measurement device which could behave maliciously, associating this situation with a MoE game where Eve's goal is to guess the value of Alice's (playing now the role of the referee) raw key~$X$. They prove that it can tolerate a noise up to 1.5\% asymptotically.  Here, we apply the techniques used above for QPV to pass from a MoE game to SDP to get numerical results. We prove security for $n=1$, still remaining if the results can be generalized to arbitrary~$n$, and furthermore we prove security considering photon loss.  The latter consideration takes into account the loss of photons after Bob confirmed their reception in step 1., see e.g.~\cite{Niemietz_2021}. \\
In a similar way as deriving the restrictions of \eqref{eq upperbound p_ans all constraints}, we maximize the probability of answering for Eve controlling maliciously Bob's measurement device. Alice measures $\{V^x_0,V^x_1\}_{x\in\{0,1\}}$, where $V_0^{x}=\ketbra{0_{x}}{0_{x}}$ and $V_1^{x}=\ketbra{1_{x}}{1_{x}}$, i.e.~measures in the computational or Hadamard basis and the two cooperative adversaries (Eve and Bob's device), as argued above, perform projective measurements $\{E^x_e\}$,  $\{B^x_b\}$, respectively, where $e,b\in\{0,1,\perp\}$. Therefore, the above techniques applied to QKD reduce to maximize  $p_{ans}=\frac{1}{2}\sum_{x\in \{0,1\},e\in \{0,1\}}\expectedbraket{E_e^xB_e^x}$ subject to \textit{(i)} the strategy for the extended MoE game in $\mathcal{Q}_{\ell}$, \textit{(ii)} the restrictions of the QKD protocol and \textit{(iii)} Bob's device subject to a measurement error $p_{err}$. \\
Constraint \textit{(i)} remains unaltered from the above discussion, and \textit{(ii)} and \textit{(iii)} are given in the following way:\\
\textbf{Constraint \textit{(ii)}} Since Eve's task is to guess Alice's raw key, her measurements cannot be distinct, and therefore, for all $e\in\{0,1\}$, and for all $b$,
\begin{equation}
\label{eq Eve equal meas Alice constriant}
    \expectedbraket{V_eE_{1-e}B_b}=0.
\end{equation}
\textbf{Constraint \textit{(iii)}} Bob's device measurement error is mathematically expressed as
\begin{equation}
\label{eq condition perr QKD by def 1}
    \frac{\expectedbraket{V_{0}^{x}E_{0}^{x}B_{1}^{x}}}{\expectedbraket{V_{0}^{x}(E_{0}^{x}+E_{1}^{x})(B_{0}^{x}+B_{1}^{x})}}\leq p_{err}, \hspace{1cm} \frac{\expectedbraket{V_{1}^{x}E_{1}^{x}B_{0}^{x}}}{\expectedbraket{V_{1}^{x}(E_{0}^{x}+E_{1}^{x})(B_{0}^{x}+B_{1}^{x})}}\leq p_{err}.
\end{equation} 
\begin{prop} \label{prop ineq perr for QKD } Let $e,b\{0,1\}$, then the terms $\expectedbraket{E_e^{x}B_{b}^{x'}}$ can be bounded by $p_{err}$ by the below inequality:
\begin{equation}
    \label{eq perr constraint sum ab QKD}
    \sum_{eb}(2-\norm{V_e^x+V_b^{x'}})\expectedbraket{E_e^{x}B_{b}^{x'}}\leq p_{err}\sum_{e,b}\expectedbraket{E_e^{x'}B_{b}^{x'}}. 
\end{equation}
\end{prop}
The proof, see Appendix~\ref{Appendix proof ineq perr QKD}, uses the same techniques as for the proof of Proposition~\ref{prop ineq perr for QPV}. 
Using the above constraints, we can find a subset of the security region ($SSR$), defined analogously for the current case, for the QKD protocol with the following SDP:
\begin{equation}
\label{eq upperbound p_ans all constraints QKD}
\boxed{
\begin{split}
    &\max \frac{1}{2}\sum_{x}(\expectedbraket{E_0^{x}B_0^{x}}+\expectedbraket{E_1^{x}B_1^{x}});\\
    &\textrm{subject to: the linear constraints for } \mathbf{S}^{\eta}_{MoE} \in \mathcal{Q}_n,\\& \textrm{ \hspace{15mm} and Equation~\eqref{eq perr constraint sum ab QKD}}. 
\end{split}}
\end{equation}
A $SSR$ from the solutions of \eqref{eq upperbound p_ans all constraints QKD} for different $p_{err}$ for the first level of the NPA hierarchy is plotted in Fig.~\ref{Figure plot pans for QKD} 
\begin{figure}[h]
\centering
\includegraphics[width=100mm]{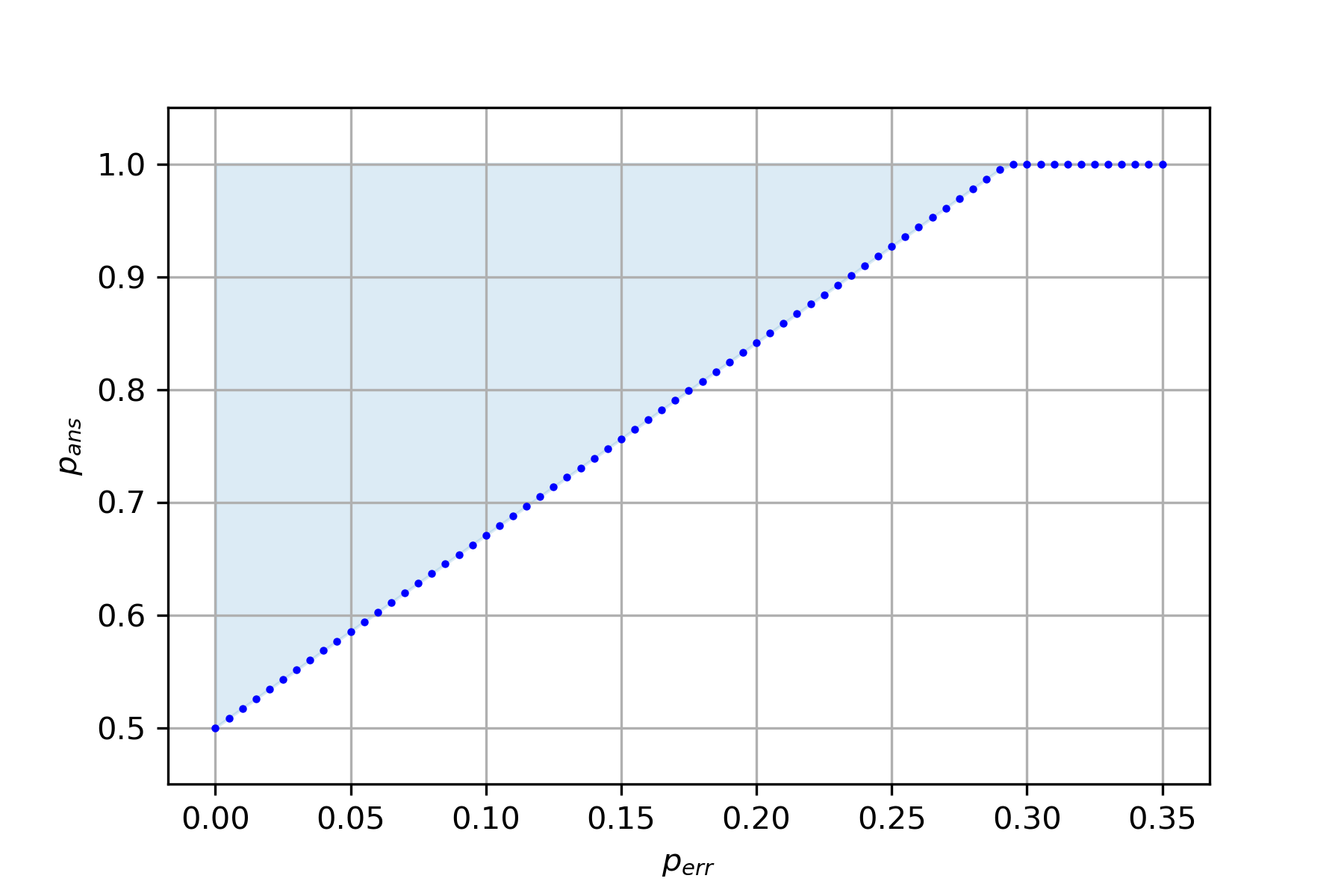}
\caption{Solution of the first level of the NPA hierarchy of the SPD \eqref{eq upperbound p_ans all constraints QKD} (blue dots) and the light blue area corresponds to a $SSR$.}
\label{Figure plot pans for QKD}
\end{figure}

From the solutions of \eqref{eq upperbound p_ans all constraints QKD}, represented in Fig.~\ref{Figure plot pans for QKD}, we find that for $p_{err}\approx 0.2929$, an eavesdropper Eve controlling Bob's device can always answer without being caught, implying that the probability of winning such a MoE game is upper bounded by $0.7071$. Therefore, the protocol, for $n=1$, can tolerate a noise at least up to $0.2929$. 

\newpage
\section{Discussion}

We have studied the $\mathrm{QPV_{BB84}}$ protocol, providing a tight characterization of its security with loss of quantum information and a prover subject to an experimental error under non-entangled attackers that can do LOQC.
Since this protocol is secure only for a transmission rate of photons $\eta\geq\frac{1}{2}$, which is still hard for current technology to achieve, we introduced an extension of the protocol more resistant to transmission loss of the quantum information.
The new protocol has the advantage, like the $\mathrm{QPV_{BB84}}$ protocol, that only a single qubit is required to be transmitted in per round of the protocol and the honest prover only needs to broadcast classical information, which does not encounter the problems of loss and slow transmission, making it a good candidate for future implementation. 

We have also extended our analysis to the lossy version of the \QPVBBf~ protocol, thereby exhibiting a protocol which is secure against attackers sharing an amount of entanglement that scales in the amount of \emph{classical} information, while the honest parties only need to manipulate a single qubit per round.
By showing (partial) loss-tolerance, the resulting (especially multi-basis) protocol will be much easier to implement in practice.
The security proof of that protocol still holds identically in case the transmission of the quantum state is slow, but only the classical bits are transmitted fast, which can be experimentally convenient.

We applied the proof techniques used to show security to improve the upper bounds known so far for certain types of monogamy-of-entanglement games and, as a particular application, we improve the security analysis of one-sided device-independent quantum key distribution for the particular case of $n=1$. However, we leave as an open question whether this can be generalized to show security for arbitrary $n$.

\paragraph{Open questions.}
A technical question in the QPV setting which we leave open is to obtain a tight finite statistical analysis for these protocols over multiple rounds.
It seems intuitively clear that, since honest parties can achieve a better combination of error and response rate than any attacker for any round, the serial repetition of the proposed protocols should be secure---and we indeed expect this to be the case.
However, attackers have a choice in what attack strategy to apply every round, and might even gain a slight amount by being adaptive, e.g.\ play a low-loss low-error round if the attackers had a lucky guess in a previous round.
Because of this, finding the technical tools to bounds the best possible attack over many rounds is non-trivial.

{The work of Johnston, Mittal, Russo, and Watrous \cite{Extended_non-local_games_andMoE_games} shows techniques how to handle extended non-local games, such as the monogamy-of-entanglement games, using SDPs.
Because our work directly extends the earlier independent partial results of Buhrman, Schaffner, Speelman, and Zbinden~\cite[Chapter 5]{FlorianThesis}, we use a slightly-different SDP formulation.
It would be very interesting to attempt to extend the results from \cite{Extended_non-local_games_andMoE_games} to extended non-local games where the players are allowed to return `loss' with some probability---and investigate whether their SDP formulation gives equivalent results to ours.}

Additionally, we note that (also for the proposed lossy protocols) an exponential gap remains between the entanglement required of the best attack known and the lower bounds we are able to prove. That is, we show security against attackers sharing a linear amount of entanglement, but we only know of an explicit attack whenever the attackers share an \emph{exponential} number of qubits.
More efficient attacks are known for specific functions, such as $f$ computable in logarithmic space~\cite{Buhrman_2013} (for a routing version of the protocol, however the technique can easily be adapted), cf.~\cite{cree2022code}, but even these take a polynomial amount of entanglement.
Closing this gap remains a large and interesting open problem.

A related open question from the other side therefore also remains: It would be helpful to exhibit these linear lower bounds for specific efficiently-computable functions, instead of just showing that the bound holds for \emph{most} functions.

\textbf{Acknowledgments.} We thank the Dutch Ministry of Economic Affairs and Climate Policy (EZK), supporting this work as part of the Quantum Delta NL programme, and Jaume de Dios Pont for fruitful discussions about the probability simplex. 

\bibliographystyle{alphaurl}
\bibliography{references}

\newpage
\begin{appendices}
\section{Non-local games and the NPA hierarchy} \label{appendix non-logal games and NPA hierarchy}
Let $\mathcal{G}$ be a non-local game where two non-communicating distant parties, Alice and Bob, have respective questions $x\in \mathcal{X}$ and $y\in \mathcal{Y}$, given according to a probability distribution $q(x,y)$, they input their questions in a respective black box, and they get as outputs measurement outcomes  $a\in \mathcal{A}$ and $b\in\mathcal{B}$, for $\mathcal{X},\mathcal{Y},\mathcal{A},\mathcal{B}$, finite alphabets. The winning condition is determined by the predicate $f(a,b,x,y)$, taking value 1 if the game is won and 0, otherwise. The behavior of the box is completely characterized by the probability of getting outcomes $a$ and $b$ having measured $x$ and $y$, $p(a,b|x,y)$, and the set of all probabilities, $\{p(a,b|x,y)\}$, encoded in a stochastic matrix $P\in L(\mathbb{R}^{\mathcal{X}}\otimes \mathbb{R}^{\mathcal{Y}},\mathbb{R}^{\mathcal{A}}\otimes \mathbb{R}^{\mathcal{B}})$, where $L$ is the set of linear operators, such that $P(a,b|x,y)=p(a,b|x,y)$, is called behavior. The average winning probability is given by
\begin{equation}
    \omega(\mathcal{G})=\sum_{x,y,a,b}q(x,y)f(a,b,x,y)p(a,b|x,y):=\expectedbraket{K,P},
\end{equation}
where $K$ is the matrix defined as $K(a,b|x,y)=q(x,y)f(a,b,x,y)$.

\begin{definition}\label{def of behaviour} A behavior $P$ is quantum if there exists a pure state $\ket{\psi}$ in a Hilbert space $\mathcal{H}$, a set of measurement operators $\{A_a^x\}_{a\in\mathcal{A}}$ for Alice, and a set of measurement operators $\{B_b^y\}_{b\in\mathcal{B}}$ for Bob such that for all $a\neq a'\in\mathcal{A}$ and $b\neq b'\in\mathcal{B}$,
\begin{equation}
\label{eq def behaviour}
    p(a,b|x,y)=\bra{\psi}A_a^xB_b^y\ket{\psi},
\end{equation}
with the measurement operators satisfying
\begin{enumerate}
    \item $A_a^{x\dagger}=A_a^x$ and $B_b^{y\dagger}=B_b^y$,
    \item $A_a^xA_{a'}^x$=0 and $B_b^yB_{b'}^y=0$,
    \item $\sum_{a\in\mathcal{A}}A_a^x=\mathbb{I}$ and $\sum_{b\in\mathcal{B}}B_b^y=\mathbb{I}$,
    \item $[A_a^x,B_b^y]=0$. 
\end{enumerate}
The tuple $\mathbf{S}=\{\ket{\psi},A_a^x,B_b^y\}_{x\in\mathcal{X},y\in\mathcal{Y},a\in\mathcal{A},b\in\mathcal{B}}$ is called strategy, and the set of all quantum behaviors is denoted by $\mathcal{Q}$. Abusing notation, we will denote $\mathbf{S}\in\mathcal{Q}$.
\end{definition}
Similarly, a behavior $P$ belongs to the set of quantum behaviors $\mathcal{Q}'$ if the Hilbert space can be written as $\mathcal{H}=\mathcal{H}_A\otimes\mathcal{H}_B$ and the measurement operators for Alice and Bob act on $\mathcal{H}_A$ and $\mathcal{H}_B$, respectively, and fulfil the same constraints as in Definition \ref{def of behaviour} (notice that commutativity is immediately implied because of the tensor product structure). Notice that by construction, $\mathcal{Q}'\subseteq\mathcal{Q}$ and for finite dimensional Hilbert space, they turn out to be identical \cite{NPA2008}.

Therefore, the behaviors \eqref{eq def behaviour} can be obtained via tensor product structure, which is the case that we will consider from now on.
The maximum winning probability using a quantum behavior $P$ is given by
\begin{equation}
\label{eq supremum w non-local game}
    \omega^*(\mathcal{G})=\sup_{P\in\mathcal{Q}}\expectedbraket{K,P}.
\end{equation}

In \cite{NPA2008}, Navascués, Pironio and Acín (NPA) introduced an infinite hierarchy of conditions satisfied by any set of quantum correlations $\mathcal{Q}$ that, each of them, can be tested using semidefinite programming, where the set $\mathcal{Q}$ is fully characterized by it. 
First, consider the Gram matrix $G$ of the vectors
\begin{equation}
    \mathcal{T}=\{\ket{\psi}\}\cup \{A_a^x\ket{\psi}: a\in\mathcal{A},x\in \mathcal{X}\}\cup \{B_b^y\ket{\psi}: b\in\mathcal{B},y\in \mathcal{Y}\},
\end{equation}
then $G$ contains all the values appearing in $P$. The matrix $G$, which we naturally label by the set $\mathcal{T}$, and its entreis fulfill the following constraints: 
\begin{enumerate}
    \item $G$ is positive semidefinite,
    \begin{equation}
        G	\succeq 0.
    \end{equation}
    \item $\ket{\psi}$ is a normalized state, thus $\braket{\psi}{\psi}=1$, i.e.\ $ G_{0,0}=1$.
    \item $A_a^x$ and $B_b^y$ are projector operators, therefore,  $ \forall x \in \mathcal{X}$, $\forall y \in \mathcal{Y}$, $\forall a\neq a'\in \mathcal{A}$, $\forall b\neq b'\in \mathcal{B}$ and  $\forall T\in\mathcal{T}$:

     \begin{enumerate}
     \item  $\expectedbraket{A_a^xA_a^x}=\expectedbraket{A_a^{x}}$ and likewise for $B_b^y$.
    \item Because of completeness of measurements:
    \begin{equation}
    \begin{split}
        &\sum_{a\in\mathcal{A}}\expectedbraket{\psi|A_a^xT}=\expectedbraket{\psi| T}, \hspace{10mm} \sum_{a\in\mathcal{A}}\expectedbraket{TA_a^x|\psi}=\expectedbraket{ T |\psi},\\& 
        \sum_{b\in\mathcal{B}}\expectedbraket{\psi|B_b^yT}=\expectedbraket{\psi| T}, \hspace{10mm} \sum_{b\in\mathcal{B}}\expectedbraket{TB_b^y|\psi}=\expectedbraket{ T| \psi}.
    \end{split}
    \end{equation}
    \item Because they are orthogonal projections: $\expectedbraket{A_{a}^{x}A_{a'}^{x}}=0=\expectedbraket{B_{b}^{y}B_{b'}^{y}}$.
    \item Because they commute: $\expectedbraket{A_a^xB_b^y}=\expectedbraket{B_b^yA_a^x}$.
    \end{enumerate}
\end{enumerate}

Define
\begin{equation}
    \mathcal{Q}_1=\{P\textrm{ } \mid \textrm{ }\exists \textrm{ } G \textrm{ fulfilling 1.-3. and } P(a,b|x,y)=\expectedbraket{A_a^xB_b^y} \} \subset  L(\mathbb{R}^{\mathcal{X}}\otimes \mathbb{R}^{\mathcal{Y}},\mathbb{R}^{\mathcal{A}}\otimes \mathbb{R}^{\mathcal{B}}),
\end{equation}
By construction, $\mathcal{Q}_1\supseteq \mathcal{Q}$ and therefore,
\begin{equation}
\label{eq upperbound 1st level NPA}
    \omega^*(\mathcal{G})=\sup_{P\in\mathcal{Q}}\expectedbraket{K,P}\leq\sup_{P\in\mathcal{Q}_1}\expectedbraket{K,P}.
\end{equation}
Define the Hermitian operator $H$ as the matrix whose entries are $H(xa,yb)=H(yb,xa)=\frac{1}{2}q(x,y)f(a,b,x,y)$ for all $x\in\mathcal{X}$, $y\in\mathcal{Y}$, $a\in\mathcal{A}$, $b\in\mathcal{B}$ and all the other entries are $0$, so that $\expectedbraket{K,P}=\expectedbraket{H,G}$. The semidefinite program over all positive semidefinite matrices $G$ satisfying items $1$ to $3$ is the first level of the NPA hierarchy. The level $l$ of the hierarchy, $\mathcal{Q}_l$, see \cite{NPA2008} for a formal definition, is built considering the Gram matrix of the vectors of the Gram matrix of level $l-1$ and vectors corresponding to $l$ degree products of the projection operators. By construction, $\mathcal{Q}_l\supseteq \mathcal{Q}$. 
\begin{theorem}\cite{NPA2008} The NPA hierarchy converges to the set of quantum behaviors:
\begin{equation}
    \mathcal{Q}=\bigcap_{\ell\in \mathbb{N}}\mathcal{Q}_\ell.
\end{equation}
\end{theorem}

\section{Proof of Proposition \ref{prop ineq perr for QPV}}\label{Appendix proof ineq perr QPV}
Combining both expressions in \eqref{eq condition perr by def 1}, using the properties of the projectors and equation \eqref{eq restriction error epsilon}, we obtain the inequality
\begin{equation}
\label{eq inequality p_err}
    \expectedbraket{V_1^xA_0^xB_0^x}+\expectedbraket{V_0^xA_1^xB_1^x}\leq p_{err}(\expectedbraket{A_0^xB_0^x}+\expectedbraket{A_1^xB_1^x}).
\end{equation}
Because of \eqref{eq restriction error epsilon}, from \eqref{eq inequality p_err} we get
\begin{equation}
\label{eq inequality p_err 1}
    \expectedbraket{V_1^xA_0^x}+\expectedbraket{V_0^xA_1^x}\leq p_{err}(\expectedbraket{A_0^xB_0^x}+\expectedbraket{A_1^xB_1^x}),
\end{equation}
\begin{equation}
\label{eq inequality p_err 2}
    \expectedbraket{V_1^xB_0^x}+\expectedbraket{V_0^xB_1^x}\leq p_{err}(\expectedbraket{A_0^xB_0^x}+\expectedbraket{A_1^xB_1^x}),
\end{equation}
\begin{equation}
\label{eq inequality p_err 3}
    \expectedbraket{V_1^xA_0^x}+\expectedbraket{V_0^xB_1^x}\leq p_{err}(\expectedbraket{A_0^xB_0^x}+\expectedbraket{A_1^xB_1^x}), 
\end{equation}
\begin{equation}
\label{eq inequality p_err 4}
    \expectedbraket{V_1^xB_0^x}+\expectedbraket{V_0^xA_1^x}\leq p_{err}(\expectedbraket{A_0^xB_0^x}+\expectedbraket{A_1^xB_1^x}). 
\end{equation}

We will use it to find linear constraints on the entries of the gram matrix $G$ corresponding to $\expectedbraket{A_a^xB_b^{x'}}$. Consider
\begin{equation}
\label{equation split projectors}
\begin{split}
    2\expectedbraket{A_a^{x}B_{b}^{x'}}&=2\expectedbraket{\mathbb{I}\otimes A_a^{x}\otimes B_{b}^{x'}}= \expectedbraket{(V_a^{x}+V_{1-a}^{x})A_a^{x}B_{b}^{x'}}+
    \expectedbraket{(V_b^{x'}+V_{1-b}^{x'})A_a^{x}B_{b}^{x'}}\\&=
    \expectedbraket{(V_a^{x}+V_b^{x'})A_a^{x}B_{b}^{x'}}+\expectedbraket{V_{1-a}^{x}A_a^{x}B_{b}^{x'}}+\expectedbraket{V_{1-b}^{x'}A_a^{x}B_{b}^{x'}},
\end{split}
\end{equation}
then, summing over $a$ and $b$, we get
\begin{equation}
    \label{eq sum AaBb'}
    \begin{split}
    2\sum_{ab}\expectedbraket{A_a^{x}B_{b}^{x'}}&=\sum_{ab}\expectedbraket{((V_a^{x}+V_b^{x'}))A_a^{x}B_{b}^{x'}}+
    \expectedbraket{V_{1}^{x}A_0^{x}(B_{0}^{x'}+B_1^{x'})}+
    \expectedbraket{V_{0}^{x}A_1^{x}(B_{0}^{x'}+B_1^{x'})}+\\&
    \expectedbraket{V_{1}^{x'}(A_0^{x}+A_1^{x})B_{0}^{x'}}+
    \expectedbraket{V_{0}^{x'}(A_0^{x}+A_1^{x})B_{1}^{x'}}
   \\& \leq \sum_{ab}\expectedbraket{((V_a^{x}+V_b^{x'}))A_a^{x}B_{b}^{x'}}+
    \expectedbraket{V_{1}^{x}A_0^{x}}+
    \expectedbraket{V_{0}^{x}A_1^{x}}+
    \expectedbraket{V_{1}^{x'}B_{0}^{x'}}+
    \expectedbraket{V_{0}^{x'}B_{1}^{x'}},
    \end{split}
\end{equation}
where we used that $A_0^x+A_1^x\preceq \mathbb{I}$ and $B_0^{x'}+B_1^{x'}\preceq \mathbb{I}$. Then, using \eqref{eq inequality p_err 1} and \eqref{eq inequality p_err 2}, we recover \eqref{eq perr constraint sum ab}.

On the other hand, recall that 
\begin{equation}\label{eq projectors V attack as ketbras}    V_a^{x}=\ketbra{a_{x}}{a_{x}} \textrm{  and  } V_b^{x'}=\ketbra{b_{x'}}{b_{x'}},\end{equation}
and we can write 
\begin{equation}
\begin{split}
     &\ket{a_{x}}= \alpha_0^a\ket{0_{x'}}+\alpha_1^a\ket{1_{x'}}\\
     &\ket{b_{x'}}= \beta_0^b \ket{0_{x}}+\beta_1^b\ket{1_{x}},\\
\end{split}
\end{equation}
where $\alpha_i^a=\braket{i_{x'}}{a_{x}}$ and $\beta_j^b=\braket{j_{x}}{b_{x'}} $, where the dependence on $x$ and $x'$ is omitted for simplicity. We write the projectors \eqref{eq projectors V attack as ketbras} in the other basis in such a way that
\begin{equation}
\label{eq decompose projectors}
\begin{split}
    &V_a^{x}=\abs{\alpha_0^a}^2V_0^{x'}+\abs{\alpha_1^a}^2V_1^{x'}+\alpha_0^a\alpha_1^{a*}\ketbra{0_{x'}}{1_{x'}}
    +\alpha_0^{a*}\alpha_1^a\ketbra{1_{x'}}{0_{x'}},\\
    &V_b^{x'}=\abs{\beta_0^b}^2V_0^{x}+\abs{\beta_1^b}^2V_1^{x}+\beta_0^b\beta_1^{b*}\ketbra{0_{x}}{1_{x}}
    +\beta_0^{b*}\beta_1^{b}\ketbra{1_{x}}{0_{x}}.
\end{split}
\end{equation}
Plugging \eqref{eq decompose projectors} in \eqref{equation split projectors}, summing \eqref{equation split projectors} and summing over $a$ and $b$ in $\{0,1\}$,
\begin{equation}
    \begin{split}
        &4\sum_{a,b}\expectedbraket{A_a^{x}B_{b}^{x'}}=\\&
        \expectedbraket{((1+\abs{\beta_0^0 }^2)V_0^{x}+
        (1+\abs{\alpha_0^0 }^2)V_0^{x'}+\beta_0^0\beta_1^{0*}\ketbra{0_{x}}{1_{x}}+\beta_0^{0*}\beta_1^{0}\ketbra{1_{x}}{0_{x}}+\\&
        \alpha_0^0\alpha_1^{0*}\ketbra{0_{x'}}{1_{x'}}+\alpha_0^{0*}\alpha_1^{0}\ketbra{1_{x'}}{0_{x'}})A_0^xB_0^{x'}}+(2+\abs{\beta_1^0}^2)\expectedbraket{V_1^xA_0^xB_0^{x'}}+(2+\abs{\alpha_1^0}^2)\expectedbraket{V_1^{x'}A_0^xB_0^{x'}}+\\&
        \expectedbraket{((1+\abs{\beta_0^1 }^2)V_0^{x}+
        (1+\abs{\alpha_1^0 }^2)V_1^{x'}+\beta_0^1\beta_1^{1*}\ketbra{0_{x}}{1_{x}}+\beta_0^{1*}\beta_1^{1}\ketbra{1_{x}}{0_{x}}+\\&
        \alpha_0^0\alpha_1^{0*}\ketbra{0_{x'}}{1_{x'}}+\alpha_0^{0*}\alpha_1^{0}\ketbra{1_{x'}}{0_{x'}})A_0^xB_1^{x'}}+(2+\abs{\beta_1^1}^2)\expectedbraket{V_1^xA_0^xB_1^{x'}}+(2+\abs{\alpha_0^0}^2)\expectedbraket{V_0^{x'}A_0^xB_1^{x'}}+\\&
        \expectedbraket{((1+\abs{\beta_1^0 }^2)V_1^{x}+
        (1+\abs{\alpha_0^1 }^2)V_0^{x'}+\beta_0^0\beta_1^{0*}\ketbra{0_{x}}{1_{x}}+\beta_0^{0*}\beta_1^{0}\ketbra{1_{x}}{0_{x}}+\\&
        \alpha_0^1\alpha_1^{1*}\ketbra{0_{x'}}{1_{x'}}+\alpha_0^{1*}\alpha_1^{1}\ketbra{1_{x'}}{0_{x'}})A_1^xB_0^{x'}}+(2+\abs{\beta_0^0}^2)\expectedbraket{V_0^xA_1^xB_0^{x'}}+(2+\abs{\alpha_1^1}^2)\expectedbraket{V_1^{x'}A_1^xB_0^{x'}}+\\&
        \expectedbraket{((1+\abs{\beta_1^1 }^2)V_1^{x}+
        (1+\abs{\alpha_1^1 }^2)V_1^{x'}+\beta_0^1\beta_1^{1*}\ketbra{0_{x}}{1_{x}}+\beta_0^{1*}\beta_1^{1}\ketbra{1_{x}}{0_{x}}+\\&
        \alpha_0^1\alpha_1^{1*}\ketbra{0_{x'}}{1_{x'}}+\alpha_0^{1*}\alpha_1^{1}\ketbra{1_{x'}}{0_{x'}})A_1^xB_1^{x'}}+(2+\abs{\beta_0^1}^2)\expectedbraket{V_0^xA_1^xB_1^{x'}}+(2+\abs{\alpha_0^1}^2)\expectedbraket{V_0^{x'}A_1^xB_1^{x'}}\\
    \end{split}
\end{equation}

Using that $A_0^x+A_1^x\preceq \mathbb{I}$ and $B_0^{x'}+B_1^{x'}\preceq \mathbb{I}$ and \eqref{eq inequality p_err 1} and \eqref{eq inequality p_err 2}, as in derivation of \eqref{eq perr constraint sum ab},
and bounding the terms that do not correspond to $\expectedbraket{V_{1-a}A_aB_b}$ or $\expectedbraket{V_{1-b}'A_aB_b}$ by the operator norm, we obtain \eqref{eq constraint perr 2 simplified}.

\section{Proof of Proposition \ref{prop ineq perr for QKD }}\label{Appendix proof ineq perr QKD}
Combining both expressions in \eqref{eq condition perr QKD by def 1}, using the properties of the projectors and equation \eqref{eq Eve equal meas Alice constriant}, we obtain the inequality, for $e,b\in\{0,1\}$
\begin{equation}
\label{eq inequality p_err QKD}
    \expectedbraket{V_0^xE_0^xB_1^x}+\expectedbraket{V_1^xE_1^xB_0^x}\leq p_{err}\sum_{e,b}\expectedbraket{E_e^xB_b^x}.
\end{equation}
Because of \eqref{eq Eve equal meas Alice constriant}, from \eqref{eq inequality p_err QKD} we get
\begin{equation}
\label{eq inequality p_err 1 QKD}
    \expectedbraket{V_0^xB_1^x}+\expectedbraket{V_1^xB_0^x}\leq p_{err}\sum_{e,b}\expectedbraket{E_e^xB_b^x}.
\end{equation}
In an analogy with \eqref{equation split projectors}, consider

\begin{equation}
\label{equation split projectors QKD}
    2\expectedbraket{E_e^{x}B_{b}^{x'}}=
    \expectedbraket{(V_e^{x}+V_b^{x'})E_e^{x}B_{b}^{x'}}+\expectedbraket{V_{1-e}^{x}E_e^{x}B_{b}^{x'}}+\expectedbraket{V_{1-b}^{x'}E_e^{x}B_{b}^{x'}},
\end{equation}
because of \eqref{eq Eve equal meas Alice constriant}, $\expectedbraket{V_{1-e}^{x}E_e^{x}B_{b}^{x'}}=0$, summing \eqref{equation split projectors QKD} over $e,b\in\{0,1\}$, using $A_0^x+A_1^x\preceq \mathbb{I}$ applying \eqref{eq inequality p_err 1 QKD}, we recover \eqref{eq perr constraint sum ab QKD}. 

\end{appendices}

\end{document}